\documentclass[11pt, oneside]{article}  

\usepackage[]{geometry} 
\usepackage{graphicx}
\usepackage{amssymb}
\usepackage{amsmath} 
\usepackage{amsthm} 
\usepackage{subcaption}
\usepackage{algorithmic}
\usepackage{algorithm}
\usepackage{float}
\usepackage{mathtools}
\usepackage{color}
\usepackage{url}
\usepackage{comment}
\usepackage{authblk}
\usepackage{float}

 \newtheorem{proposition}{ Proposition}[section]
 
 \newtheorem{remark}{ Remark}

\title{PCA for Implied Volatility Surfaces}

\author{M. Avellaneda\thanks{Courant Institute of Mathematical Sciences, New York University, 251 Mercer Street,
New York, NY 10012-1185, USA (avellaneda@courant.nyu.edu)}, B. Healy\thanks{Department of Financial Computing and Analytics, University College London, Gower Street, London WC1E 6BT, UK  (brian@decisionsci.net)}, A. Papanicolaou\thanks{Department of Finance and Risk Engineering, NYU Tandon School of Engineering,
6 MetroTech Center, Brooklyn, NY 11201-3840, USA (ap1345@nyu.edu)}~ and G. Papanicolaou\thanks{Department of Mathematics, Stanford University, Stanford, CA 94305 (papanicolaou@stanford.edu)}}

\date{\today}

\begin{document}
\maketitle

\begin{abstract}
  Principal component analysis (PCA) is a useful tool when trying to construct factor models from historical asset returns. For the implied volatilities of U.S. equities there is a PCA-based model with a principal eigenportfolio whose return time series lies close to that of an overarching market factor. The authors show that this market factor is the index resulting from the daily compounding of a weighted average of implied-volatility returns, with weights based on the options' open interest (OI) and Vega. The authors also analyze the singular vectors derived from the tensor structure of the implied volatilities of S\&P500 constituents, and find evidence indicating that some type of OI and Vega-weighted index should be one of at least two significant factors in this market.
\end{abstract}

\newpage

\section{Introduction} 
\label{sec:intro}

We show through principal component analysis (PCA) that a relatively small number of factors can account for most of the variation in the collective movements of the implied volatilities derived from U.S. equity options. In fact, a matrix formed with normalized implied volatility returns over time has positive covariance with the first principal eigenvector, which is closely related to an open interest and vega (OI-Vega)-weighted basket of implied volatilities. This draws parallels with the first principal component obtained from the covariance of the matrix of normalized equity returns and its proximity to the capitalization-weighted portfolio (see \cite{avellaneda2010,boyle2014}), since OI is a measure of market size for options just as capitalization is a measure of market size for equities, and both implicitly carry liquidity information. OI is the number of open contracts for a given option (name, strike, maturity) at a given time. Our findings highlight and give detailed insight into how PCA can be used to extract information from the covariance structure for a large dataset of implied volatilities, and how new and improved implied volatility factors can be constructed using OI and Vega -- both of which will be useful for portfolio and risk managers who have a need for better statistical prediction models to improve their estimated risk metrics (such as VaR and expected shortfall). To date, the preeminent volatility index is VIX, which is constructed from index options. The VIX has proven reliable but on occasion has shown susceptibility to outlier prices and manipulative trading (see \cite{griffin2017}). The OI-Vega based factors constructed in this paper are more robust as they are based on hundreds of implied volatilities and place more emphasis on those contracts having most trading interest.

The study in this paper builds on the work in \cite{avellanedaDobi2014}, where they consider a large dataset of implied volatility surfaces for a few thousand U.S. equities, and use PCA to find the smallest number of factors needed to explain the collective movements of these volatilities. They also construct principal eigenportfolios and examine the qualitative structure of each from the $1^{st}$ through $4^{th}$ eigenvectors. We draw from the same data source as \cite{avellanedaDobi2014}, namely, the implied volatility surface (IVS) data available from OptionMetrics through Wharton Research Data Services (WRDS). We hypothesize that the normalized covariance matrix of market-wide implied volatilities has a \textit{low-rank plus random} structure (known as the ``spike model"), and similar to \cite{avellanedaDobi2014}, we find that removal of the low-rank components leaves a residual whose squared singular values are close in distribution to a Marchenko-Pastur law. Using Random Matrix Theory (RMT), the presence of principal factors should make it possible to reject a model of purely random noise. RMT was used in \cite{avellanedaDobi2014}, with the spectra limiting Marchenko-Pastur distribution providing the basis for establishing cutoffs for the identification of the non-random structure. The analyses in this paper consistently (for multiple years) show there to be at least two outliers in the singular value distribution, indicating that at least two factors are driving the time series of IVS returns. 

As already noted, we focus on the $1^{st}$ principal eigenportfolio and its closeness to various OI and Vega-weighted portfolios. An eigenportfolio is a vector of portfolio weights that is derived from an eigenvector of the returns' covariance matrix. In equities, it is well known that the eigenportfolio constructed from the $1^{st}$ principal eigenvector has explanatory power for the cross section of U.S. equity returns (\cite{avellaneda2010,boyle2014}), and that it tracks closely with a dominant factor such as a market portfolio. Perhaps the most significant finding in this paper is a sizable body of evidence indicating that option OI and OI-Vega are key elements for construction of factors for explaining the collective cross-sectional changes of implied volatility surfaces. Specifically, we show that various factors constructed from OI and OI-Vega-weightings of implied volatility returns have significant explanatory power for interpreting the principal eigenportolio's returns. This finding can be considered as the \textit{implied volatility analogue} to the equity market's $1^{st}$ principal factor, namely the capitalization-weighted returns portfolio. As already noted, it appears that OI plays a similar role for implied volatility to that played by capitalization for equities. However, such a comparison is very informal as the CAPM and its related economic theory bind equities and capitalization closely together, whereas the implied volatility results shown in this paper are, at present, statistical findings. 

The time series of implied volatility surfaces can be put into vector form, but its natural representation is a 4-dimensional tensor, with the 4 dimensions being time, name, option maturity and option delta (normalized strike). Our study of the $1^{st}$ principal eigenportfolio can be extended to this tensor setting, which provides an example of how factor construction can, in fact, be improved by considering the natural representation offered by the tensor structure. The maturity and strike dimensions lend themselves to individual factors, for which we can construct individualized OI-weighted factors for each option maturity or for each maturity-delta pair. Individualized factors allow for a more nuanced weighing of changes in implied volatilities, and this leads to improved explanatory power for the covariance structure's, suitably defined, principal eigenportfolio.

\subsection{Review of Literature}

The work of \cite{avellanedaDobi2014} and \cite{dobi2014modeling} provides a cross-sectional classification of U.S. equity options based on implied volatility data for the period from August 2004 to August 2013, jointly with equity returns. The spectrum of the joint equity-IVS is used, in particular the leading eigenvalues, to classify options into those carrying mostly systemic risk and into those carrying mostly idiosyncratic risk. Then employing methods from principal component analysis and results from random matrix theory, the significant eigenvalues are identified, and it is shown that approximately nine principal components suffice to reproduce the implied volatility surfaces of all equities studied, with even fewer risk factors for so-called systemic names, such as SPY, QQQ and AAPL. An explicit model is introduced, which can be used to track the dynamics of the implied volatility surface, yet is compact and computationally tractable. 

Focusing on the implied volatility surface for a single asset, \cite{cont2002dynamics} examine time series of option prices for options on the S\&P500 and FTSE100 indices. They show how the implied volatility surface can be deformed and represented as a randomly fluctuating surface driven by a small number of orthogonal random factors and find a simple factor model compatible with the empirical observations. Also of interest are methods for pricing baskets of many assets using option implied volatilities (see \cite{avellanedaBaskets2002}).

Prior to work on implied volatilities there was a bounty of research on equities. The original work on portfolio composition dates back to \cite{markowitz1952portfolio}, and subsequent work on the CAPM model which focused on the expected return of an asset relative to the risk-free instrument and derived a relationship between this and the excess return of a market or benchmark portfolio. Work in subsequent years suggested that factors other than the market excess return were being priced. In particular,  \cite{roll1980empirical} used data for individual equities during the 1962–1972 period and found that at least three are priced in the generating process of returns, and \cite{fama1992cross} found the most significant factors to be market excess return, company size and the ratio of the book value to the market value of the firm. An early use of PCA and random matrix theory for analyzing equity returns is \cite{plerou2002random}.

The importance of eigenportfolios is highlighted in \cite{boyle2014}, where there is examination of conditions under which frontier portfolios have positive weights on all assets. This is of interest since the market portfolio given by CAPM is mean variance efficient and has positive weights on all assets. Prior to this work is \cite{avellaneda2010}, which studies statistical arbitrage strategies in U.S. equities with trading signals generated using PCA. Modeling the residuals of stock returns as a mean-reverting process, \cite{avellaneda2010} develop contrarian trading signals and then back-test these over the broad universe of U.S. equities. The fact that these PCA-based strategies have an average annual Sharpe ratio that is statistically and economically significant provides empirical support for the PCA approach.

\subsection{Structure and Results of the Paper}

This paper has three main sections after this introduction. The first section addresses the estimation of the low-rank principal component structure from the standardized returns of options' implied volatility. The main result of this section is the introduction of an effective dimension approach for assessing the randomness of residuals, especially when there is only randomness in time since the vectorized IVS data produces residuals that do retain some of the structure of the data.  The second section explores the role of the $1^{st}$ principal component in constructing an eigenportfolio, with option OI and Vega as weights, as the primary factor in evaluating collective movements of implied volatility surfaces. The final section makes use of the data's natural tensor structure for construction of improved principal eigenportfolios. The main contribution of this paper is the presented evidence demonstrating the importance of OI when measuring changes in implied volatilities. Performing PCA to determine the number of relevant factors is a fairly standard procedure once we've standardized the data, but construction and analysis of eigenportfolios requires a deeper understanding of the data, including OI. The tensor analysis does provide more depth of understanding, as it shows us that construction of factors individualized to sub-categories of options (e.g., separate OI-based factors for each of the options' maturities) leads to a clear improvement in the eigenportfolio's ability to account for implied volatilities' movements.

\section{Matrix of Implied Volatility Returns}
\label{sec:PCA}

Let $t$ be an index denoting calendar days. Let $i$ be an index denoting an individual option contract, and denote the implied volatility for this particular option contract as $\hat\sigma_i(t)$. We define at time $t$ the vector of daily returns on the $i^{th}$ contract's implied volatility as

\begin{align}
\label{eq:returns}
    r_i(t) = \frac{d\hat\sigma_{i}(t)}{\hat\sigma_{i}(t)} \qquad \text{for} \quad 1\leq i\leq N \quad \text{and} \quad 1\leq t\leq T  \ ,  
\end{align}
where $d\hat\sigma_{i}(t) = \hat\sigma_{i}(t+dt)-\hat\sigma_{i}(t)$ with $dt=1/252=$ 1 day. \textbf{We assume throughout that the number of contracts far exceeds the number of days, $N\gg T$,} which means that the covariance/correlation matrix has several eigenvalues equal to zero. We standardize these returns and then place them into a matrix $R\in\mathbb R^{N\times T}$, given by

\begin{align}
    \label{eq:standardizedReturns}
    R &= \left[\frac{r_{i}(t) -\bar r_i}{h_i}\right]_{1\leq i\leq N,1\leq t\leq T}
\end{align}
where $\bar r_i = \frac{1}{T}\sum_tr_{i}(t)$ and $h_i = \sqrt{\frac{1}{T-1}\sum_t(r_{i}(t)-\bar r_i)^2}$. 
Our hypothesis is that $R$ can be decomposed into a low-rank factor matrix $F$ and a random matrix $X$
\begin{align}
    \label{eq:lowRankMatrixModel}
    R = F + X
\end{align}
which can be tested using the ``spike-model" approach  (see \cite{benaych2011eigenvalues}). The low-rank matrix $F$ can be further decomposed into orthogonal components
\begin{align}
    \label{eq:factorDecomposition}
    F = \sum_{i=1}^d f_i\theta_i^*    
\end{align}
where each vector $f_i$ is a principal characteristic of $R$ with orthogonality between $f_i$ and $f_j$ $\forall \ i\neq j$, each $\theta_i$ is its loading, and where $*$ denotes matrix/vector adjoint. In particular, if $\|f_i\|=1$ then $\|\theta_i\|^2$ is an eigenvalue of $FF^*$.

There are two obvious issues to address: the value of $d$ in \eqref{eq:factorDecomposition} and the vectors $(f_i)_{i=1,\dots,d}$. In the RMT literature it is equally as important to determine the $\theta_i$'s, as the criticality of a given $\theta_i$ will determine whether or not component $f_i$ is distinguishable from $R$'s bulk eigenvectors. In applications to financial data however, the top components are typically much greater than the critical threshold and attempts to include the middle ranks of $F$ will usually lead to overfitting.

A brief description of the IVS data is in the appendix.
For our analysis, we extracted 56 implied volatility values for each of the (roughly) 500 S\&P500 constituents, for each of the approximately 252 business days in each year of data, 2012-2017. Specifically, we use options with time to maturity: 30 days, 60 days, 91 days, 122 days, 152 days, 182 days, 273 days \& 365 days which have delta (normalized strike; see Appendix \ref{sec:data}): -20, -30, -40, 50, 40, 30, 20. Therefore, for each maturity we use three out of the money put options (-20, -30, -40), one at-the-money option (50) and three out-of-the-money call options (40, 30, 20). Out-of-the-money options are used as these are more widely traded and hence are more liquid and have more reliable prices. Thus $N=500\times 8\times 7=28,000$ and
$T=252$ (or $250$ or $251$ depending on when holidays fall in the year) if we use a one-year estimation window. \\

\subsection{Singular Values of Non-Principal Structure}

A very basic estimator of the covariance matrix is $\widehat \rho = RR^*/T$. Much of the literature has addressed methods for improvements of this estimator, including shrinkage of the eigenvalues in \cite{ledoit2004well} and asymptotic behavior of eigenvectors in \cite{ledoit2011eigenvectors} and \cite{ledoit2012nonlinear}. Perhaps the most applicable reference for what we're seeking in this paper is \cite{benaych2011eigenvalues}, which contains a result stating that if the magnitude of the $\|\theta_i\|$'s are greater than some critical threshold with respect to the variance of $X$'s entries in \eqref{eq:lowRankMatrixModel}, then the principal component vectors of $\widehat\rho$ will be inside a cone centered around the principal eigenvectors of $FF^*/T$. For financial data, the $1^{st}$ principal component often accounts for as much as 50\% of the total variance and thus has an eigenvalue that is well over the threshold, but higher-order factors may be closer to the critical level. Detection of factors whose eigenvalue(s) are near the critical threshold is interesting but not the main focus in this paper. Instead, we focus on finding an estimate of the minimum number of factors needed to have a statistical non-rejection of the estimated low-rank model.

In practice, the estimator $\widehat\rho$ is not calculated because usually $N\gg T$ making it more efficient to compute the Singular Value Decomposition (SVD). The SVD represents $R$ as
\begin{align}
    \label{eq:SVD}
    R &= U SV^*\ ,
\end{align}
where $U=[U_1,U_2,\dots,U_T]$ is an $N\times T$ matrix with orthonormal columns, $S$ is a $T\times T$ diagonal matrix with entries $S_{11}\geq S_{22}\geq \dots\geq S_{TT}\geq 0$, and $V=[V_1,V_2,\dots V_T]$ is a $T\times T$ matrix with orthonormal columns. The non-zero eigenvalues of the correlation structure are the $S_{ii}^2$ values, and if $R$ were completely random (i.e., if $F=0$ in \eqref{eq:lowRankMatrixModel}), then the histogram of these values would be close to a Marchenko-Pastur (MP) density when $N$ and $T$ are large. However, from Figure \ref{fig:singularValues}, it is clear that at least two values separated visibly from the bulk of $(S_{ii}^2/N)_{1\leq i\leq T}$. Hence, the rank of $F$ is at least two, which means at least two principal components need to be removed from the data for the remainder to be considered ``noise".
\begin{figure}[t]
    \centering
    \begin{subfigure}[b]{0.3\textwidth}
        \includegraphics[width=\textwidth]{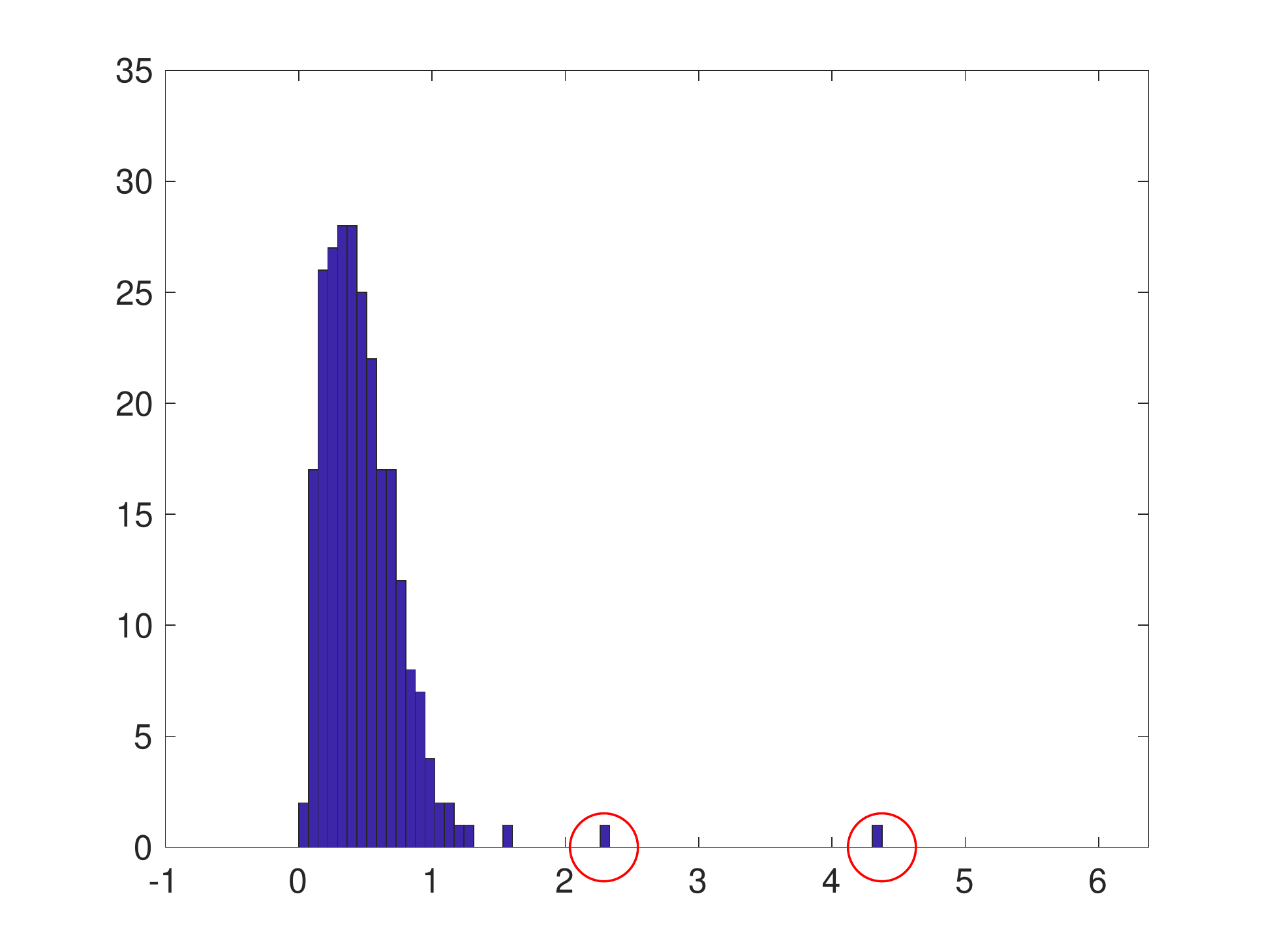}
        \caption{2012}
        \label{fig:2012SVflat}
    \end{subfigure}
\begin{subfigure}[b]{0.3\textwidth}
        \includegraphics[width=\textwidth]{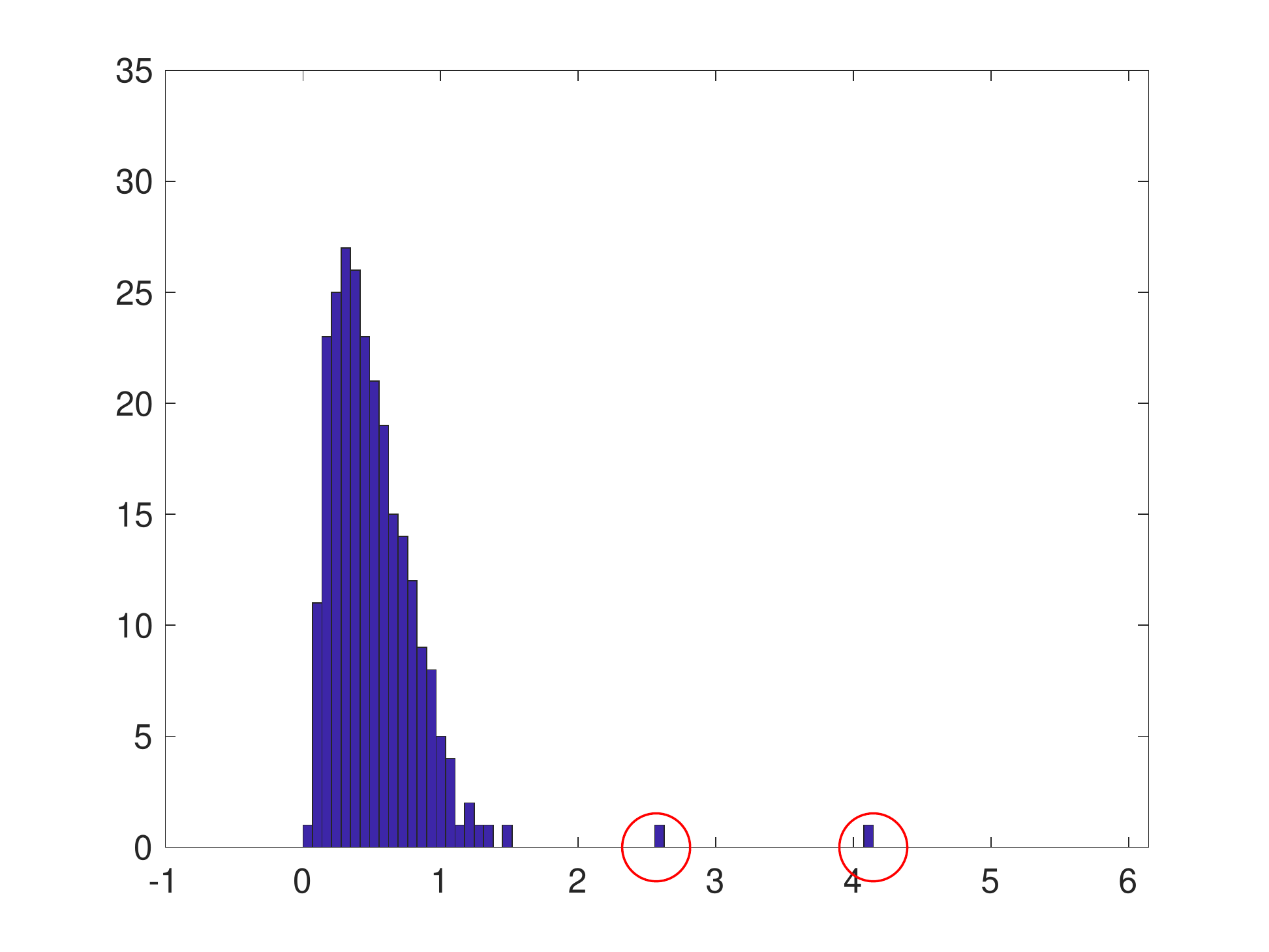}
        \caption{2013}
        \label{fig:2013SVflat}
    \end{subfigure}      
\begin{subfigure}[b]{0.3\textwidth}
        \includegraphics[width=\textwidth]{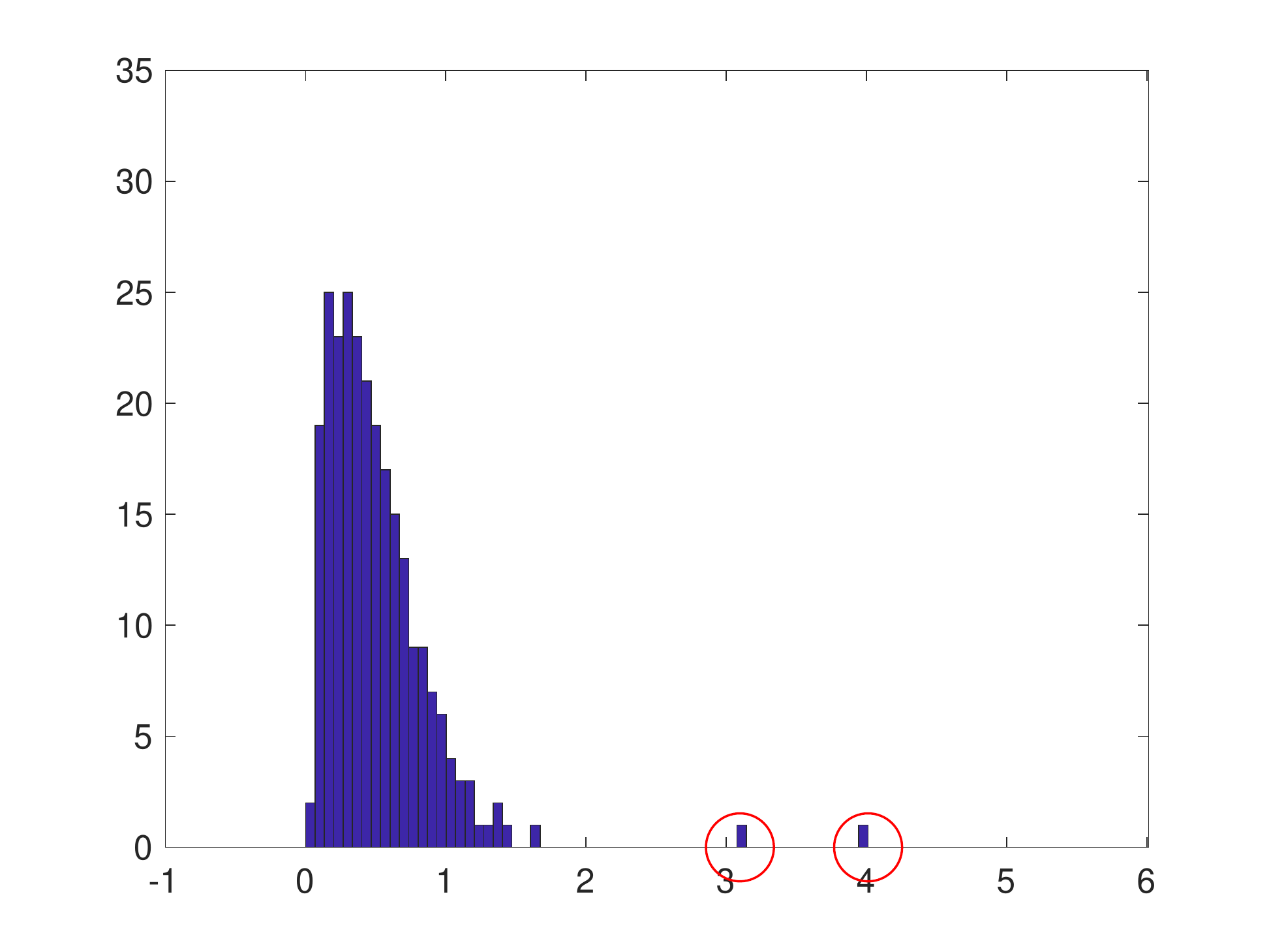}
        \caption{2014}
        \label{fig:2014SVflat}
    \end{subfigure}\\
\begin{subfigure}[b]{0.3\textwidth}
        \includegraphics[width=\textwidth]{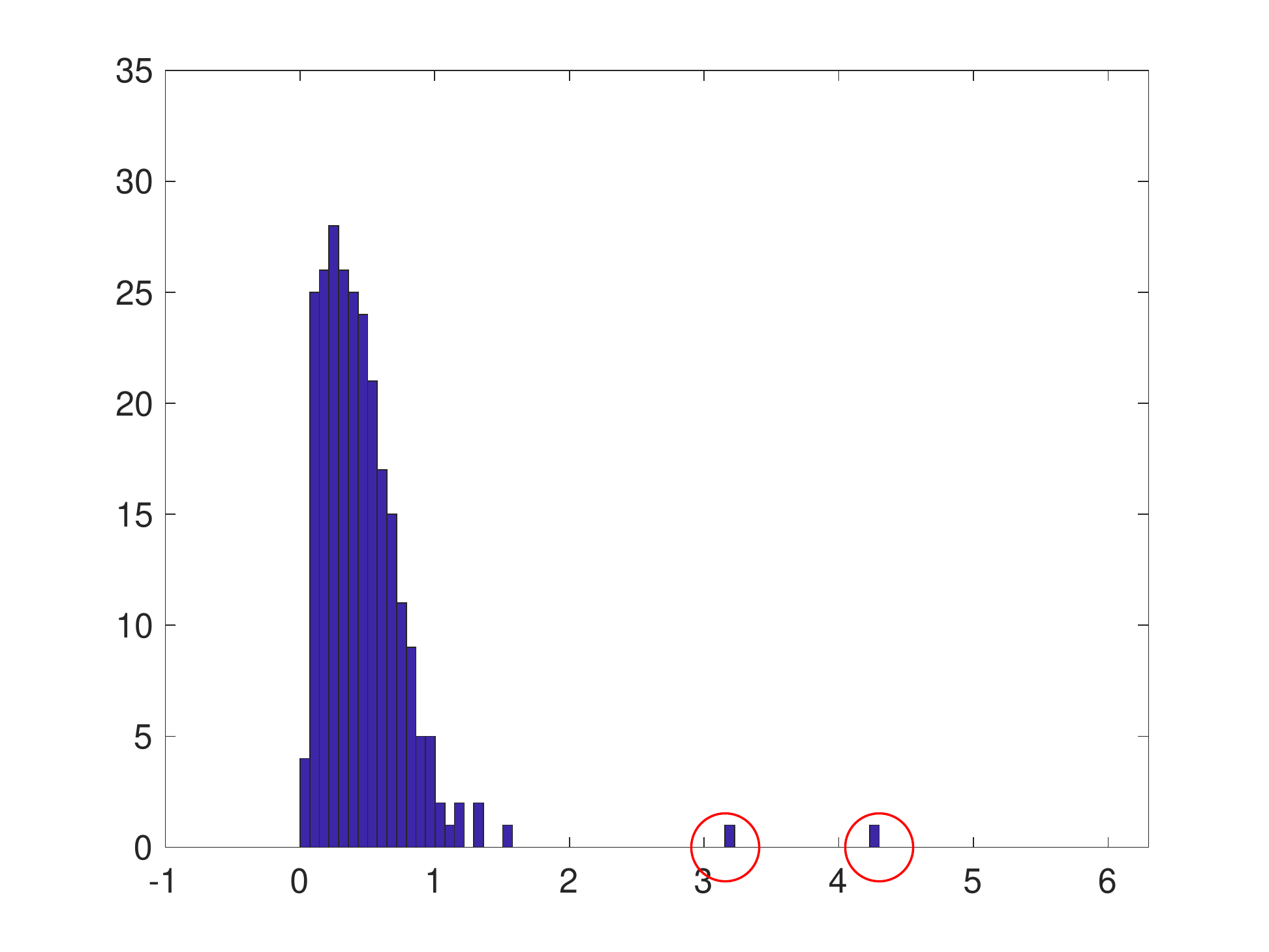}
        \caption{2015}
        \label{fig:2015SVflat}
    \end{subfigure}
\begin{subfigure}[b]{0.3\textwidth}
        \includegraphics[width=\textwidth]{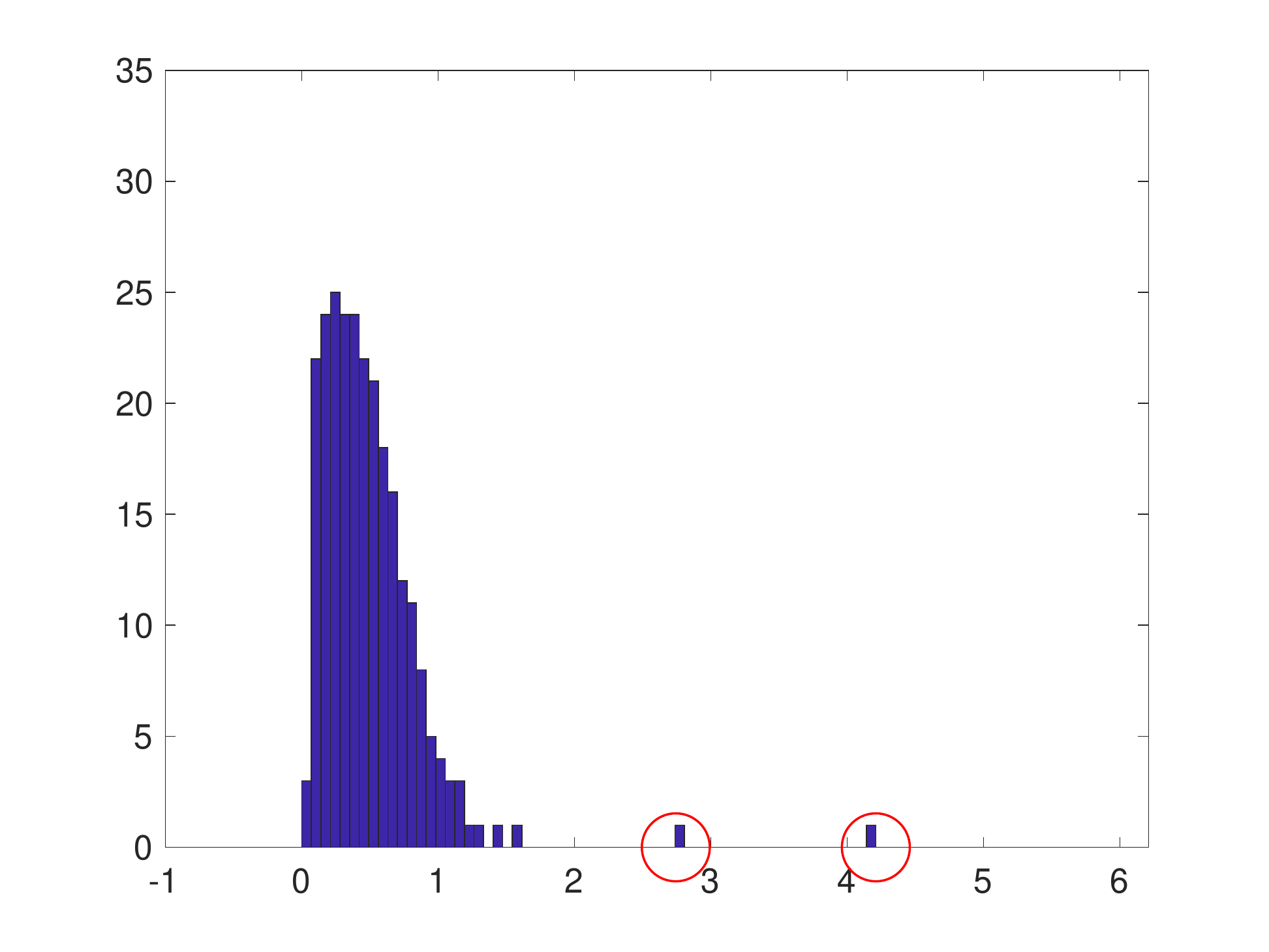}
        \caption{2016}
        \label{fig:2016SVflat}
    \end{subfigure}
\begin{subfigure}[b]{0.3\textwidth}
        \includegraphics[width=\textwidth]{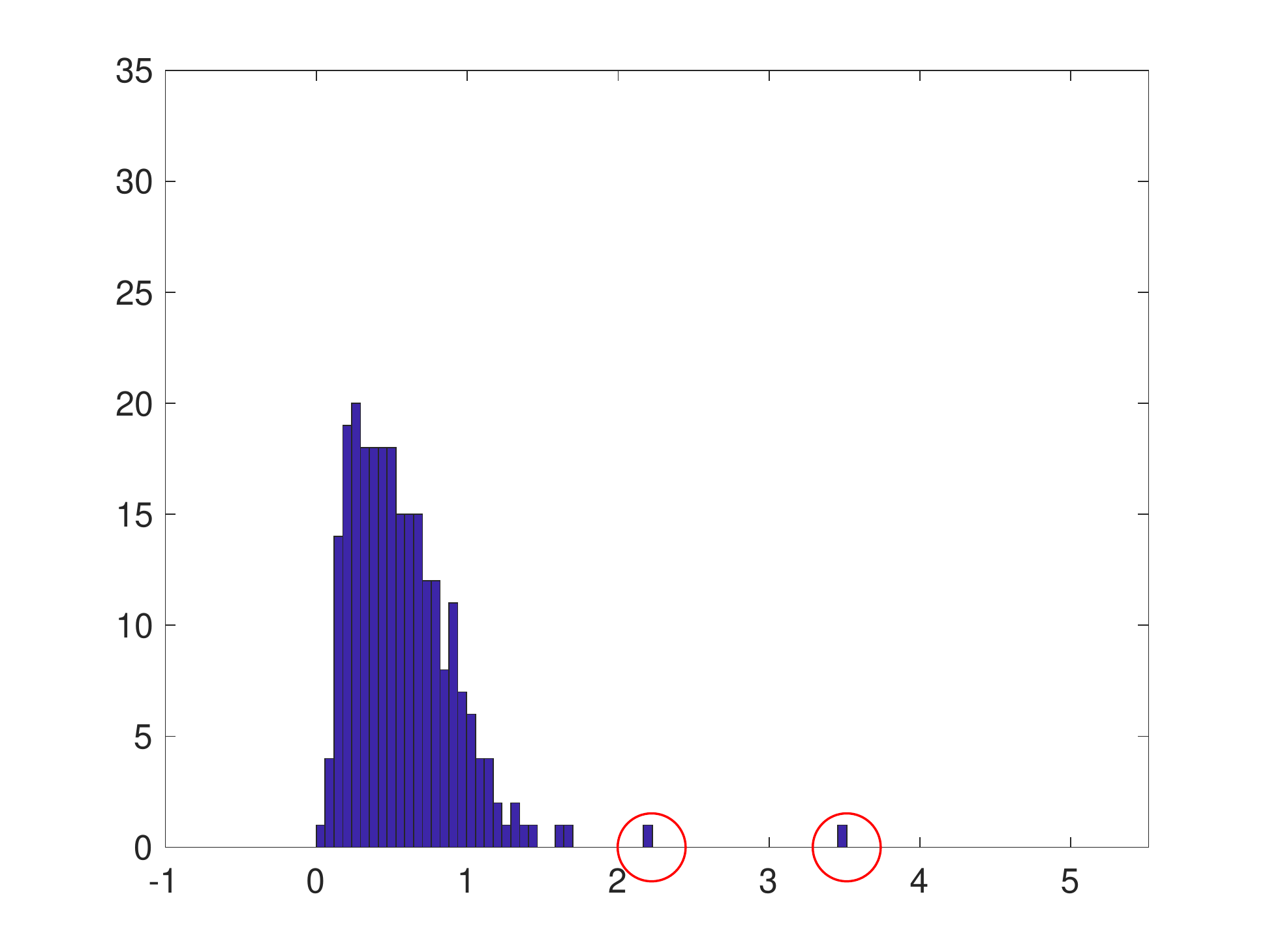}
        \caption{2017}
        \label{fig:2017SVflat}
    \end{subfigure}
     \caption{Histograms of $\log( 1 + S_{ii}^2/N)$ where $S_{ii}$ are the non-zero singular values of $R$, estimated with one-year windows from 2012 to 2017. If $F=0$, there will be low probability of outliers from the bulk, and the histogram can be fitted with a Marchenko-Pastur density. However, in each of these histograms we see at least two outliers (circled marks), which indicates the rank of $F$ is at least two. Hence, principal components need to be removed.
     \label{fig:singularValues}}
\end{figure}

For some $d\leq T$, the best rank-$d$ estimator of $F$ that minimizes the Frobenius norm of error, is 
\[\widehat F = \sum_{i\leq d}S_{ii}U_iV_i^*\ ,\] 
with residual $\widetilde R = R-\widehat F$. The true rank of $F$ is greater than $d$ if a statistical test of the residual rejects the hypothesis of purely random entries in $\widetilde R$. A Kolmogorov-Smirnoff (KS) test is likely to reject this hypothesis if $d$ is too small. We devise a simple KS test using the data's own estimated MP distribution, that is, we use the estimated asymptotic distribution parameters obtained from $(S_{ii}^2/N)_{d< i\leq T}$ and then check for significance of the associated KS statistic as we now describe.

The MP density is
\begin{align}
\label{eq:MPdensity}
\nu(x) = \frac{1}{2\pi \gamma^2} \frac{\sqrt{(\lambda_+-x)(x-\lambda_-)}}{\lambda x}\qquad\hbox{for }\lambda_-\leq x\leq \lambda_+
\end{align}
where 
\begin{align}
\label{eq:MPlmbda}
\lambda_\pm = \gamma^2(1\pm\sqrt\lambda)^2 ~\hbox{and}  ~\lambda>0 \ .
\end{align}
If $R$ were a random $N\times T$ matrix with independent identically distributed entries of mean zero and variance $\gamma^2$, then for $N>T$ the RMT tells us that the empirical spectral distribution of the covariance $\frac{1}{N} R^* R$ 
\begin{align*}
    \frac{1}{T}\sum_{i=1}^T {\bf{1}}_{\{S_{ii}^2/N \leq x\}},
\end{align*}
converges in probability (pointwise in $x$) to the distribution of the MP density in \eqref{eq:MPdensity} as $N$ and $T$ tend to infinity with $\lim\frac{T}{N}=\lambda\in(0,1)$ fixed. Conversely, the spectral distribution of $\frac{1}{T}RR^*$ is the case of $\lambda > 1$, wherein the limit law has an additional discrete mass at zero with weight $1-1/\lambda$, which appears because the $N\times N$ covariance has rank at most $T < N$, and so there are $N-T$ zero eigenvalues. Since we are interested only in eigenvalues through their empirical spectral density, we can consider the $T\times T$ covariance $\frac{1}{N} R^*R$ for which the dimension ratio $\lambda = T/N$ is less than one and there is no mass at zero in the asymptotic MP law.

The issue with the IVS data matrix $R$ and its residual $\widetilde R = R-\widehat F$ is that we are not dealing with matrices with independent identically distributed entries. We are in fact very far from it, and so it is not at all clear that the empirical spectral density of the residual matrices will be close to the MP law. There are significant correlations among the entries of the residual matrices, even without addressing the normalization issue. As already noted, there is considerable theory on separating the bulk spectrum from the spike eigenvalues for idealized random matrix spike models, and we may also cite the survey, \cite{johnstone2018pca}, and in dealing with normalization issues, \cite{ElKaroui2008}. The theoretical criteria provided in the literature do not work with the IVS data, as expected. Writing a data matrix as a factor matrix plus a residual so that the residual is ``noise", or has no useful information, is a problem that arises often and in many different disciplines, not only with financial data, but also for example in imaging in materials science, \cite{berman2019improved}. With real data, this is almost always treated with a variety of empirical estimation methods whose validity is assessed on the basis of the results produced in specific applications.

For the IVS data, we will fit the empirical spectrum of the residuals to the MP law by matching supports as we now describe. The quality of the fit is quantified by a KS test. The main result of this empirical fit, which works well for the IVS data, is to extract an estimated dimension ratio $\hat\lambda$ and standard deviation $\hat\gamma$ (see Table \ref{tab:eff_dim}).

Let $X_i = S_{ii}^2/N$ be the eigenvalues (normalized) of the IVS data matrix $R$ (normalized). Given the number of factors $d$ that we want to retain, the estimators for for the support $\lambda_\pm$ of the empirical spectral density, or histogram, of the residual that we use are
\begin{align}
\label{eq:lambdaPM_Est}
\hat\lambda_+ = X_{d+1}\qquad \hbox{and} \qquad \hat\lambda_- = X_{T}\ ,
\end{align}
which then using \eqref{eq:MPlmbda} with $\lambda\in(0,1)$ gives
\begin{align}
\label{eq:lambdaGamma_Est}
\hat\gamma = \frac{\sqrt{\hat\lambda_+ }+\sqrt{\hat\lambda_-}}{2} \qquad \hbox{and} \qquad \hat\lambda = \left( \frac{\sqrt{\hat\lambda_+ }-\sqrt{\hat\lambda_-}}{2\hat\gamma}\right)^2\ .
\end{align}

The fitted MP densities to the empirical spectral densities of the residuals with 9 components removed, are shown in Figure \ref{fig:fittedRank_d_Rehiduals} for each of the years from 2012 to 2017. For each of these years a 2-sample KS test does not reject, and increasing $d$ to $10$, $20$, $30$ and $50$ continues to result in a visibly good fit and non-rejection by the KS test. Hence, we conclude that $9$ factors is typically enough to describe the daily systematic movements among all implied volatility surfaces over a single year of daily data. In contrast, for equity returns for the S\&P500 constituents, it typically requires (roughly) 20 factors to account for the majority of daily movement, with the number dropping below 10 during the 2008 financial crisis (see \cite{avellaneda2010}). 

A direct comparison  with \cite{avellanedaDobi2014}, wherein a cutoff for purely random entries was determined to be around 108 factors (out of $\sim$3,000 names), 
is not possible because the dataset in \cite{avellanedaDobi2014} includes equity returns normalised using the strike of the ATM option. This is a more complex data set because the vectorized time series of IVS plus equities (of size $28,000+500$ in the context of this paper) is quite heterogeneous and therefore the ``low-rank plus random" decomposition needs additional attention.

In this paper, the 9 factors (out of $\sim$500 names for IVS) 
describe the systematic movements. Compared to equity returns by themselves, the number of factors are of comparable size for the different years 2012-2017 and $d=20$ factors are needed for the S\&P500 constituents to produce a ``random" residual.

\begin{figure}[t]
    \centering
    \begin{subfigure}[b]{0.3\textwidth}
        \includegraphics[width=\textwidth]{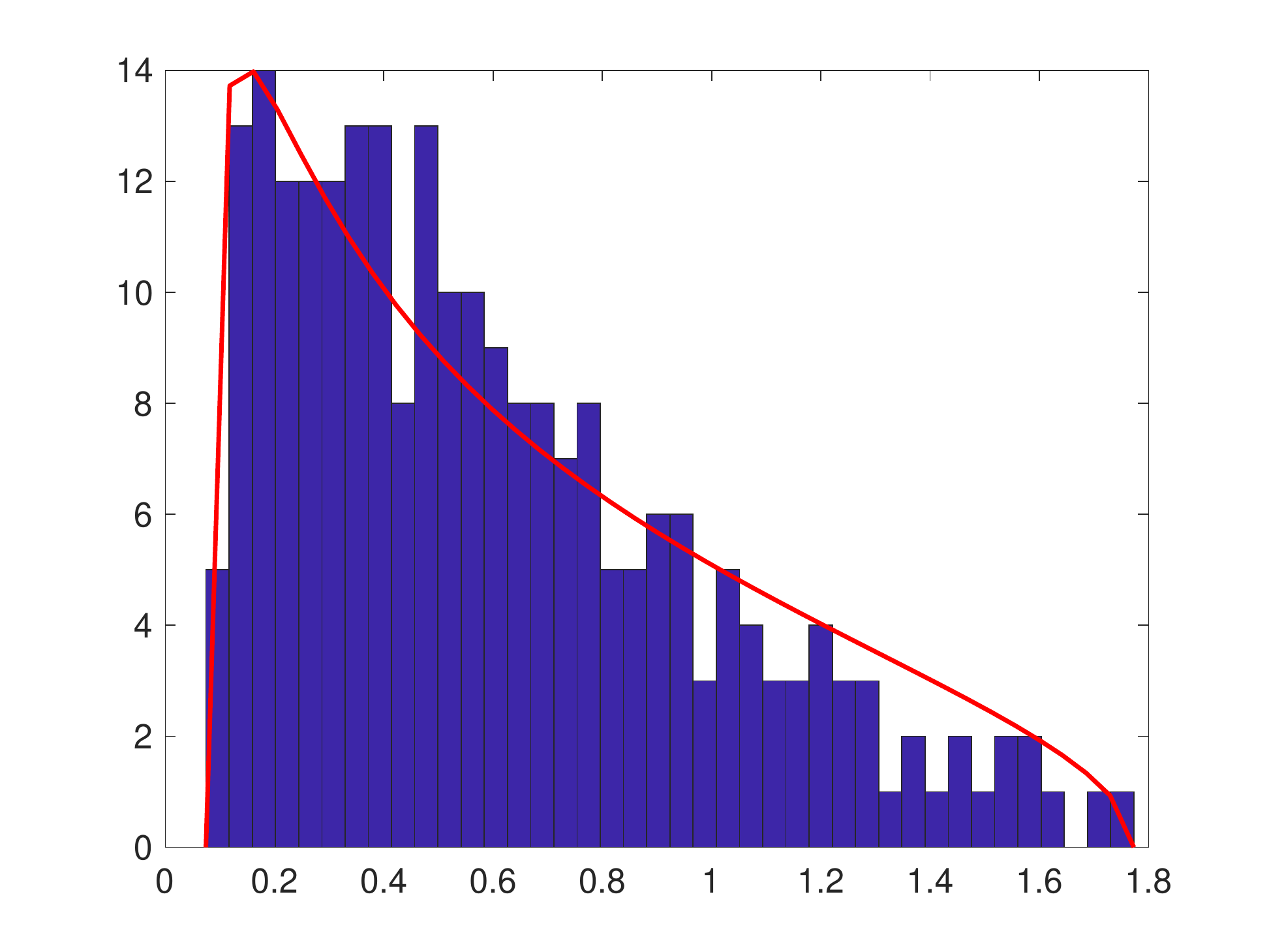}
        \caption{2012}
        \label{fig:2012MPfit}
    \end{subfigure}
\begin{subfigure}[b]{0.3\textwidth}
        \includegraphics[width=\textwidth]{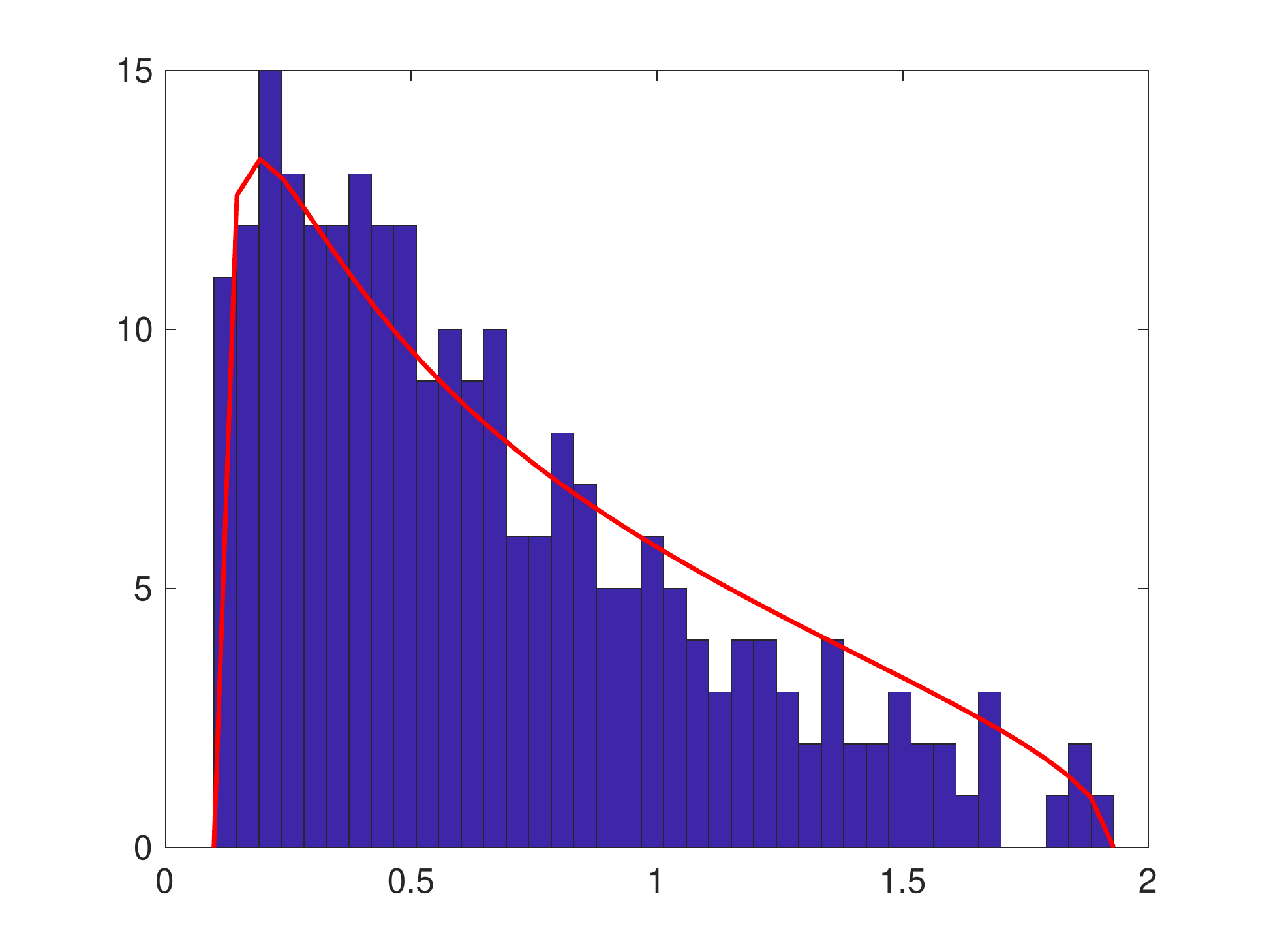}
        \caption{2013}
        \label{fig:2013MPfit}
    \end{subfigure}      
\begin{subfigure}[b]{0.3\textwidth}
        \includegraphics[width=\textwidth]{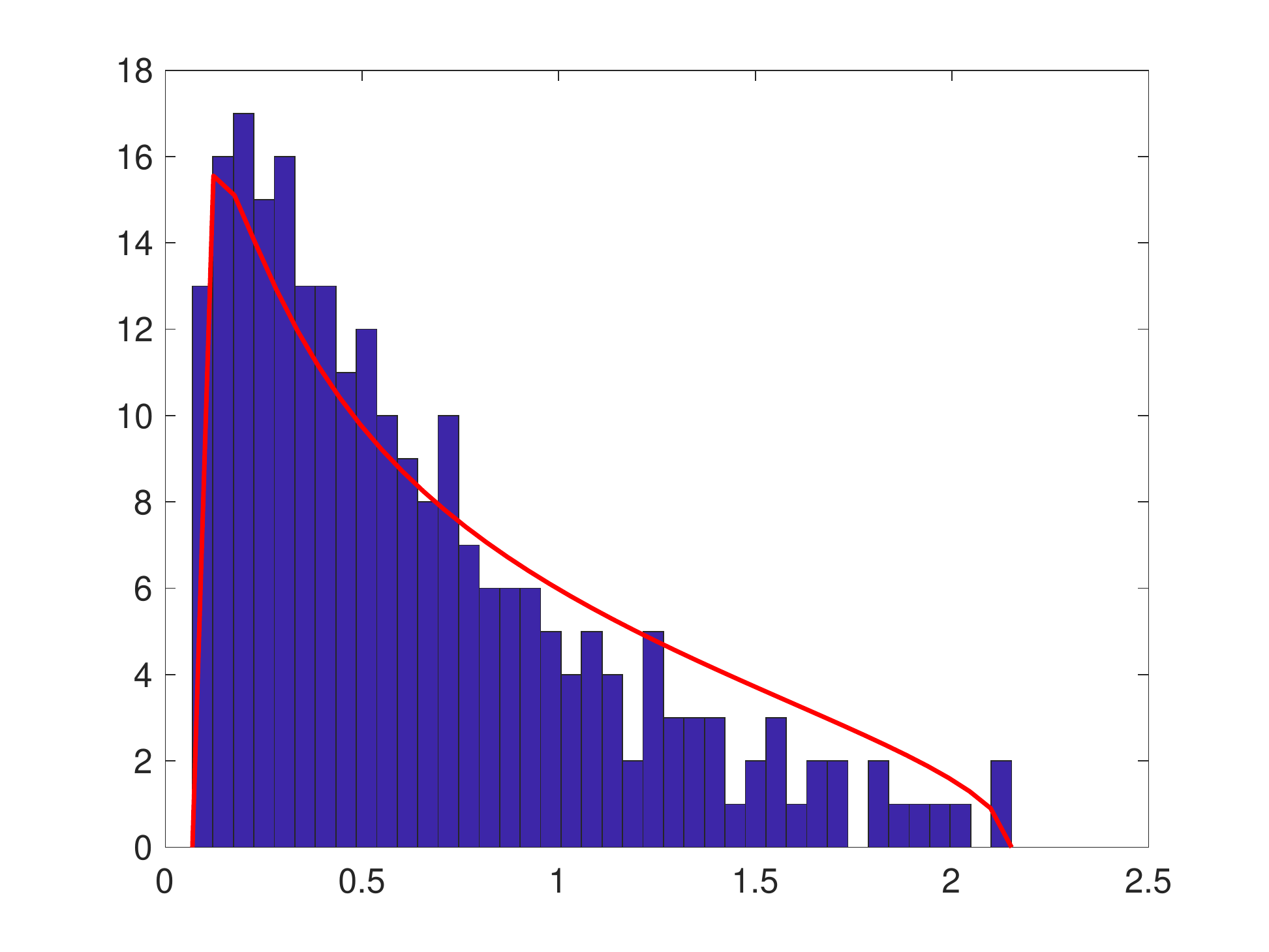}
        \caption{2014}
        \label{fig:2014MPfit}
    \end{subfigure}\\
\begin{subfigure}[b]{0.3\textwidth}
        \includegraphics[width=\textwidth]{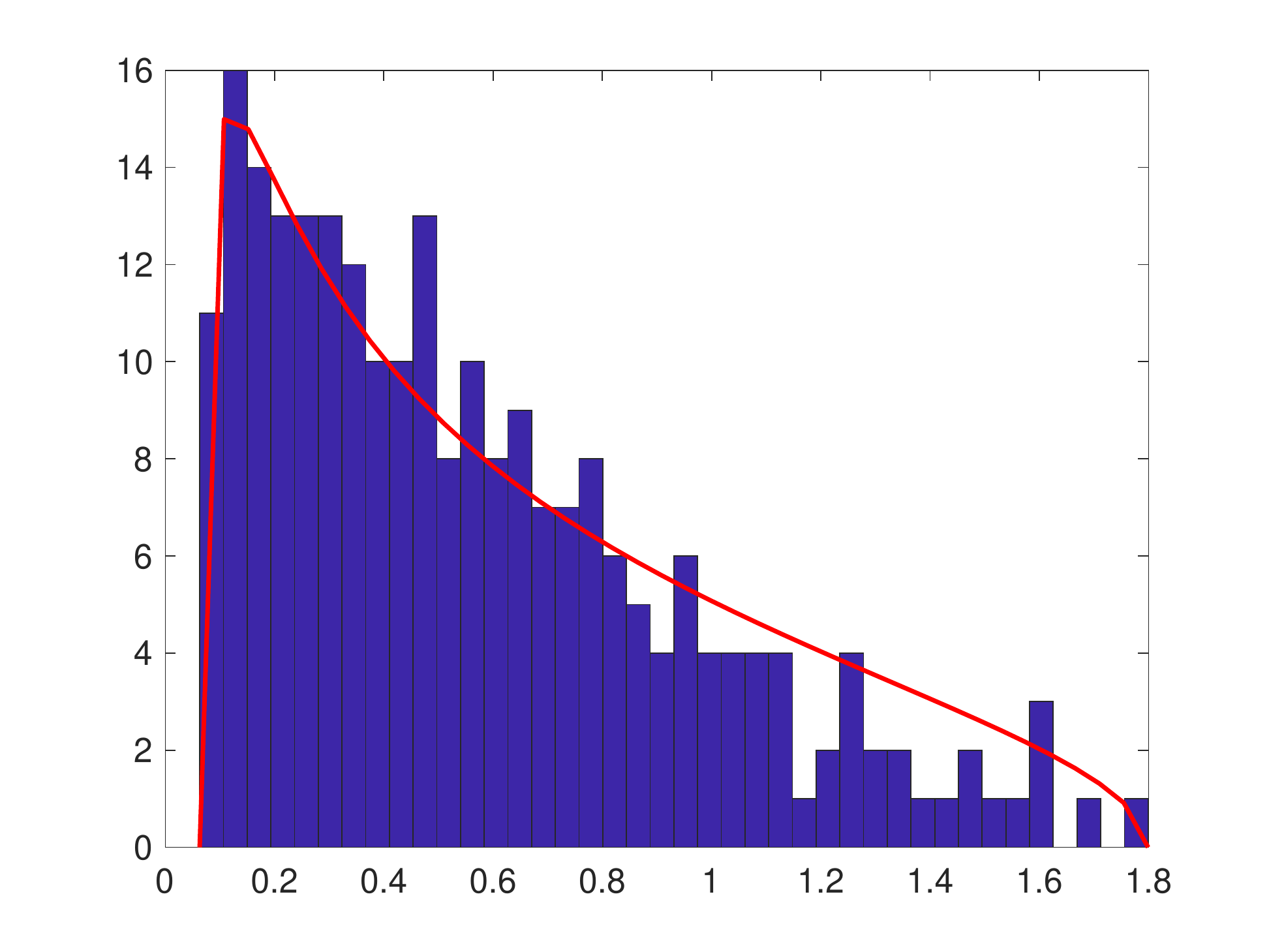}
        \caption{2015}
        \label{fig:2015MPfit}
    \end{subfigure}
\begin{subfigure}[b]{0.3\textwidth}
        \includegraphics[width=\textwidth]{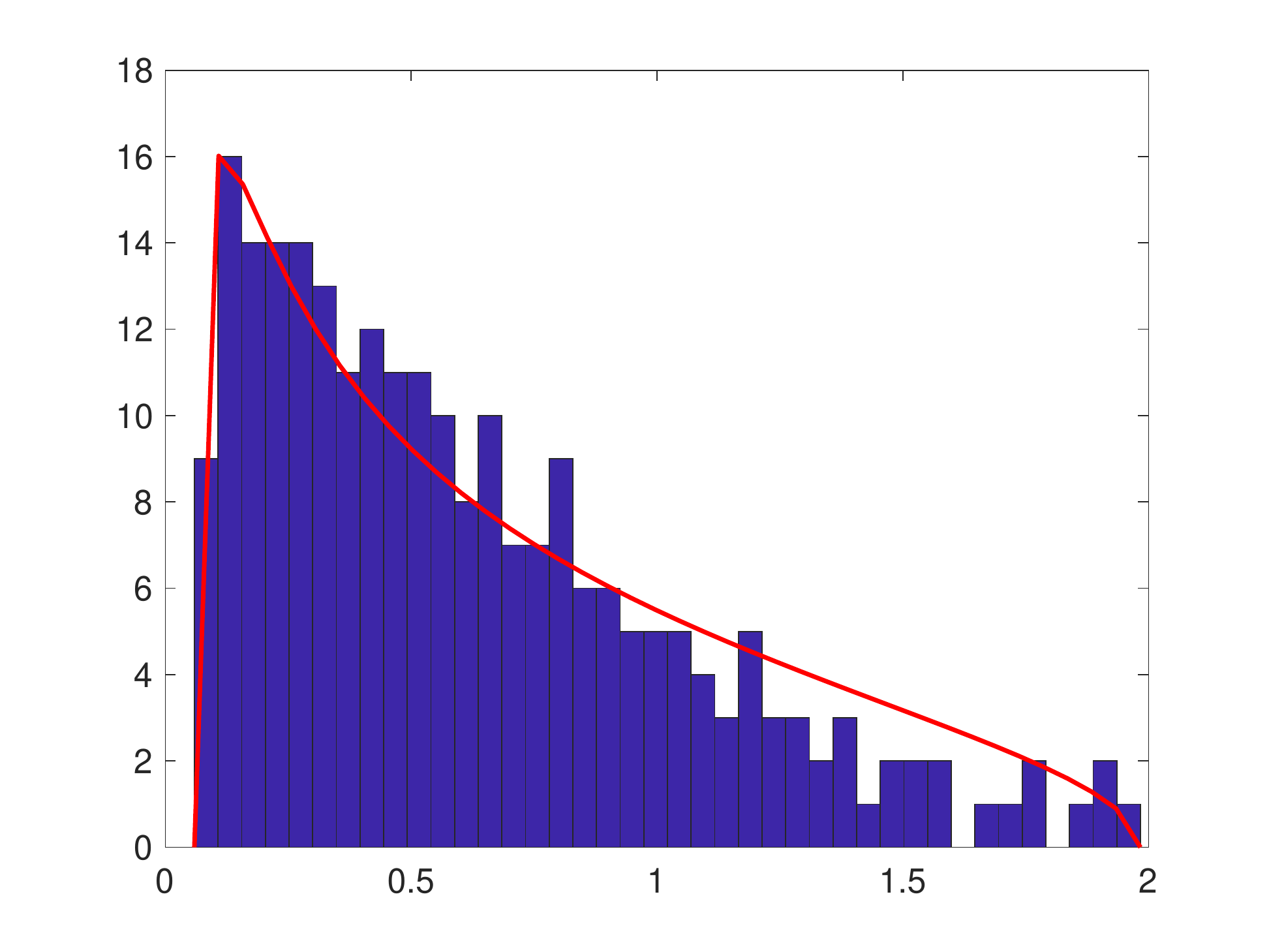}
        \caption{2016}
        \label{fig:2016MPfit}
    \end{subfigure}
\begin{subfigure}[b]{0.3\textwidth}
        \includegraphics[width=\textwidth]{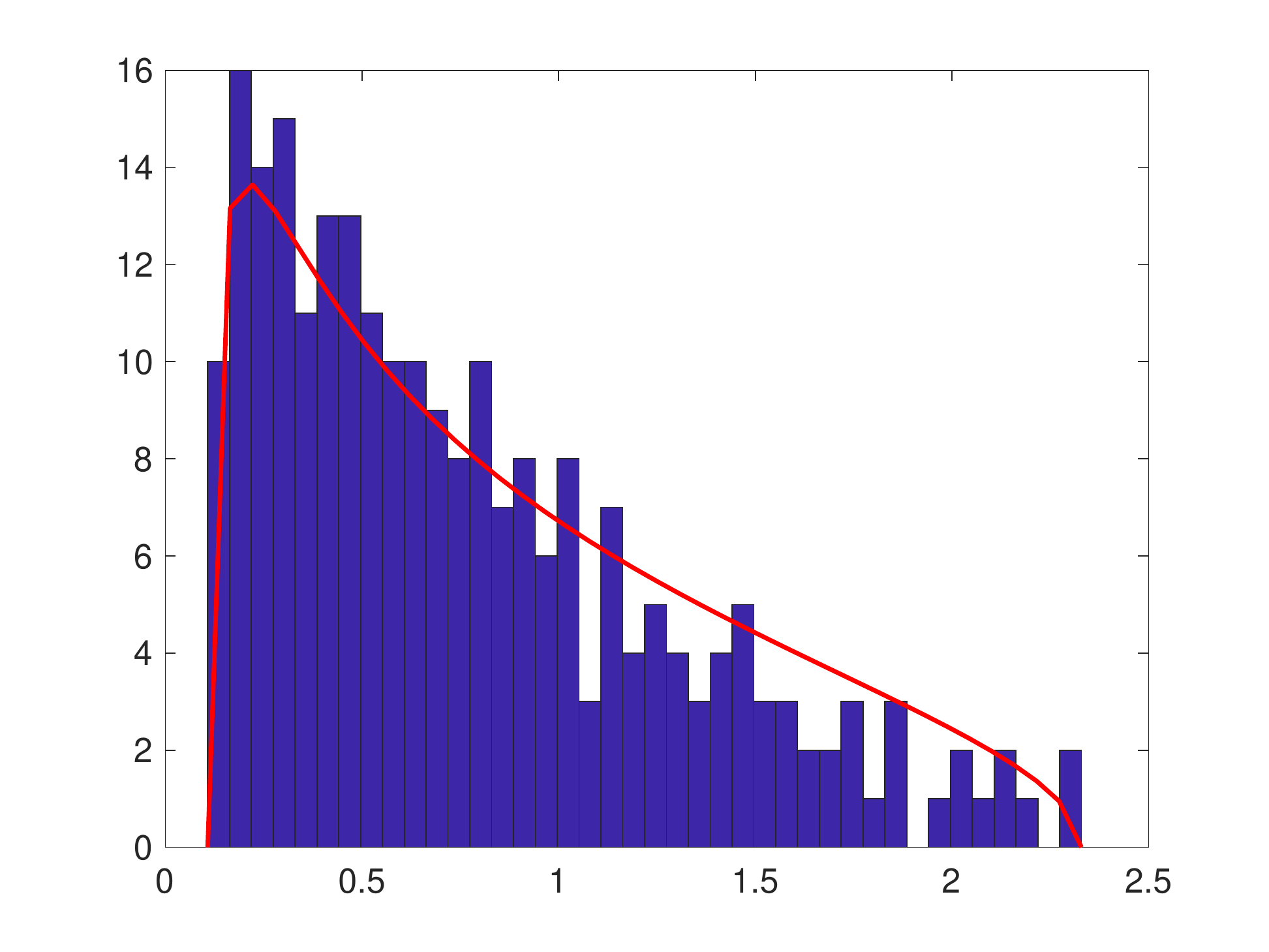}
        \caption{2017}
        \label{fig:2017MPfit}
    \end{subfigure}
     \caption{For each of the years from 2012 to 2017, we removed the top $d=9$ principal singular values and fitted a MP density to the histogram of the remaining $T-d$ squared singular values over $N$, i.e., $S_{ii}^2/N$ for $i>d$. For each of these years, a 2-sample Kolmogorov-Smirnoff test does not reject. Removal of more than 9 components also leads to non-rejection of the hypothesis that the residual matrices $\widetilde R$ have purely random entries.
     \label{fig:fittedRank_d_Rehiduals}}
\end{figure}

\subsection{Effective Dimension in Residual Matrices}

We will now introduce the notion of effective dimension of the data as follows. The parameter estimates given in \eqref{eq:lambdaPM_Est} and \eqref{eq:lambdaGamma_Est} for the dimension ratio $\lambda =T/N$ and the data standard deviation $\gamma$ imply that the empirical spectral density of the data fits a Marchenko-Pastur distribution associated with a random matrix whose dimensions are different from those of the actual data matrix $R$. This is because the IVS data has entries with significant correlations between them, and the effective dimension idea is a way to account for or quantify this feature. We could say roughly that the residual data matrix does not have independent entries, but it does have some kind of independence when the data is grouped in blocks and the number of such blocks plays the role of an effective dimension. This assessment is made using the empirical spectral density of the residual.

Denote by $\widetilde N = T/\hat\lambda $, which we call this the effective dimension associated with a residual that is purely random. In other words, the $N$-dimensional columns of $U_{d+1},U_{d+2},\dots$ and $U_{T}$ do not affect the spectrum and we have
\begin{equation}
    \label{eq:effDimEquiv}
    \frac{1}{N}\widetilde R^*\widetilde R= \frac{1}{N}\sum_{i=d+1}^TS_{ii}^2 V_iV_i^* \stackrel{\mathcal D}{\approx} \frac{1}{\widetilde N}Y^*Y\ ,
\end{equation}
where $Y$ is a $\widetilde N\times T$ matrix with purely random entries, and where ``$\stackrel{\mathcal D}{\approx}$" denotes approximate equality in distribution, in the sense of the fitting of the empirical spectral density to the MP law that was described in the previous section. Hence, we are looking for pure randomness in the temporal loadings, and we are not concerned if there is non-random structure remaining in the higher-order spatial components. In fact, for the implied volatility data there are clear patterns of non-random structure in the residual $\widetilde R$ even for $d$ large enough for non-rejection of the randomness hypothesis, which is clearly seen in Figure \ref{fig:Raw_Residuals_for_2017_IVS} for the year 2017 data. The reason for the non-random patterns is simple: the implied volatility's 4-dimensional tensor structure was flattened into a 2-dimensional $N\times T$ matrix using a lexicographical ordering for vectorizing the IVS data (name, strike and maturity) that prevails even after removal of PCA factors. This patterned structure does not mean that we have incorrectly concluded "pure" randomness in the residual as seen by the empirical spectral density, but rather it suggests that the spatial modes (name, strike and maturity) retain some of their structure and that the randomness we have concluded from the KS test is due to randomness in the temporal loadings. Indeed, randomness of temporal loadings is precisely what is suggested by the approximate distributional equivalence expressed in \eqref{eq:effDimEquiv}.
\begin{figure}[t]
\center
\includegraphics[width=0.6\linewidth]{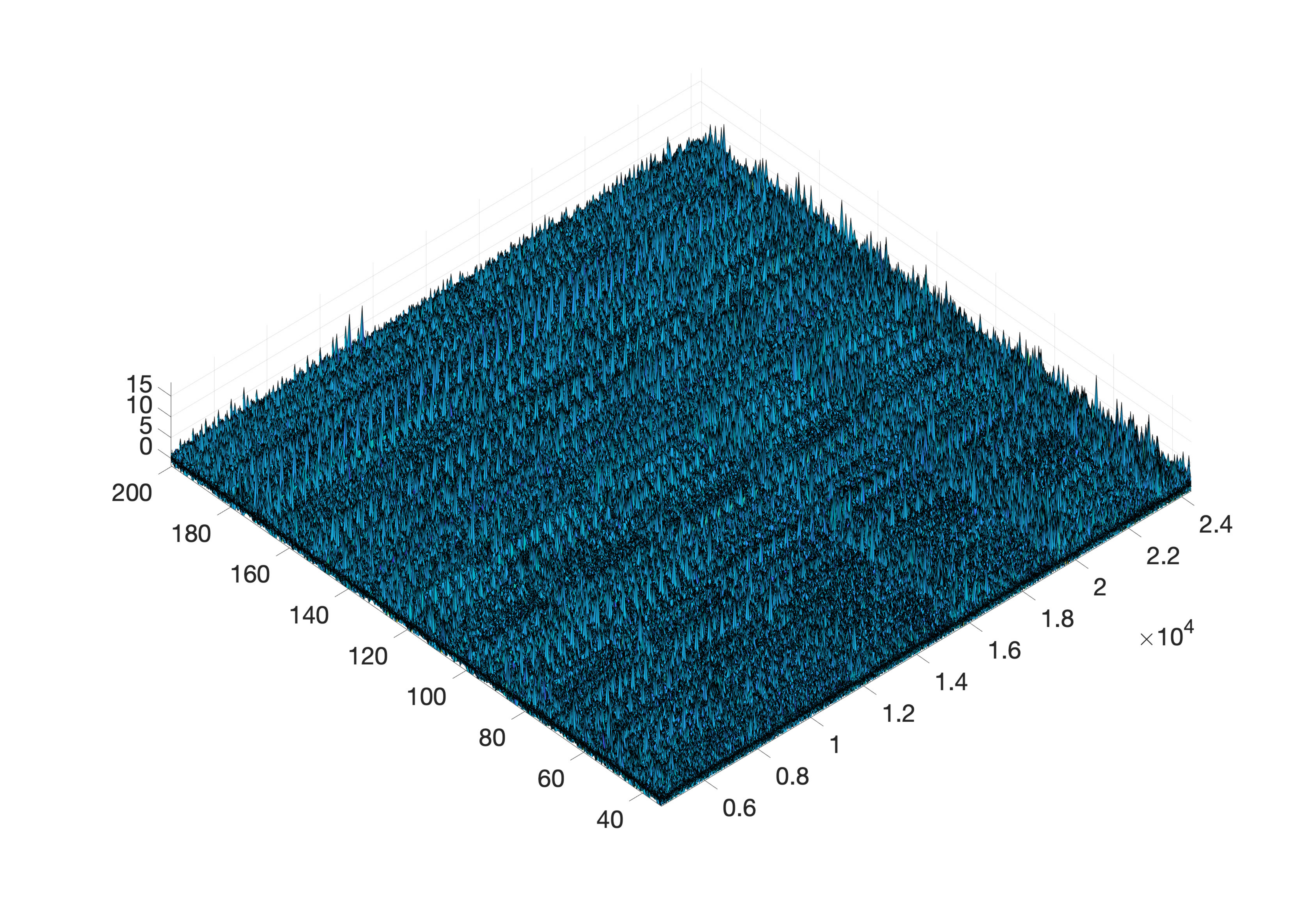} 
\caption{Non-random structure can be seen in the residual matrix $\widetilde R$ even after sufficiently many components have been removed for non-rejection by the KS test. The horizontal coordinate on the right is IVS name, delta and maturity (vectorized), and that on the left is time. This patterned structure suggests that the spatial modes are non-random and that the randomness we have concluded from the KS test is due to randomness in the temporal loadings. Random temporal loadings is precisely what is suggested in \eqref{eq:effDimEquiv}.}
\label{fig:Raw_Residuals_for_2017_IVS}
\end{figure}
\begin{figure}
\center
\includegraphics[width=0.6\linewidth]{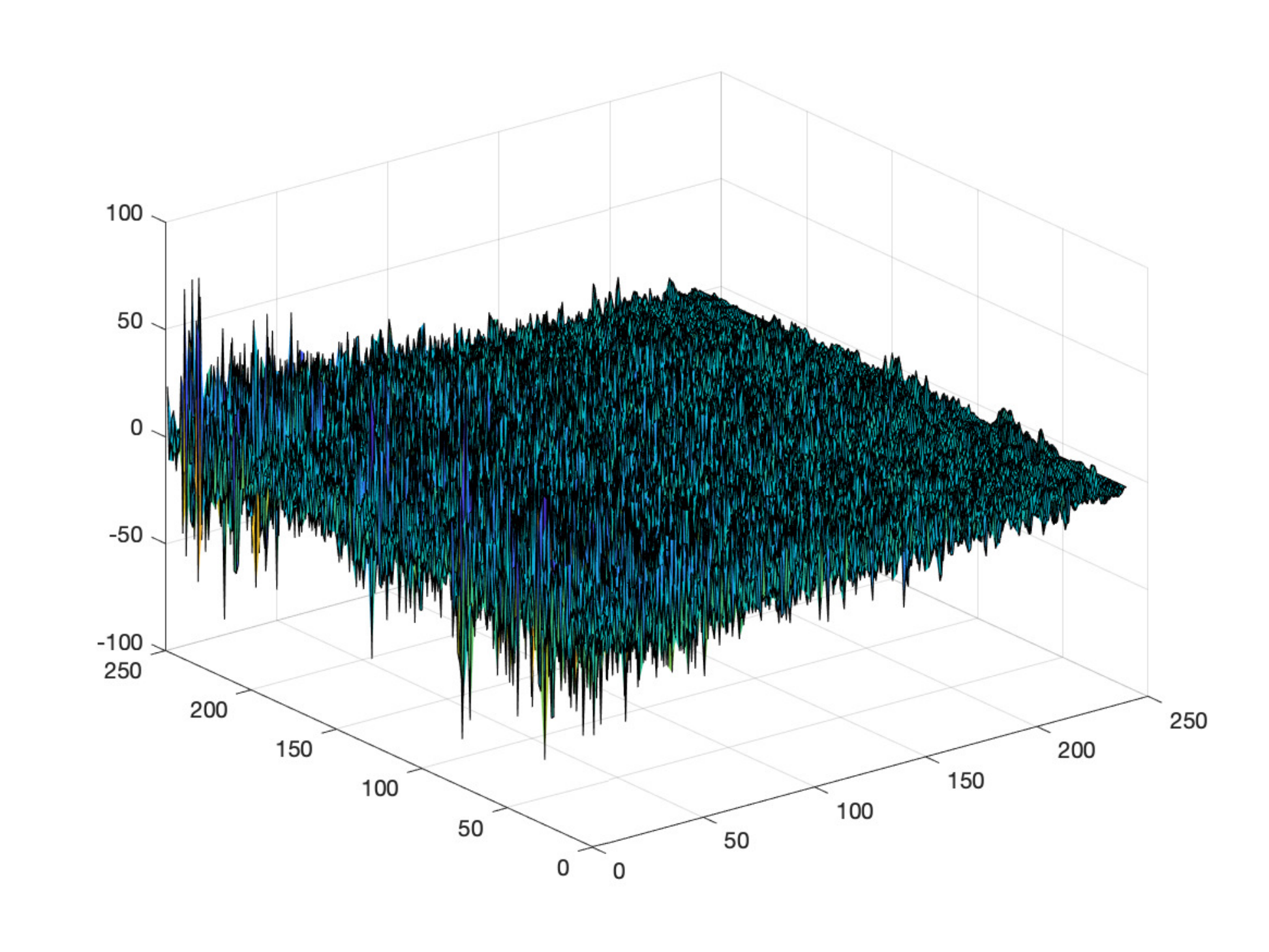} 
\caption{The matrix $S\widetilde V^*$ of temporal loadings for the 2017 data, of dimension $(T-d)\times T$ with $d=9$. The distinctive tapering is not an indication of non-randomness in the residual, but instead says that $S\widetilde V^*/N$ is close in distribution to $\Sigma Q^*$ where $Q\Sigma Q^*$ is the spectral form of another matrix $Y^*Y/\widetilde N$ and $Y$ is an $\widetilde N\times T$ matrix with purely random entries of variance $\hat\gamma^2$. }
\label{fig:EPS_Residuals_for_2017_IVS}
\end{figure}
\begin{table}
    \centering
    \begin{tabular}{c|c|c|c|c|c}
    year&$\hat\lambda_+$&$\hat\lambda_-$&$\hat\gamma$&$\hat\lambda$&$\widetilde N$\\
    \hline
      2012   &1.77 & 0.07 & 0.80 & 0.44 & 573 \\
      2013   &1.93 & 0.10 & 0.85 & 0.40 & 633 \\
      2014   &2.15 & 0.07 & 0.86 & 0.49 & 519 \\
      2015   &1.80 & 0.06 & 0.80 & 0.47 & 537 \\
      2016   &1.98 & 0.06 & 0.83 & 0.50 & 506 \\
      2017   &2.33 & 0.11 & 0.93 & 0.42 & 601 
    \end{tabular}
    \caption{The estimated parameters for the MP distribution and the effective dimension after the removal of 9 principal components, for each of the 251 or 250 IVS daily returns observed in each year. The data points are $X_i= S_{ii}^2/N$ for $i>9$, the estimates are given by \eqref{eq:lambdaPM_Est} and \eqref{eq:lambdaGamma_Est}, and the effective dimension is $\widetilde N = T/\hat \lambda$. Notice in particular that $\widetilde N\ll N \sim 25,000$. }
    \label{tab:eff_dim}
\end{table}

Denote the $t^{th}$ column of $R$ as $R_t$, which we can write as
\[R_t = \sum_{i=1}^d f_i\theta_{it} + \sum_{i=d+1}^TS_{ii}U_iV_{it}\ ,\]
where $\theta_{it}$ and $V_{it}$ are the $t^{th}$ entry of $\theta_i$ and $V_i$, respectively; $U_i$ is the $i^{th}$ column of $U$. Clearly $\theta_{it}$ are the temporal loadings on the $i^{th}$ principal factor, and for $i>d$ the temporal loadings are $S_{ii}V_{it}$. Denote the higher-order temporal modes as $\widetilde V = [V_{d+1},\dots, V_T]$. The KS test has indicated randomness of these loadings for $i>d$, which means the $(T-d) \times T$ matrix $S \widetilde V^*$ has covariance spectrum close in distribution to that of a random matrix. But this matrix has a distinctive tapering that would counter any claim of pure randomness. Indeed, Figure \ref{fig:EPS_Residuals_for_2017_IVS} shows this tapering for the 2017 data for the $S_{ii}^2/N$ with $i>9$, that is, with $d=9$. However, if one considers a purely random matrix $Y$ of dimension $T\times \widetilde N$ with entries having variance $\hat\gamma^2$, and expresses it in spectral form such that $Y^*Y/\widetilde N=Q\Sigma Q^*$, then $\Sigma Q^*$ is a tapered matrix that is derived from a purely random matrix. Moreover, the squared singular values of $\Sigma Q^*$ fit the MP law and are equal in distribution to the squared singular values of $S\widetilde V^*/N$. 

The notion of effective dimension as we have introduced it here is robust and useful in understanding the real information content of the IVS data. Note in particular that in Table \ref{tab:eff_dim}, which has the parameter estimators for the MP fit as well as the effective dimension for each of the 6 years we've considered, a striking finding is that we have a very low effective dimension for each year, namely $500\sim\widetilde N\ll N\sim 25,000$. The effective dimension of the IVS residuals in a factor decomposition indicates that there is structure in these residuals; they are not purely random matrices.

\section{Principal Eigenportfolios and OI-Weighted Indices}
\label{sec:eigenPortfolios}

The portfolio constructed from the first eigenvector of the normalized covariance, the eigenportfolio, has been analyzed in \cite{avellaneda2010} and \cite{boyle2014} and elsewhere for equities returns. In this section, we will apply a similar approach and terminology to construct a portfolio of implied volatilities and compare it to an analog of a market portfolio. Although such portfolios can exist in theory, they are not directly\footnote{A portfolio can be constructed with returns equal to a linear combination of implied volatilities returns, i.e., a linear combination of the $r_i(t)$'s.  However, such a portfolio would probably incur a roll yield due to the daily rebalancing required to maintain a constant maturity and constant delta position.} tradeable. However it is still of interest to introduce them and use them as a proxy for the market's collective volatility, similar to the VIX index. 

Central to our analysis of the eigenportfolio is its ``closeness" to an open-interest (OI)-weighted factor of our construction. Analogous to the case in equities where the eigenportfolio's returns track the capitalization-weighted market portfolio (with varying degrees of closeness), we show that the implied volatility eigenportfolio's returns track closely to various OI-weighted factors or indices. It is essential to include the OI as it is the implied volatilities' analogue to capitalization, thereby capturing each option contract's importance.\footnote{In practice, the CBOE's VIX calculation uses OI too. Specifically, the VIX construction includes only SPX options that are between the at-the-money mark and the last strike before the first two consecutive strikes that have zero OI.} Both in the analysis of implied volatilities and equities, the underlying spike model (see \cite{benaych2011eigenvalues}) provides a mathematical framework for closeness of the eigenportfolio and factor returns. Figure \ref{fig:2017_ATM_OI} shows daily OI amounts for at-the-money options for the year 2017. It is important to note the spikes and the zeros in the OI. Most options have zero open trades at any point in time, which is an indication that a rather small percentage of options can explain changes in implied volatilities.

\begin{figure}
    \centering
\begin{subfigure}[b]{0.4\textwidth}
        \includegraphics[width=\textwidth]{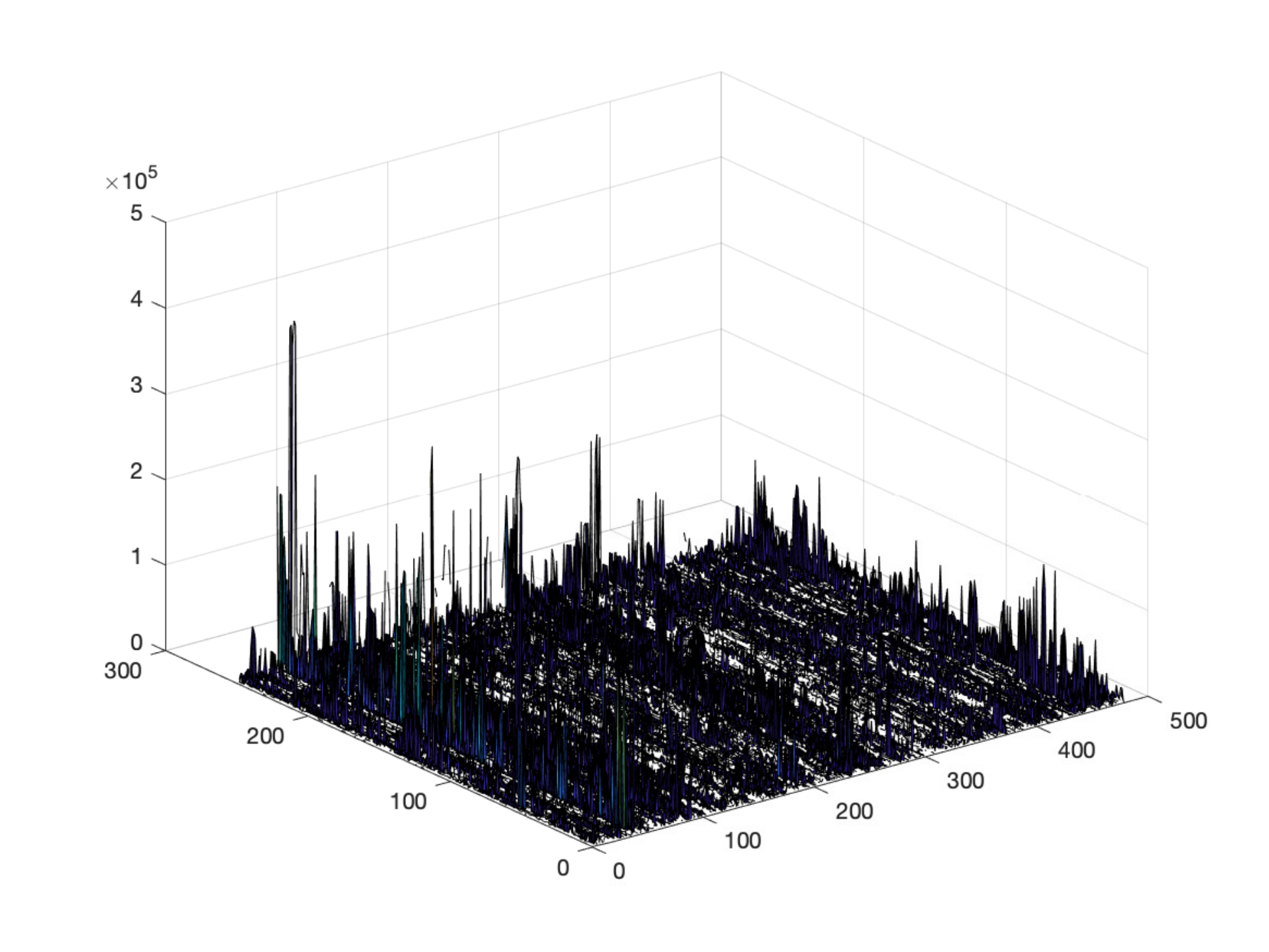}
        \caption{ATM Front Month}
        \label{fig:2017_ATM_OI_1sr}
    \end{subfigure}
\begin{subfigure}[b]{0.4\textwidth}
        \includegraphics[width=\textwidth]{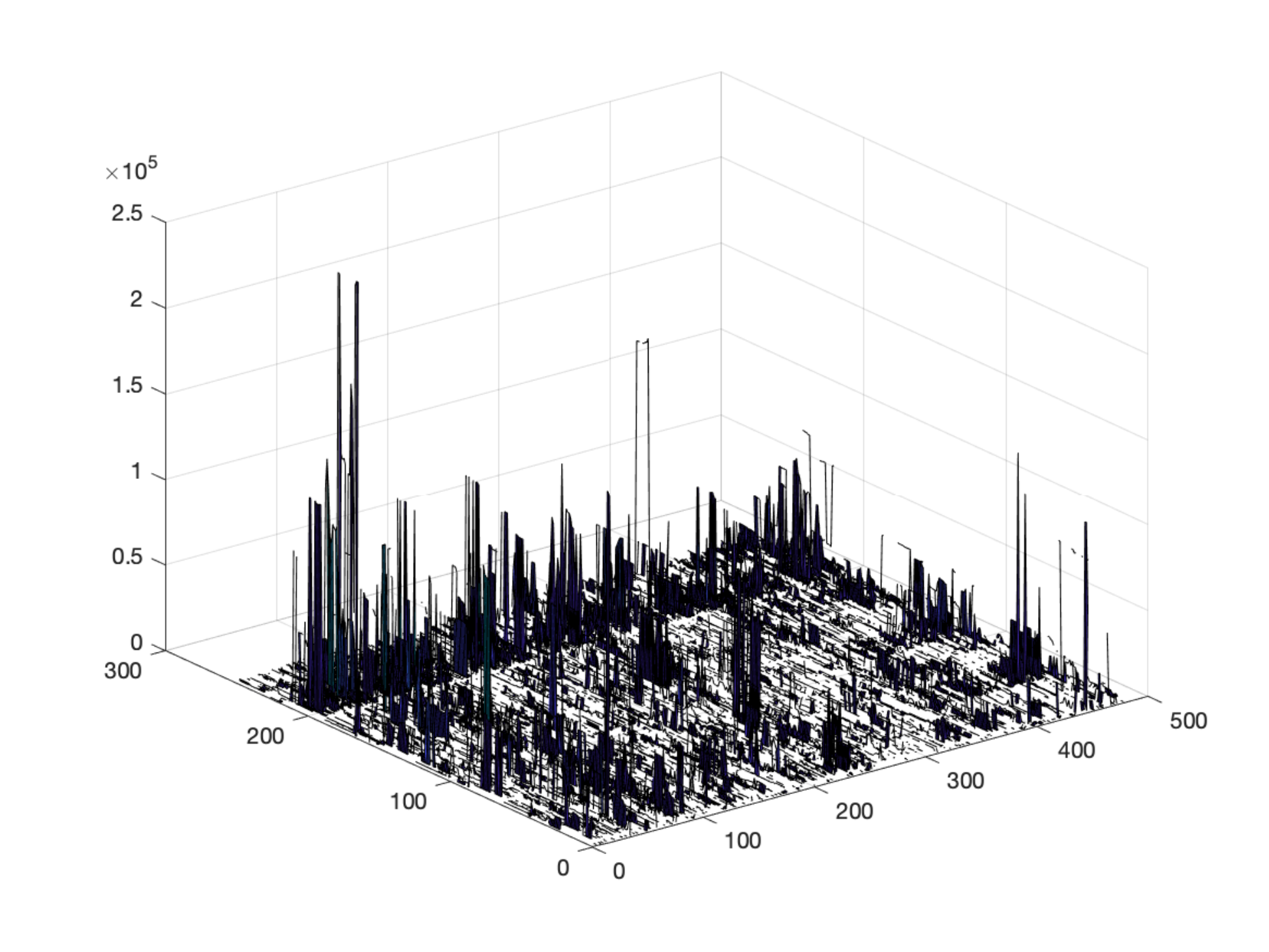}
        \caption{ATM Third Month}
        \label{fig:2017_ATM_OI_3rd}
    \end{subfigure}
     \caption{The daily values of open interest (OI) for the at-the-money options with 30 days to maturity (the front month) and 90 days to maturity (third month). The horizontal coordinate on the left is time and on the right option name. Note the spikes and the zeros in the OI, which indicate that only a small percentage of options can explain changes in implied volatilities because most options do not have open trades any on given day. \label{fig:2017_ATM_OI}}
\end{figure}

\subsection{``Trading" Implied Volatility and Construction of OI-Based Factors}

It is not only OI, however, that must be taken into consideration in constructing IVS returns portfolios. We must also account for the sensitivity of IVS to fluctuations as measured by the Vega, for example. We can motivate this sensitivity in portfolios we are about to create by first introducing a synthetic market of exchange-traded notes (ETNs). For each option we can think of an ETN whose prospectus states the daily returns to be 
\begin{align}
    \label{eq:ETNreturns}
    \frac{dE_i(t)}{E_i(t)}: = r_{i}(t)\ ,
\end{align}
where $r_i(t)$'s are the components of the implied volatility return vector defined in \eqref{eq:returns}, and $dt=1/252=$ 1 day. The ETN whose returns are given by \eqref{eq:ETNreturns} are the stochastic component in the returns of a $\Delta$-neutral options position. Indeed, letting $C_i(t)$ denote the option price with underlying price, time-to-maturity and delta $(S_i(t),\tau_i,\Delta_i)$, from It\^o's lemma we compute the (unitless) differential of a $\Delta$-neutral position (up to a term of size ``Big-Oh"),
\begin{equation}
    \label{eq:deltaNeutralReturns}
    \frac{dC_i(t) - \Delta_idS_i(t)}{S_i(t)} = \frac{\mathcal V_i(t)}{S_i(t)}d \hat\sigma_i(t)+O(dt)= \mathcal V_i^{untls}(t)r_i(t)+O(dt)\ ,
\end{equation}
where $\mathcal V_i(t)$ is the Vega for the $i^{th}$ option and
\begin{align}
\label{eq:unitlessvega}
 \mathcal V_i^{untls}(t) = \frac{\hat\sigma_i(t)\mathcal V_i(t)}{S_i(t)}\ ,
\end{align}
is a unitless Vega, which is the dollar-Vega divided by the price of the underlying. We will use \eqref{eq:unitlessvega} in constructing weighing factors, in addition to OI, after we first discuss them more generally. 

We would like to construct a global factor that can describe upwards of 50\% of daily variance for all of the IVS ETNs. In equities, such a factor is the market portfolio, which suggests to us that the number of outstanding shares (contracts) should have some bearing on the relevance of an individual equity in factor construction. Indeed, in options this is precisely the OI, and a general form for a global factor with only OI weighing is
\begin{align}
    \label{eq:OIfactor}
    \frac{dQ(t)}{Q(t)}:=\frac{\sum_i\omega\Big(\mathcal{O I}_i(t)\Big)r_i(t)}{\sum_i\omega\Big(\mathcal{O I}_i(t)\Big)}\ ,
\end{align}
where $\mathcal{O I}_i(t)$ is the OI for the $i^{th}$ option and $\omega(~\cdot~)$ is a weighting function of our choosing. In the simplest case we have $d=1$ in \eqref{eq:lowRankMatrixModel}, and the ETN returns have a simple factor-based returns model,
\begin{align}
    \label{eq:CAPM_return_model}
    \frac{dE_i(t)}{E_i(t)}= \beta_i\frac{dQ(t)}{Q(t)} +\xi_i(t)\ ,
\end{align}
where $\xi_i$ is an idiosyncratic noise component independent of $Q(t)$. Ordinary least squares regression shows us that the $\beta_i$'s are given by the covariance with the factor, 

\begin{align}
\beta_i = \hbox{cov}\left(\frac{dE_i}{E_i},\frac{dQ}{Q}\right)/h_q^2  \ ,  
\end{align} 
where $h_q^2 = \hbox{var}(\frac{dQ}{Q}) $. 

\begin{remark}
\cite{boyle2014} explains how the principal eigenportfolio is a frontier portfolio if the Perron-Frobenius theorem applies. However, for options the frontier/CAPM theory does not apply because the lifetime of an option is too short. Therefore, comparisons with equities are merely an informal, statistical analogy. 
\end{remark}

\subsection{The Spike Model and the Principal Eigenportfolio}

Let us consider the $N\times N$ empirical covariance matrix for the returns of the synthetic ETNs,
\begin{align*}
\widehat\Sigma_{ij} := \frac{1}{T-1} \sum_{t=1}^T (r_i(t)-\bar r_i)(r_j(t)-\bar r_j)\qquad\hbox{for }1\leq i,j\leq N .
\end{align*}
Using the returns model in \eqref{eq:CAPM_return_model}, the population covariance matrix is
\begin{equation}
    \label{eq:equitySpikeModel}
    \Sigma
= h_q^2\beta\beta^* + \Omega\ ,
\end{equation}
where $\Omega$ is the covariance matrix with $\Omega_{ij} = \hbox{cov}(\xi_i(\cdot),\xi_j(\cdot))$, and $h_q^2 = \hbox{var}\left(\frac{dQ(\cdot)}{Q(\cdot)}\right)$. Equation \eqref{eq:equitySpikeModel} is a spike model, as referred to above, because the $\beta$'s describe a substantial portion of variance and cause a single eigenvalue to stick out from the rest of the spectrum. The $1^{st}$ principal component of the ETN empirical covariance matrix will be nearly proportional to the $\beta$'s of the OI-weighted portfolio if the $\xi(t)$ covariance is not too large. For the model in \eqref{eq:lowRankMatrixModel}, the distribution of the empirical matrix principal component is shown in \cite{benaych2011eigenvalues} to be within a cone surrounding the spike model's low-rank component if the difference between $\|\theta\|^2$ and the variances of the noise is over a critical amount. For the spike model in \eqref{eq:equitySpikeModel}, the critical threshold is crossed if $h_q^2\|\beta\|^2$ exceeds a threshold determined by the covariances of the $\xi_i(t)$'s, which should happen as $N$ grows.

In finance there usually are differing sizes among the $\Sigma_{ii}$'s, which means better statistical estimation of principal eigenvectors results from consideration of correlations rather than covariances, which is the normalization issue we noted earlier. The ETNs' empirical correlation matrix is
\[\widehat\rho =  h^{-1}\widehat \Sigma  h^{-1}\ ,\]
where $ h$ is a diagonal matrix of standard deviations $h_i$ defined in \eqref{eq:standardizedReturns}. Letting $u_1$ denote the principal eigenvector of $\widehat\rho$, the spike model suggests $u_1\approx \frac{1}{c}h^{-1}\beta$ for $c$ a normalizing constant. Using another (orthogonal) eigenvector $\widetilde u$ such that $\widetilde u\perp u_1$, we can also construct portfolios as done in \cite{avellaneda2010},
\begin{align*}
      \pi_1=\frac{ h^{-1}u_1}{\sum_i( h^{-1}u_1)_i}\ ,\hspace{1.5cm}
    \widetilde \pi=\frac{ h^{-1}\widetilde u}{\sum_i( h^{-1}\widetilde u)_i}\ ,
\end{align*} 
which are orthogonal in the sense that covariance of these portfolios' returns is zero,
\[\pi_1^*\widehat\Sigma\widetilde \pi = u_1^*\widehat\rho~ \widetilde u= \lambda_1u_1^*\widetilde u = 0\ ,\]
where $\lambda_1>0$ is the principal eigenvector such that $\widehat\rho u_1 = \lambda_1u_1$. 

\begin{proposition}
    Returns of the top eigenportfolio tend toward the factor returns plus some tracking error. Returns of the orthogonal portfolios tend toward factor neutrality.
\end{proposition}

\begin{proof}
Assuming the parameters are such that we are over the critical levels in \cite{benaych2011eigenvalues} and letting $EP_1(t)$ be
the principal eigenportfolio, we have
\begin{align*}
    &\frac{dEP_1(t)}{EP_1(t)} -\frac{\sum_i(\beta_i/h_i)^2}{\sum_i\beta_i/h_i^2}\frac{ dQ(t)}{Q(t)}\\
    &=\left( \sum_i \left(\pi_{1i}-\frac{\beta_i/h_i^2}{\sum_j\beta_j/h_j^2}\right)\beta_i\right)\frac{ dQ(t)}{Q(t)} + \sum_i\pi_{1i}\xi_i(t)\\
    &\rightarrow \varepsilon_1(t)\ ,
\end{align*}
where factor returns disappear because $\pi_{1i}-\frac{\beta_i/h_i^2}{\sum_j\beta_j/h_j^2}\rightarrow 0$ as the random matrix's dimensions grow, and where $\varepsilon_1(t)=\lim_N\sum_i\pi_{1i}\xi_i(t)$ is the tracking error. All orthogonal portfolios $\widetilde \pi$ are approximately factor neutral,
\[\frac{d\widetilde{EP}(t)}{\widetilde{EP}(t)}  
=\underbrace{\left( \sum_i \widetilde\pi_i\beta_i\right)}_{\approx 0}\frac{dQ(t)}{Q(t)}  + \sum_i\widetilde\pi_i\xi_i(t) \rightarrow \widetilde\varepsilon(t)\ ,\]
where $\widetilde\varepsilon(t) = \lim_N \sum_i\widetilde\pi_i\xi_i(t)$, and where the limit happens because $(h^{-1}\widetilde u)^*\beta \approx \widetilde u^* h^{-1} h u_1 = 0$; the ``$\approx$" becomes more accurate and tends toward an equality as the gap between $\|h^{-1}\beta\|$ and the critical level is increased.
\end{proof}

\subsection{Empirical Analysis}

The $1^{st}$ spatial\footnote{That is, of the IVS vector of names, deltas and maturities} singular vector $U_1$ computed by the SVD in \eqref{eq:SVD} is the empirical estimator for $u_1$, and the estimator for each $h_i$ is the empirical standard deviation of each $r_i$. Hence, we can compute the eigenportfolio from the data, and then compare it with the empirically estimated $\beta$'s. In our studies we will consider two weighting functions,
\begin{align}
    \label{eq:weightingFunctions}
    &\omega(\mathcal{OI}) =\mathcal{OI}\qquad\hbox{or}\qquad\omega(\mathcal{OI}) =\log(1+\mathcal{OI})\times\mathcal V^{untls}\ ,
\end{align}
where $\mathcal V^{untls}$ denotes the unitless Vega of \eqref{eq:unitlessvega}. Generally speaking, $\omega(\mathcal{OI}) =\mathcal{OI}$ results in a factor with a less signficant intercept in ex-post regressions of eigenportfolio returns onto the factor returns, whereas the $\omega(\mathcal{OI}) =\log(1+\mathcal{OI})\times\mathcal V^{untls}$ results in the same regression having significant intercept but lower projection error. The plain OI weighting is a bit strange, however, because it counts contracts without taking into account the sensitivity of the contract to a change in the volatility of the underlying stock. It should also be noted that the log-weighting has a factor loading that is closer to unity, whereas the plain OI-weighting that is significantly less than unity. 

The OI and unitless Vega weighing does in fact matter and it performs better when the tensor data structure is taken into consideration. We have chosen the combined OI and unitless Vega weighing in the form $\log(1+\mathcal{OI})\times\mathcal V^{untls}$, which works well for the IVS data, although it is rather arbitrary at present.

In Figure \ref{fig:sortedEigenportfolios}, we see the comparison for each of the years, with each year's eigenportfolio computed using the 251 or 250 days of data from the year, and with a $Q(t)$ factor computed using the weight function $\omega(\mathcal{OI}) = \mathcal{OI}$. If the eigenportfolio's weights are sorted in descending order (i.e., we sort $\pi_1$ in descending order) and then the sorting index is used to permute the vector $ \hbox{diag}^{-2}(h)\beta/\sum_i(\hbox{diag}^{-2}(h)\beta)_i$, then the sorted vector and the permuted vector should line up. Indeed, the plots in Figure \ref{fig:sortedEigenportfolios} show this lining up, with the sorted eigenportfolio in red and the permuted $\beta$'s in blue. 
\begin{figure}[h!]
    \centering
    \begin{subfigure}[b]{0.3\textwidth}
        \includegraphics[width=\textwidth]{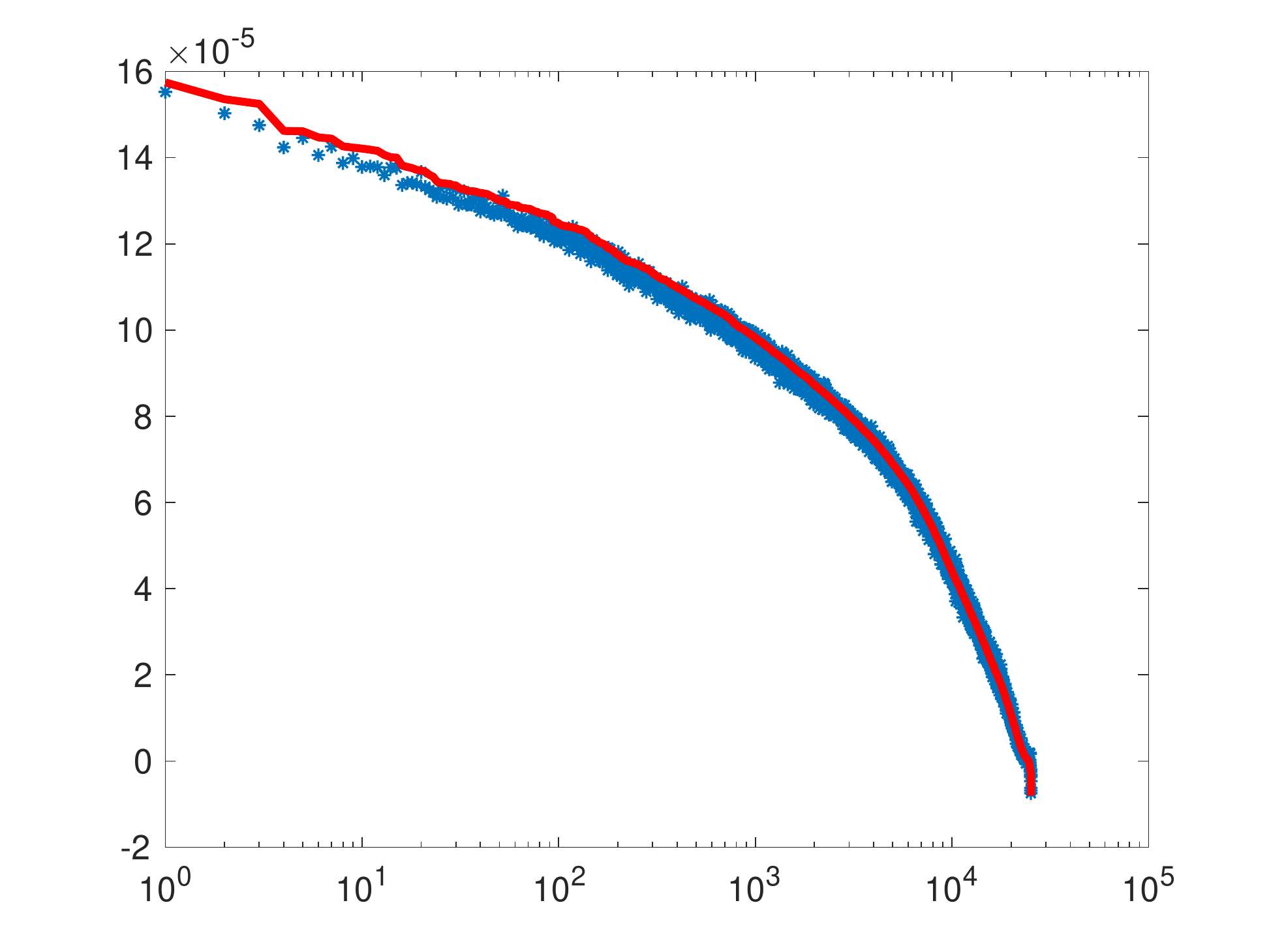}
        \caption{2012}
        \label{fig:2012IVSsorted}
    \end{subfigure}
\begin{subfigure}[b]{0.3\textwidth}
        \includegraphics[width=\textwidth]{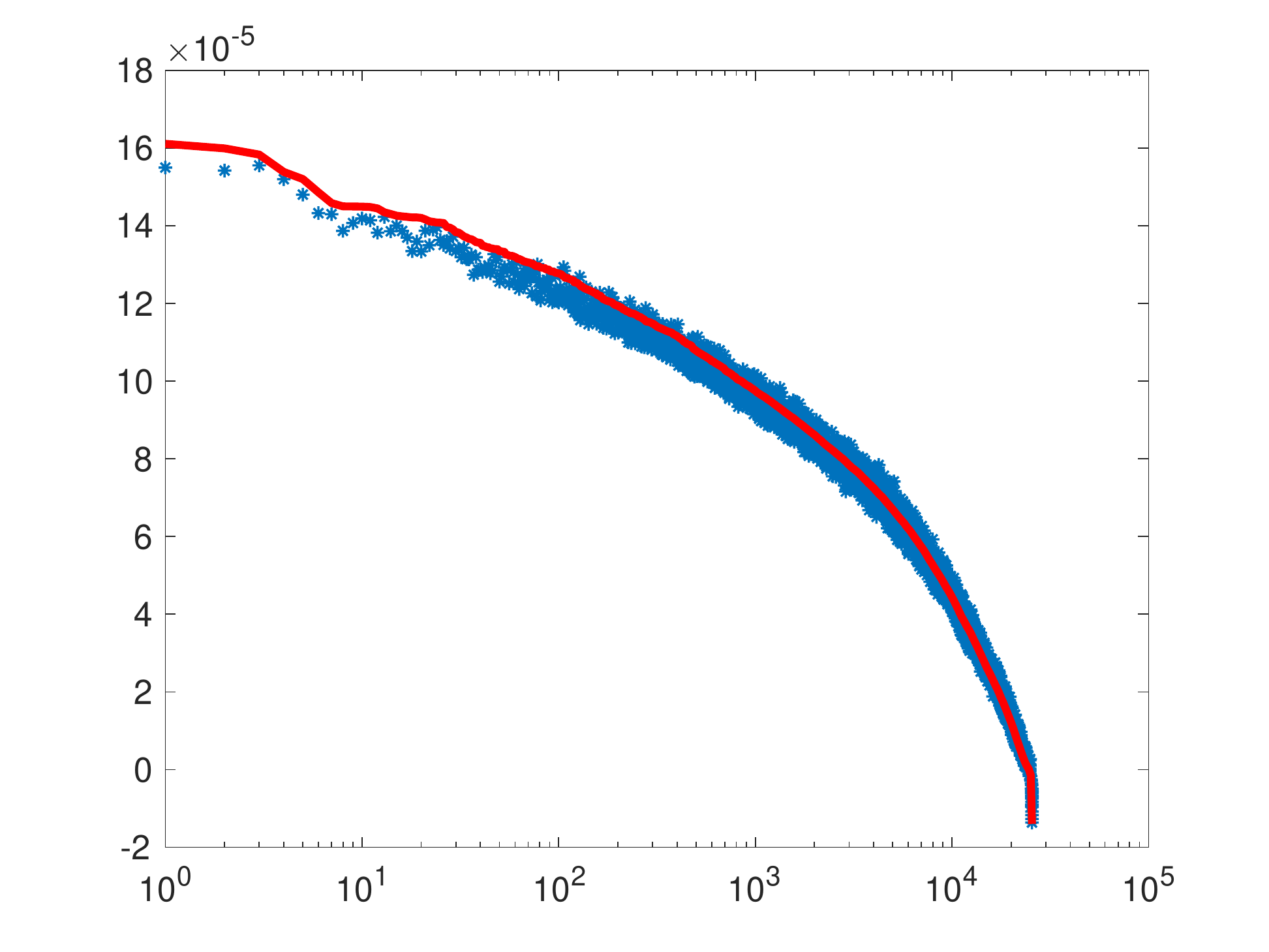}
        \caption{2013}
        \label{fig:2013IVSsorted}
    \end{subfigure}    
\begin{subfigure}[b]{0.3\textwidth}
        \includegraphics[width=\textwidth]{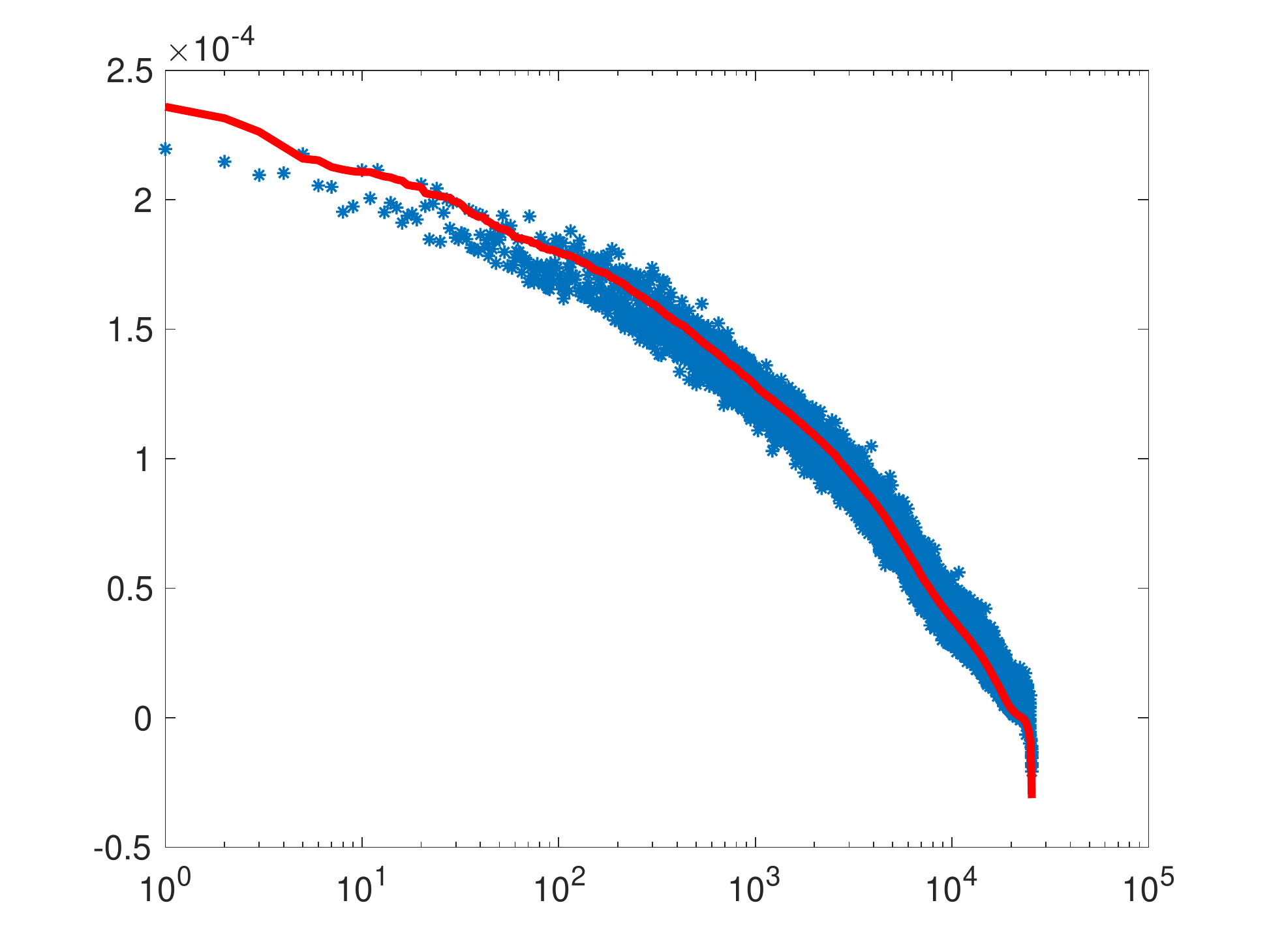}
        \caption{2014}
        \label{fig:20164IVSsorted}
    \end{subfigure}\\
\begin{subfigure}[b]{0.3\textwidth}
        \includegraphics[width=\textwidth]{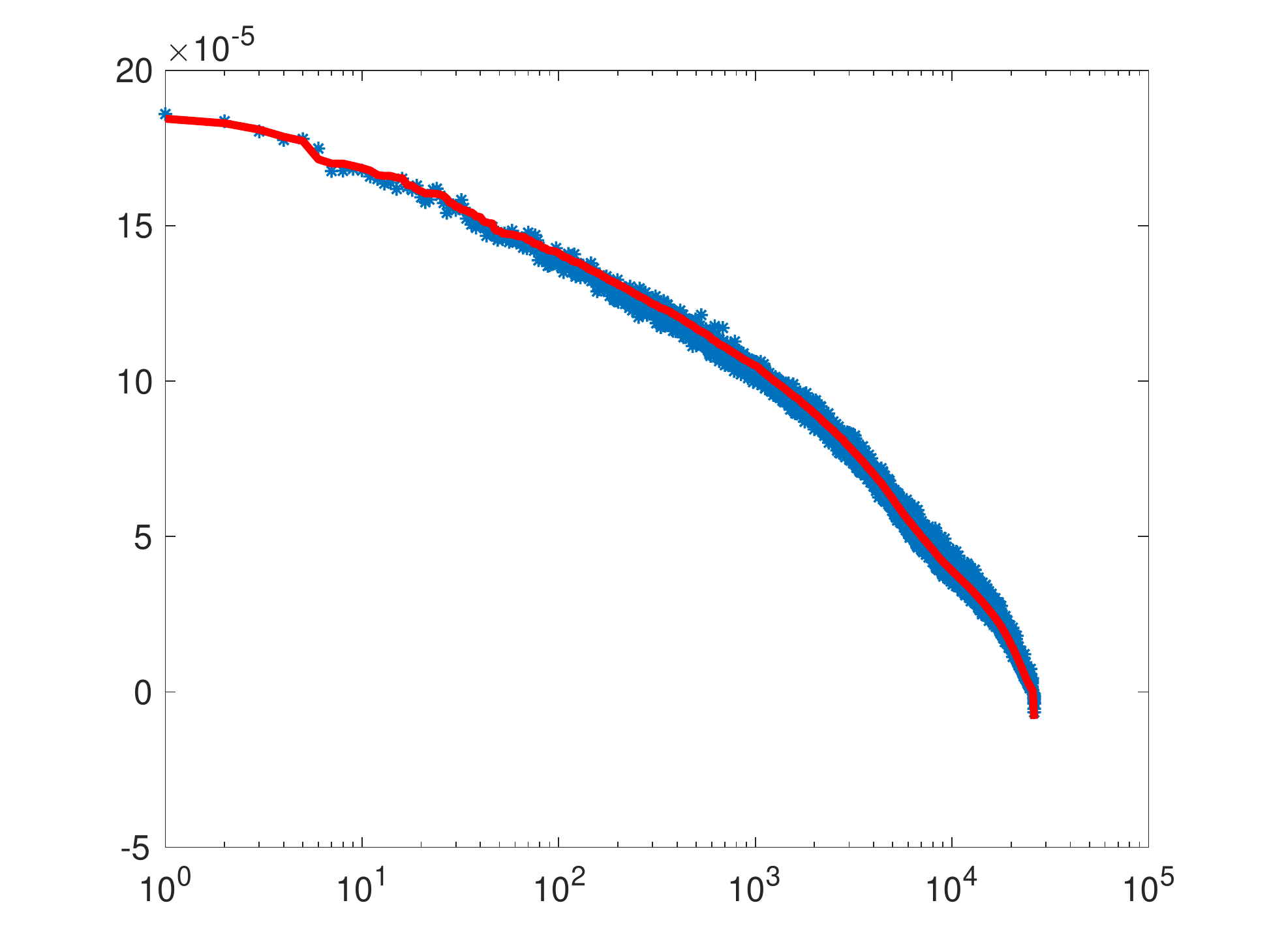}
        \caption{2015}
        \label{fig:2015IVSsorted}
    \end{subfigure}
\begin{subfigure}[b]{0.3\textwidth}
        \includegraphics[width=\textwidth]{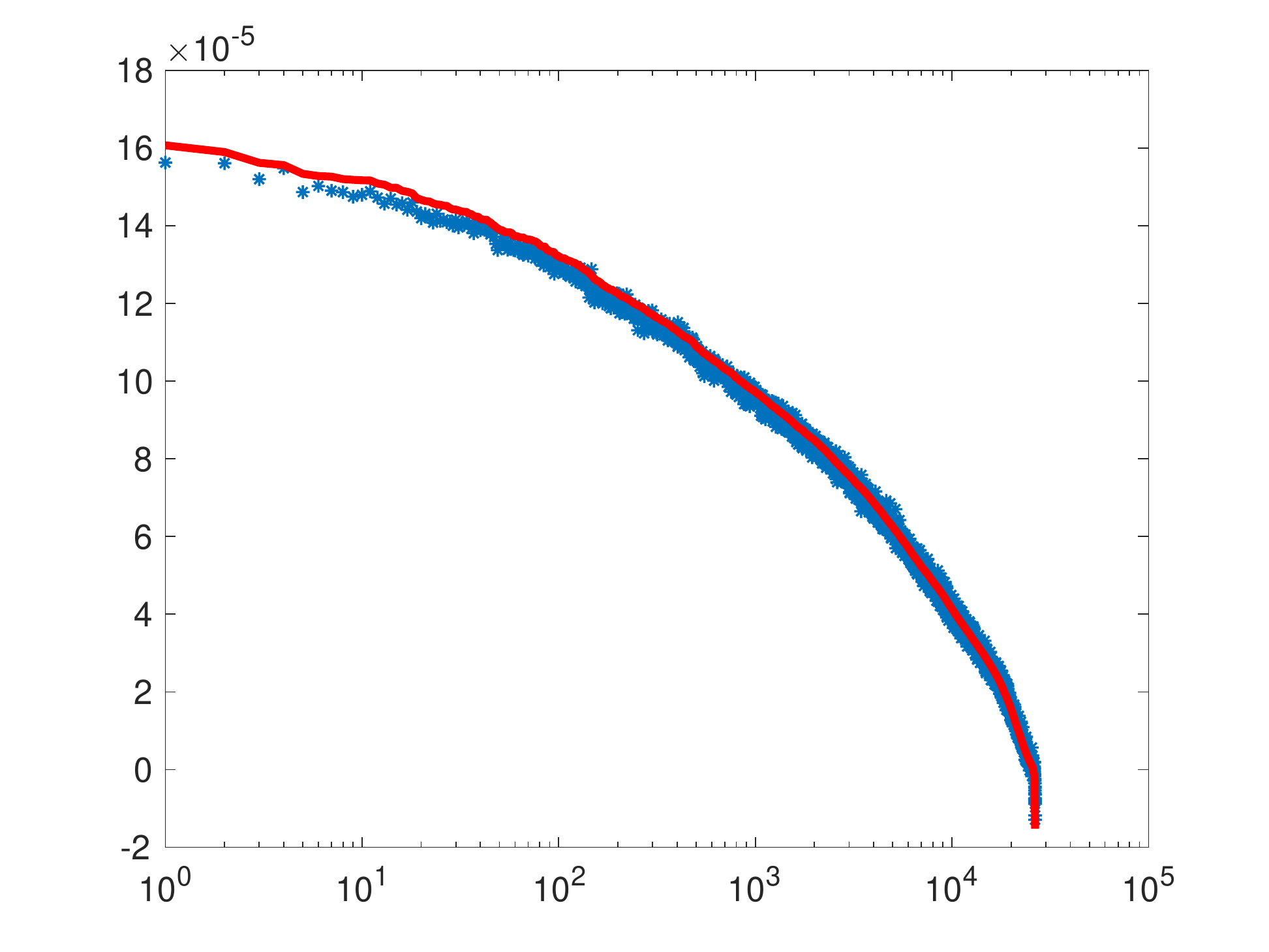}
        \caption{2016}
        \label{fig:2016IVSsorted}
    \end{subfigure}
\begin{subfigure}[b]{0.3\textwidth}
        \includegraphics[width=\textwidth]{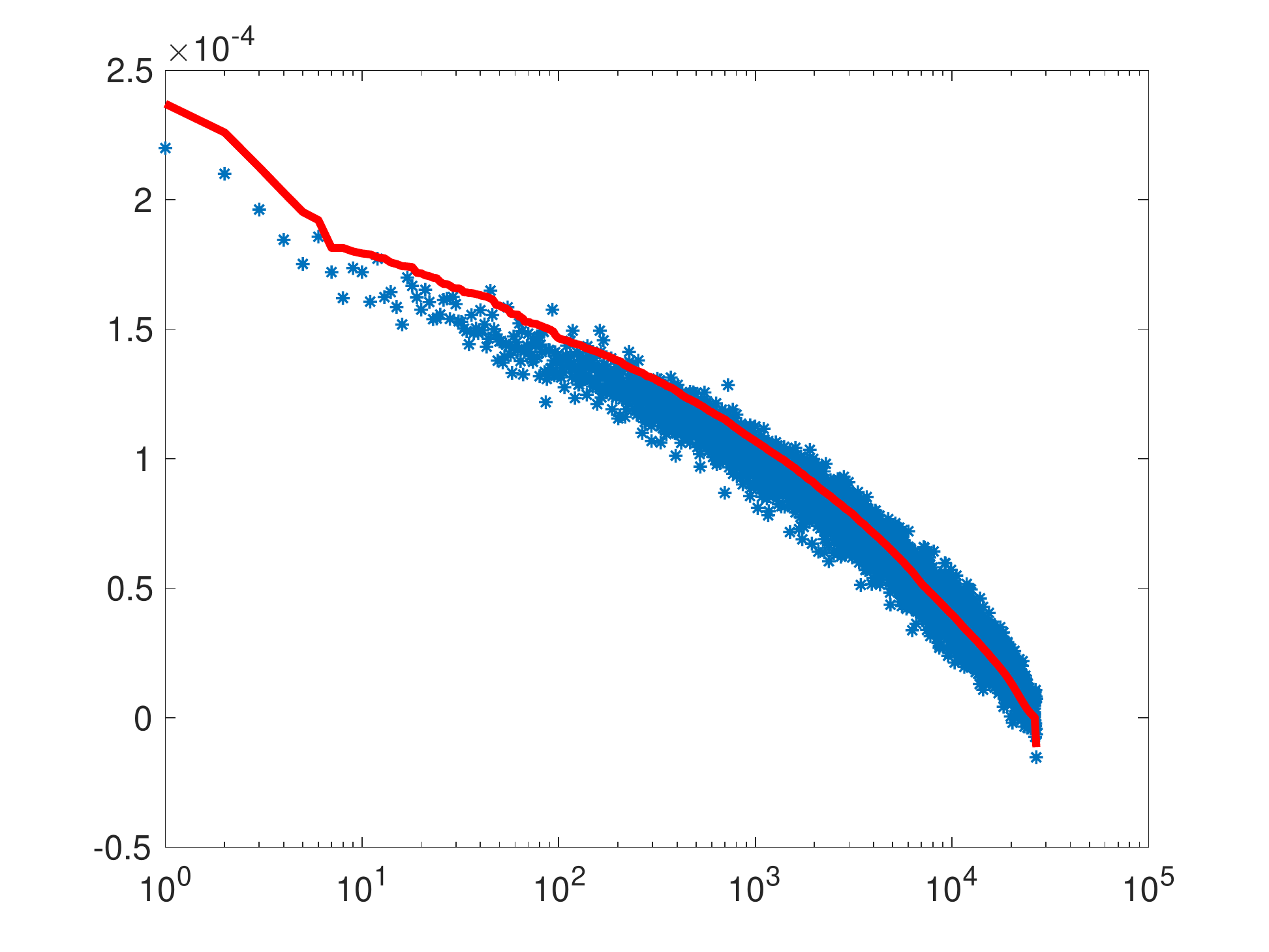}
        \caption{2017}
        \label{fig:2017IVSsorted}
    \end{subfigure}
     \caption{The eigenportfolio computed from the data (red) compared with the theoretical eigenportfolio that is close to $\beta_i/h_i^2$ (blue). The horizontal axes are in log scale. Each $\beta_i = \hbox{cov}(dE_i(\cdot)/E_i(\cdot),dQ(\cdot)/Q(\cdot))/h_q^2$, where for these plots the factor $Q(t)$ has been computed according to the formula in \eqref{eq:OIfactor} using the weighting function $\omega(\mathcal{OI}) = \mathcal{OI}$. To generate this plot, we first sort the eigenportfolio in descending order, and then insert the sorting index into the vector $h^{-2}\beta/\sum_i(h^{-2}\beta)_i$. The lining up of the two vectors using a single sorting index is evidence that the factor computed using the weighting  function $\omega(\mathcal{OI}) = \mathcal{OI}$ is close to the data's principal component.
     \label{fig:sortedEigenportfolios}}
\end{figure}

We can also check the name, the tenor, the $\Delta$, the $\beta_i/h_i^2$'s and the OI for the top-weighted options. These traits are listed for the top 32 options in the 2017 sorting. It is interesting to note that most of the top options are out-of-the-money put options, and all with 365 days to maturity (the longest-dated options in the dataset). Generally speaking, long-dated options have higher Vega\footnote{The Black-Scholes call/put option Vega of is $\mathcal V(t) = SN'(d_1)\sqrt{\tau}$ where $\tau$ is the time-to-maturity. Hence, all other things being equal, longer-dated options have higher Vega.} and therefore are most sensitive to changes in implied volatility. The interpretation of long-dated options dominating the $1^{st}$ eigenportfolio is simple: the $1^{st}$ eigenportfolio explains the most systematic movements among the options and should be the least sensitive to idiosyncratic noise, and therefore ignores short-dated options that may fluctuate idiosyncratically due to short-lived risk events. The years 2013 to 2017 had similar characteristics in the top 32 eigenportfolio options.

\begin{table}[t!]
    \centering
    \begin{tabular}{c|c|c|c|c}
    ticker&maturity (days)&$\Delta$&$\beta/h^2$&OI (average)\\
    \hline
ADSK &365 &-20 &3.5332$\times  10^3$ & 2751 \\
KLAC &365 &-20 &3.4206$\times  10^3$ & 3962 \\
DHR &365 &-20 &3.3575$\times  10^3$ & 18017 \\
KLAC &365 &-30 &3.2402$\times  10^3$ & 3906 \\
LRCX &365 &-30 &3.2904$\times  10^3$ & 4063 \\
KLAC &365 &-40 &3.1996$\times  10^3$ & 5171 \\
LRCX &365 &-40 &3.2450$\times  10^3$ & 4002 \\
INTU &365 &-40 &3.1577$\times  10^3$ & 4229 \\
INTU &365 &-30 &3.1832$\times  10^3$ & 5457 \\
ADI &365 &-20 &3.1372$\times  10^3$ & 15907 \\
XLNX &365 &-20 &3.1402$\times  10^3$ & 4237 \\
FLR &365 &-20 &3.1365$\times  10^3$ & 19285 \\
ADI &365 &-30 &3.0942$\times  10^3$ & 13276 \\
ADSK &365 &-30 &3.1348$\times  10^3$ & 4018 \\
TXN &365 &-20 &3.1324$\times  10^3$ & 19669 \\
DHR &365 &-30 &3.0415$\times  10^3$ & 23395 \\
CHRW &365 &-40 &3.0531$\times  10^3$ & 2874 \\
CHRW &365 &-30 &3.0543$\times  10^3$ & 3044 \\
FLR &365 &-30 &3.0444$\times  10^3$ & 22425 \\
ROST &365 &-30 &3.1108$\times  10^3$ & 5809 \\
RCL &365 &-40 &3.0327$\times  10^3$ & 7943 \\
FDX &365 &-30 &3.0153$\times  10^3$ & 34131 \\
RCL &365 &-20 &3.0019$\times  10^3$ & 24656 \\
EXPD &365 &-20 &2.9784$\times  10^3$ & 5796 \\
LRCX &365 &40 &3.0237$\times  10^3$ & 14541 \\
FDX &365 &-20 &2.9778$\times  10^3$ & 43925 \\
FLR &365 &40 &2.9695$\times  10^3$ & 18146 \\
LRCX &365 &30 &3.0293$\times  10^3$ & 12629 \\
ADSK &365 &-40 &2.9768$\times  10^3$ & 3314 \\
SBUX &365 &-40 &3.0220$\times  10^3$ & 39345 \\
IR &365 &-20 &2.9369$\times  10^3$ & 4699 \\
CF &365 &-20 &2.9476$\times  10^3$ & 72143 
\end{tabular}
    \caption{The ordering of names in the top 32 slots in the eigenportfolio of S\&P500 constituents' implied volatility for the year 2012. The maturity for all these top-weighted options is 365 days, which are the longest-dated options in our dataset. Generally speaking, long-dated options have higher Vega, which means that the $1^{st}$ eigenportfolio finds the most systematic explanation of implied volatility surface movements by considering the options with the greatest sensitivity to implied volatility changes.}
    \label{tab:impliedVolPCA207names}
\end{table}

We also perform checks to make sure that eigenportfolios can track in an online setting, i.e., all quantities used to compute a portfolio in real-time are adapted to the filtration generated by the options data.\footnote{The online eigenportfolio that we compute has one major anticipatory element, which is survivorship bias. We have selected names of the S\&P500 constituents from 2017 and collect their options' data going back to 2012. Hence, there is some survivorship bias in favor of names that perform well enough to stay in the index for these 6 years. We acknowledge this bias and realize that it will persist over time, but for our analysis of daily statistics it does not make our findings any less valid. Moreover, the OI weighing favors liquid options for which survivorship bias is likely to be reduced.} We group all 6 years of data into one simulated run, and then compute an eigenportfolio each month using the previous six months' daily returns. The returns are compounded daily but we only update the portfolio weights monthly. Figures \ref{fig:eigenAndOI_scatter}-\ref{fig:portfolioTrackingLOGweight} show the online 6-month sliding window performance of the eigenportfolio alongside the factor returns and the VIX returns. Notice in Figure \ref{fig:eigenAndVIX_scatter} that the eigenportfolio returns have significant explanatory power for the returns on VIX (i.e., $d \hbox{VIX}(t)/\hbox{VIX}(t)$), but also that returns for the OI factor and eigenportfolio have a much stronger linear dependence. Indeed, even if we use the weighting function $\omega(\mathcal{OI}) = \log(1+\mathcal{OI})\times\mathcal{V}^{untls}$, which allows the OI factor to have better tracking with the VIX (see Figure \ref{fig:portfolioTrackingLOGweight}), the stronger linear dependence in daily returns between factor and eigenportfolio still prevails. 
\begin{figure}[t]
\centering
    \begin{subfigure}[t]{0.45\textwidth}
    \captionsetup{width=.9\linewidth}
    \includegraphics[width=\textwidth]{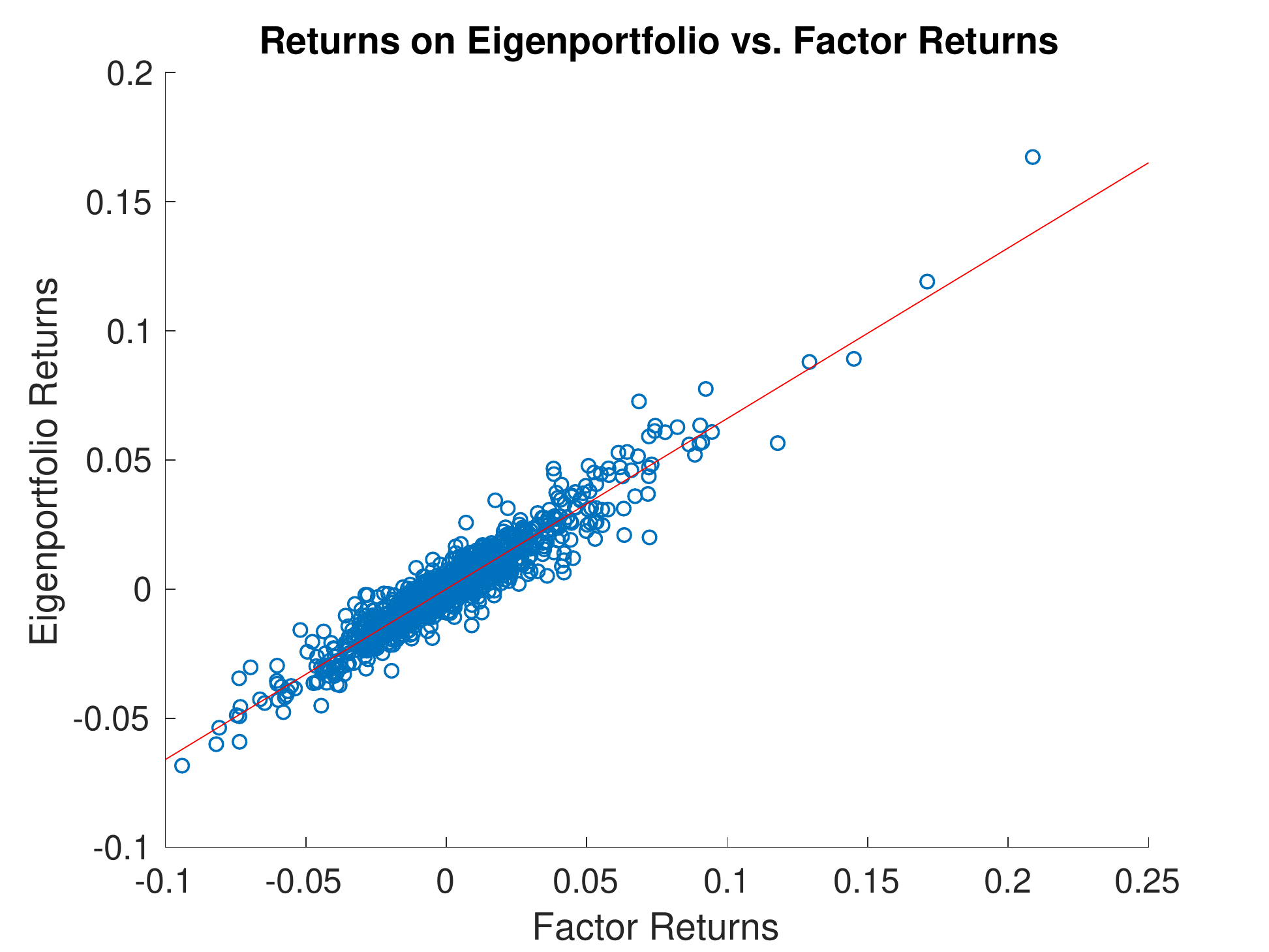}
    \caption{A scatter plot of the adapted eigenportfolio returns against the returns of OI-weighted factor computed using weighting function $\omega(\mathcal{OI})=\mathcal{OI}$, from 2012 through 2017. The eigenportfolio is computed using a 6-month sliding window, which is why the first 6 months of 2012 are not included in the output.}
    \label{fig:eigenAndOI_scatter}
    \end{subfigure}
    \begin{subfigure}[t]{0.45\textwidth}
    \captionsetup{width=.9\linewidth}
    \includegraphics[width=\textwidth]{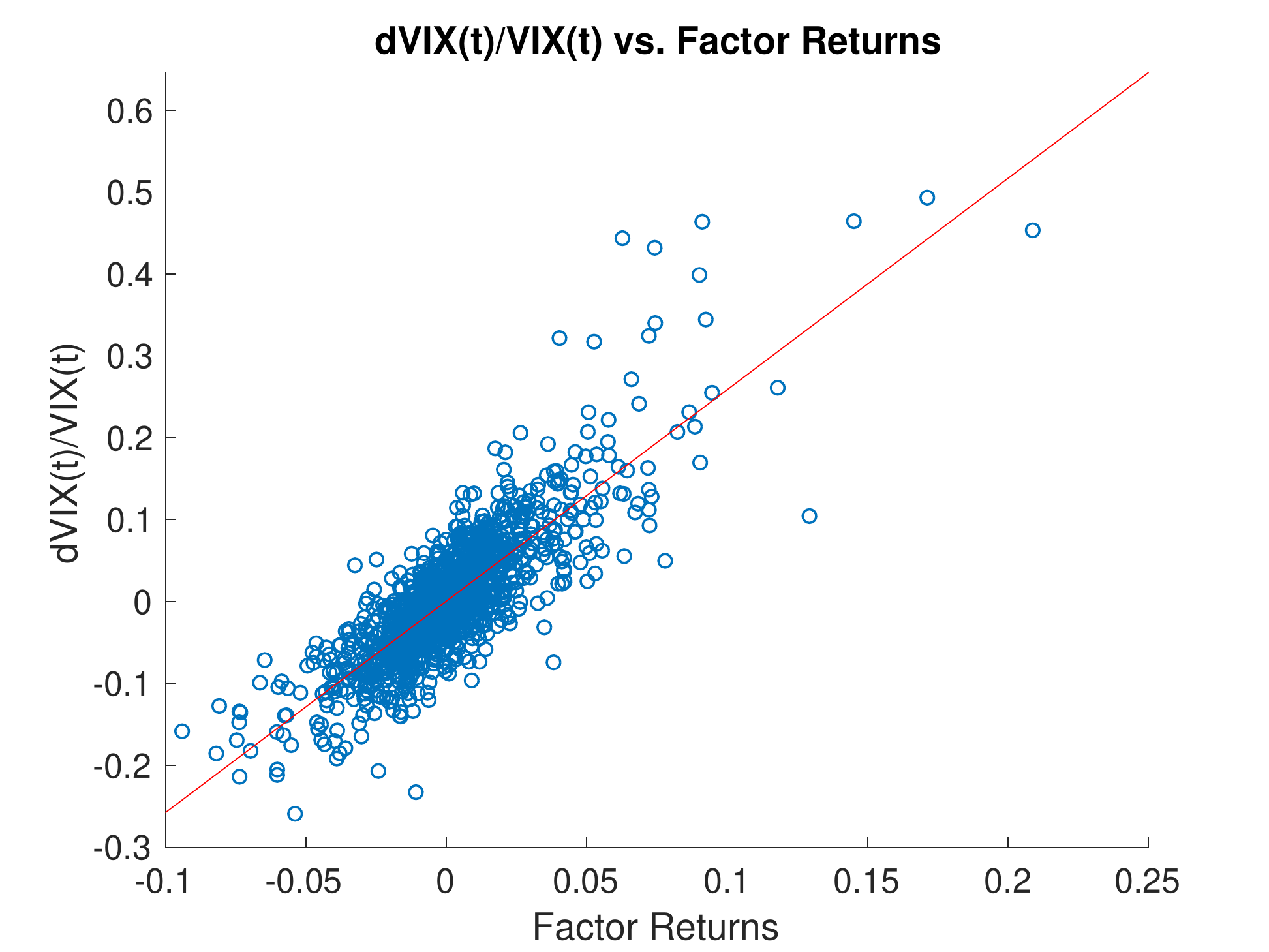}
        \caption{A scatter plot of the VIX returns against the returns of OI-weighted factor computed using weighting function $\omega(\mathcal{OI})=\mathcal{OI}$, from 2012 through 2017.}
        \label{fig:eigenAndVIX_scatter}
    \end{subfigure}\\
    \begin{subfigure}[t]{0.45\textwidth}
    \captionsetup{width=.9\linewidth}
            \includegraphics[width=\textwidth]{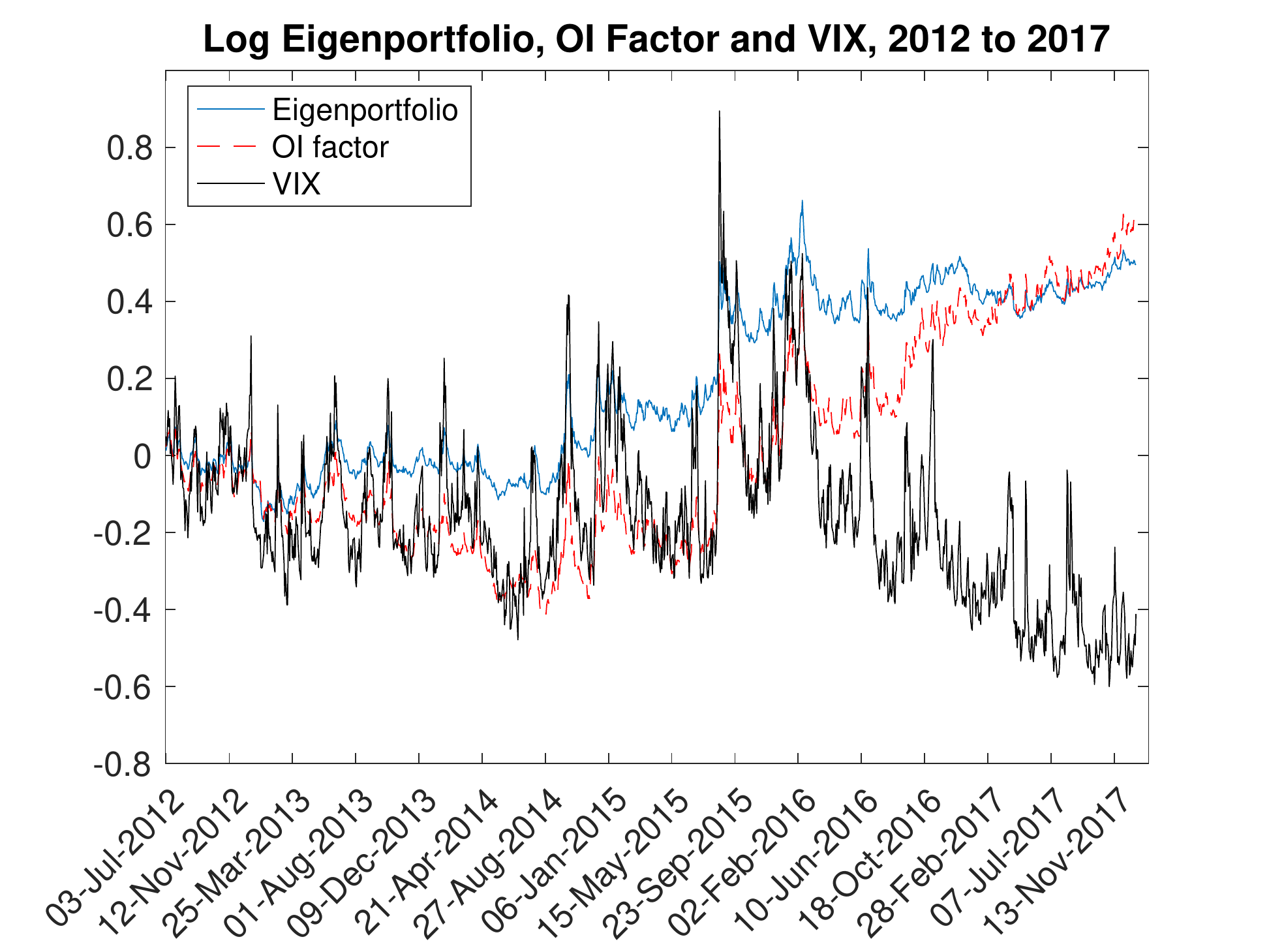} \caption{The adapted 6-month sliding-window eigenportfolio and the OI-weighted factor computed using weighting function $\omega(\mathcal{OI})=\mathcal{OI}$. }
            \label{fig:portfolioTracking}
    \end{subfigure}
    \begin{subfigure}[t]{0.45\textwidth}
    \captionsetup{width=.9\linewidth}
       \includegraphics[width=\textwidth]{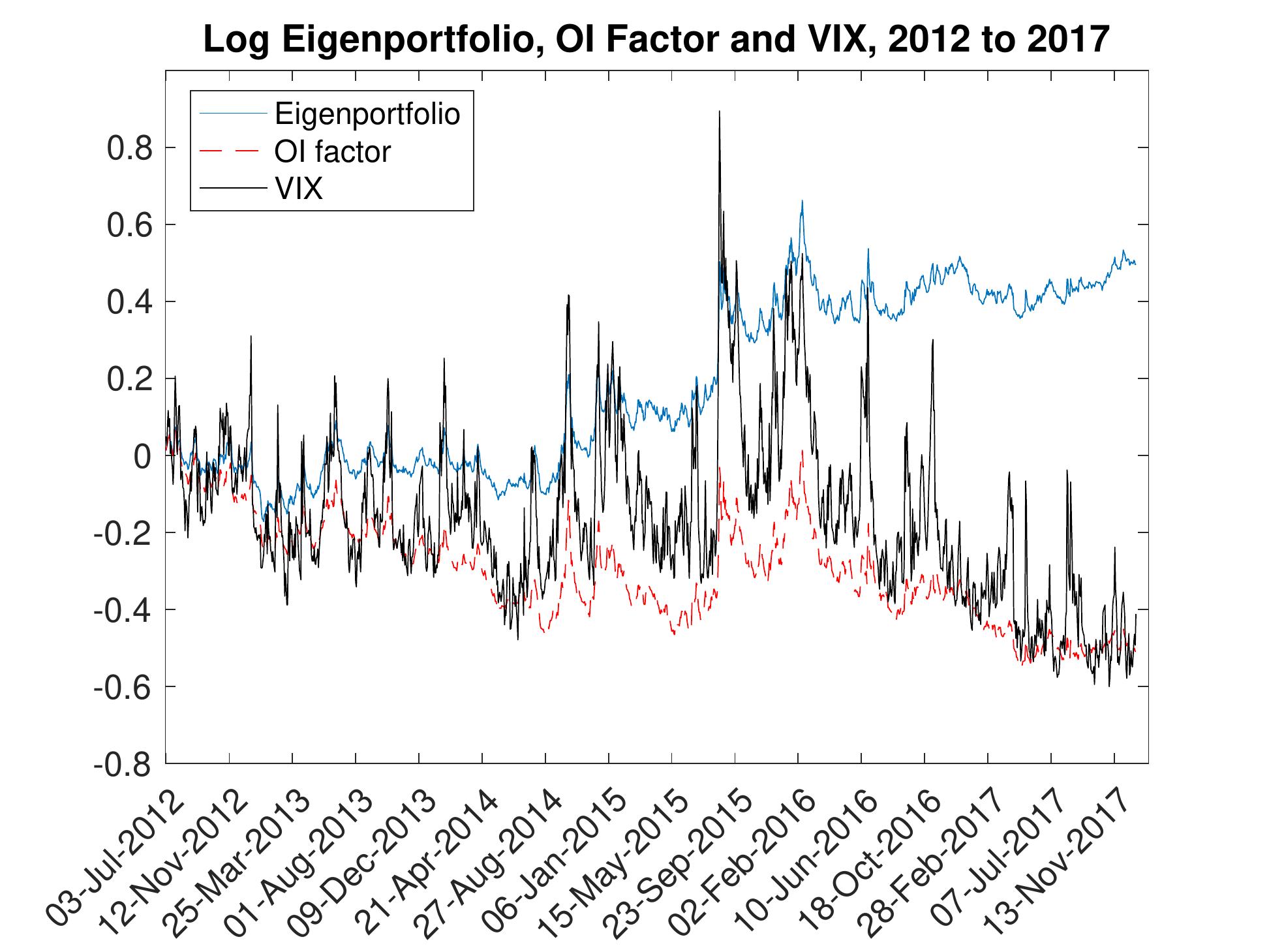} 
   \caption{The adapted 6-month sliding-window eigenportfolio and the OI-weighted index computed using weighting function $\omega(\mathcal{OI})=\log(1+\mathcal{OI})\times\mathcal{V}^{untls}$ (logarithm of each index). By taking the OI weights to be the logarithm of $1$ plus OI, we see improved tracking of the VIX compared to Figure \ref{fig:portfolioTracking}. Regression of the eigenportfolio onto the factor has significant intercept coefficient but improved projection error compared to the same regression with $\omega(\mathcal{OI})=\mathcal{OI}$ (see Table \ref{tab:alphas}). }
   \label{fig:portfolioTrackingLOGweight}     
    \end{subfigure}
    \caption{Output from adapted eigenportfolios using a 6-month sliding window.\label{fig:slidingWindowAnalysis}}
\end{figure}

Lastly, some discussion on performance with the two different weighting functions is in order. The visual evidence in Figures \ref{fig:sortedEigenportfolios} and \ref{fig:slidingWindowAnalysis} should be sufficiently convincing for the reader that OI needs to be included in factor construction. The discussion that remains is for us to decide how to determine the OI-weighted factor that is somehow 'best'. Table \ref{tab:alphas} provides evidence indicating that $\omega(\mathcal{OI}) = \mathcal{OI}$ is good because it leaves the least amount of unexplained systematic return after regression of the eigenportfolio onto the factor's returns (as well as two other factors). However, visually from Figure \ref{fig:portfolioTrackingLOGweight} we see that $\omega(\mathcal{OI}) = \log(1+\mathcal{OI})\times \mathcal{V}^{untls}$ produces a factor that is better for tracking the VIX, and the VIX is the U.S. market's premier volatility index, but this choice of $\omega$ leaves a significant amount of unexplained systematic return after the regression. Bear in mind that leaving unexplained systematic return is not entirely bad because we have thus far only been able to construct a single factor for explained implied volatility movements, but we know that we need at least two (recall Figure \ref{fig:singularValues} where it was clear that there are at least two principal components). 

It turns out that there is a setting wherein $\omega(\mathcal{OI}) = \log(1+\mathcal{OI})\times \mathcal{V}^{untls}$ is an acceptable factor, but it will require us to move from ordinary 'flat' linear algebra and use tensors in the next Section. 
\begin{table}[hbp]
    \centering
    \small
    \begin{tabular}{l|r|r|r|r|r|r}
    \multicolumn{1}{c|}{$\omega(\mathcal{OI})$}&
    \multicolumn{1}{c|}{$\alpha_3$}&
    \multicolumn{1}{c|}{$\beta_3$}&
    \multicolumn{1}{c|}{$\alpha_2$}&
    \multicolumn{1}{c|}{$\beta_2$}&
    \multicolumn{1}{c|}{$\alpha_1$}&
    \multicolumn{1}{c}{$\beta_1$}\\
    \hline
    \hline
    \multicolumn{7}{c}{2012-2014}\\
    \hline
        $\mathcal{OI}$&0.0629& 0.6348& 0.0755& 0.7196& 0.0908& 0.7187\\
&(1.6307)&& (1.8777)&& (2.2614) &\\
&\multicolumn{2}{c|}{$R^2=$ 0.9387}&\multicolumn{2}{c|}{$R^2=$ 0.9332}&\multicolumn{2}{c}{ $R^2=$ 0.9322}\\
                \cline{2-7}
    $\log(1+\mathcal{OI})\times\mathcal{V}^{untls}$&0.1224& 0.8947& 0.1418& 1.0053& 0.1745& 1.0000\\
&(4.0669)&& (4.4146)&& (5.2169) &\\
&\multicolumn{2}{c|}{$R^2=$ 0.9629}&\multicolumn{2}{c|}{$R^2=$ 0.9574}&\multicolumn{2}{c}{ $R^2=$ 0.9529}\\

    \hline
    \hline
    \multicolumn{7}{c}{2013-2015}\\
    \hline
    $\mathcal{OI}$&0.0570& 0.5500& 0.0728& 0.6794& 0.0865& 0.6769\\
&(1.1356)&& (1.3413)&& (1.5885) &\\
&\multicolumn{2}{c|}{$R^2=$ 0.9303}&\multicolumn{2}{c|}{$R^2=$ 0.9184}&\multicolumn{2}{c}{ $R^2=$ 0.9172}\\
\cline{2-7}
    $\log(1+\mathcal{OI})\times\mathcal{V}^{untls}$&0.1409& 0.8472& 0.1664& 0.9842& 0.1934& 0.9745\\
&(4.1438)&& (4.4377)&& (4.8793) &\\
&\multicolumn{2}{c|}{$R^2=$ 0.9681}&\multicolumn{2}{c|}{$R^2=$ 0.9610}&\multicolumn{2}{c}{ $R^2=$ 0.9561}\\
    \hline
    \hline
    \multicolumn{7}{c}{2014-2016}\\
    \hline
    $\mathcal{OI}$&0.0190& 0.5259& 0.0116& 0.6600& 0.0165& 0.6594\\
&(0.3091)&& (0.1746)&& (0.2491) &\\
&\multicolumn{2}{c|}{$R^2=$ 0.9183}&\multicolumn{2}{c|}{$R^2=$ 0.9049}&\multicolumn{2}{c}{ $R^2=$ 0.9047}\\
                \cline{2-7}
    $\log(1+\mathcal{OI})\times\mathcal{V}^{untls}$&0.1841& 0.8731& 0.2007& 0.9734& 0.2184& 0.9687\\
&(4.9527)&& (5.1087)&& (5.3412) &\\
&\multicolumn{2}{c|}{$R^2=$ 0.9701}&\multicolumn{2}{c|}{$R^2=$ 0.9664}&\multicolumn{2}{c}{ $R^2=$ 0.9634}\\
    \hline
    \hline
    \multicolumn{7}{c}{2015-2017}\\
    \hline
    $\mathcal{OI}$&-0.0505& 0.5164& -0.0705& 0.6168& -0.0803& 0.6170\\
&(-0.8877)&& (-1.1757)&& (-1.3413) &\\
&\multicolumn{2}{c|}{$R^2=$ 0.9102}&\multicolumn{2}{c|}{$R^2=$ 0.8997}&\multicolumn{2}{c}{ $R^2=$ 0.8992}\\
                \cline{2-7}
    $\log(1+\mathcal{OI})\times\mathcal{V}^{untls}$&0.1705& 0.8926& 0.1769& 0.9389& 0.1836& 0.9381\\
&(6.5261)&& (6.6106)&& (6.8505) &\\
&\multicolumn{2}{c|}{$R^2=$ 0.9810}&\multicolumn{2}{c|}{$R^2=$ 0.9800}&\multicolumn{2}{c}{ $R^2=$ 0.9797}\\
       \hline
    \hline
    \multicolumn{7}{c}{2012-2017 \textbf{(all years)}}\\
    \hline
    $\mathcal{OI}$&-0.0010& 0.5444& -0.0020& 0.6604& 0.0028& 0.6601\\
&(-0.0289)&& (-0.0546)&& (0.0781) &\\
&\multicolumn{2}{c|}{$R^2=$ 0.9156}&\multicolumn{2}{c|}{$R^2=$ 0.9033}&\multicolumn{2}{c}{ $R^2=$ 0.9032}\\
                \cline{2-7}
    $\log(1+\mathcal{OI})\times\mathcal{V}^{untls}$&0.1476& 0.8903& 0.1619& 0.9760& 0.1813& 0.9732\\
&(7.2613)&& (7.5857)&& (8.2952) &\\
&\multicolumn{2}{c|}{$R^2=$ 0.9697}&\multicolumn{2}{c|}{$R^2=$ 0.9665}&\multicolumn{2}{c}{ $R^2=$ 0.9645}\\
    \end{tabular}
    \caption{Factor loadings and unexplained systematic returns for the eigenportfolio from a 6-month sliding window and various OI-weightings in the factor construction. For each sample period there are three regressions: a 3-factor regression $\frac{dE^{ep}}{E^{ep}} = \alpha_3+\beta_3 \frac{dQ}{Q}+b_{eq}\frac{d\hbox{\scriptsize S\&P}}{\hbox{\scriptsize S\&P}}+b_{vx}\frac{d \hbox{\scriptsize VIX}}{\hbox{\scriptsize VIX}}+\varepsilon$, a 2-factor regression $\frac{dE^{ep}}{E^{ep}} = \alpha_2+\beta_2\frac{dQ}{Q}+b_{eq}\frac{d\hbox{\scriptsize S\&P}}{\hbox{\scriptsize S\&P}}+\varepsilon$, and a 1-factor regression $\frac{dE^{ep}}{E^{ep}} = \alpha_1+\beta_1 \frac{dQ}{Q}+\varepsilon$. The factor obtained by weighting function $\omega(\mathcal{OI}) = \mathcal{OI}$ is perhaps desireable because it leaves no significant excess return in the residual, but there is no reason why we should reject a factor with non-zero intercept. In contrast, the weighting function $\omega(\mathcal{OI})=\log(1+\mathcal{OI})\times\mathcal{V}^{untls}$ has significant intercept and a higher $R^2$, and hence is perhaps a better factor.}
    \label{tab:alphas}
\end{table}

\section{Factors and Eigenportfolios Using Tensors}
\label{sec:tensors}

The implied volatility surfaces have a natural tensor structure with 4 dimensions: time, name, maturity and delta (normalized strike). To work with this tensor we first need to redefine our notation for implied volatility returns from how they were defined in \eqref{eq:returns} and \eqref{eq:standardizedReturns}. Before we had $1\leq i\leq N$ where $N = 500\cdot 8\cdot 7=28,000$ (500 names, 8 maturities, 7 deltas). Now we have multiple indices,
\[r_{ijk}(t) = \hbox{time-$t$ implied vol. return for $i^{th}$ name, $j^{th}$ maturity and $k^{th}$ $\Delta$}\ ,\]
where now $1\leq i\leq N^{(1)}$, $1\leq j\leq N^{(2)}$, $1\leq k\leq N^{(3)}$ and $1\leq t\leq T$ (i.e., $N^{(1)}=500$, $N^{(2)}=8$ and $N^{(3)}=7$). We standardize the returns and place them in a 4-dimensional tensor,
\begin{equation}
    \label{eq:tensorReturns}
    R = \left[\frac{r_{ijk}(t) - \overline{r}_{ijk}}{h_{ijk}}\right]_{1\leq t\leq T,~1\leq i\leq N^{(1)},~1\leq j\leq N^{(2)},~1\leq k\leq N^{(3)}}\ ,
\end{equation}
where $\overline{r}_{ijk} = \frac{1}{T}\sum_tr_{ijk}(t) $ and $h_{ijk} = \sqrt{\frac{1}{T-1}\sum_t(r_{ijk}(t) -\overline{r}_{ijk})^2}$. Similar to the sector-based hierarchical PCA done for equities in \cite{avellaneda2019hierarchical}, we can define individualized factors for each value of a certain tensor dimension. The maturity dimension and the delta dimension are the two candidates for individualized factors; the following subsection demonstrates the advantage of constructing the individualized factors.

\subsection{Eigenportfolios Via Multilinear SVD (MLSVD)}

A tensor analogue for the SVD is the multilinear singular value decomposition (MLSVD) (see \cite{cichocki2015tensor,kolda2009tensor,deLathauwer2000multilinear}). The canonical polyadic decomposition (CPD) (see \cite{kolda2009tensor,tucker1966some}) can also be used, but when it comes to computing principal components with the IVS data, the results produced by these two approaches are very similar. We use the MLSVD, which is briefly described in the Appendix, as follows. We can write the tensor $R$ in the form
\begin{equation}
    \label{eq:MLSVD}
    R = \sum_{t,i,j,k}S_{tijk}~ U_t^{(1)}\circ U_i^{(2)}\circ U_j^{(3)}\circ U_k^{(4)}\ ,
\end{equation}
where $U^{(1)}$ is a $T\times T$ orthonormal matrix, $U^{(2)}$ is an $N^{(1)}\times T$ orthonormal matrix, $U^{(3)}$ is an $N^{(2)}\times N^{(2)}$ orthonormal matrix, $U^{(4)}$ is a $N^{(3)}\times N^{(3)}$ orthonormal matrix, $S$ is a $T\times N^{(1)}\times N^{(2)}\times N^{(3)}$ real-valued tensor, and where $\circ$ denotes the vector outerproduct. The main disadvantage of the MLSVD is that the so-called \textit{core tensor} $S$ is, in general, not diagonal or even sparse, which results in difficulties when computing a rank-$d$ decomposition that is best in the sense of the Frobenius norm. In practice, the rank-1 decomposition used for computing the eigenportfolio can be estimated using the top MLSVD vectors, which we write as
\begin{align}
\label{eq:tensorPrinc}
\widetilde U^{(1)} = U_1^{(2)}\circ U_1^{(3)}\circ U_1^{(4)}\ ,
\end{align}
which for the IVS data is very close to what one gets using CPD. To compute the tensor eigenportfolio, we do an elementwise division with the tensor of standard deviations,
\begin{equation}
    \label{eq:tensorEigenportfolio}
    \pi_{ijk}^{(1)}\propto \frac{\widetilde U_{ijk}^{(1)}}{h_{ijk}}\ .
\end{equation}
It remains to decide how to normalize this $\pi^{(1)}$. If we normalize with one global summation over the multi-index $(i,j,k)$, then there would have been no need to use tensors. However, if there is a dimension such that for each of the index values there corresponds a different OI-weighted factor, then the normalization needs to be done separately for the different values of this index. The logic for using different normalizations for different index values will become clearer in the following example.

\subsection{Example: Different Factors for Each Maturity}

Suppose that we construct a different OI-weighted factor for each maturity, for a total of 8 factors. Moreover, suppose that our factor considers the $\beta$ for each option to be its loading on the factor corresponding to the same maturity, i.e., there are factors $Q_j$ for $j=1,2,\dots 8$, with $dQ_j(t)/Q_j(t) \propto \sum_{i,k}\omega(\mathcal{OI}_{ijk}(t))r_{ijk}(t)$, 
where $\mathcal{OI}$ is the tensor of open interests. The factor model for returns is then
\[r_{ijk}(t) = \beta_{ijk}\frac{dQ_j(t)}{Q_j(t)} +\xi_{ijk}(t)\ ,\]
where $\xi$ is a tensor of idiosyncratic noise that is independent of the factor processes in $Q_j$. In this case, for the the $\beta_{ijk}/h_{ijk}^2$'s to line up with the eigenportfolio in \eqref{eq:tensorEigenportfolio}, the normalization needs to be done as follows for each $j$,
\begin{align}
    \nonumber
    \pi_{ijk}^{(1)} &= \frac{\widetilde U_{ijk}^{(1)}/h_{ijk}}{\sum_{i,k}\widetilde U_{ijk}^{(1)}/h_{ijk}}\\
    \label{eq:tensorNormalizations}
    \frac{dQ_j(t)}{Q_j(t)}&=\frac{\sum_{i,k}\omega(\mathcal{OI}_{ijk}(t))r_{ijk}(t)}{\sum_{i,k}\omega(\mathcal{OI}_{ijk}(t))}\ .
\end{align}
Finally, the global factor (to compare with the eigenportfolio) is simply the mean of the tenor factors, 
\begin{equation}
    \label{eq:tensorGlobalFactor}
    \frac{d\overline Q(t)}{\overline Q(t)} = \frac{1}{N^{(2)}}\sum_j\frac{dQ_j(t)}{Q_j(t)}\ .
\end{equation}
Figure \ref{fig:tensorPortfolioPerformance} shows improved results if this maturity-wise approach is used with the normalizations in \eqref{eq:tensorNormalizations} along with the weighting function $\omega(\mathcal{OI}) =\log(1+\mathcal{OI})\times \mathcal V^{untls}$. The linear dependence between the tensor eigenportfolio and the $\overline Q$ seen in Figure \ref{fig:eigenAndTensorOI_scatter} is a considerable improvement from the linear dependence shown in Figure \ref{fig:eigenAndOI_scatter} that was obtained using flat matrices. The ex-post regression diagnostics of the tensor approach show an improvement over the flat matrix approach, as Table \ref{tab:tensorAlphas} lists regression outputs that have higher $R^2$ (i.e., less projection error) than their counterparts in Table \ref{tab:alphas}.
\begin{figure}[t!]
    \centering
    \begin{subfigure}[b]{0.45\textwidth}
    \captionsetup{width=.9\linewidth}
    \includegraphics[width=\textwidth]{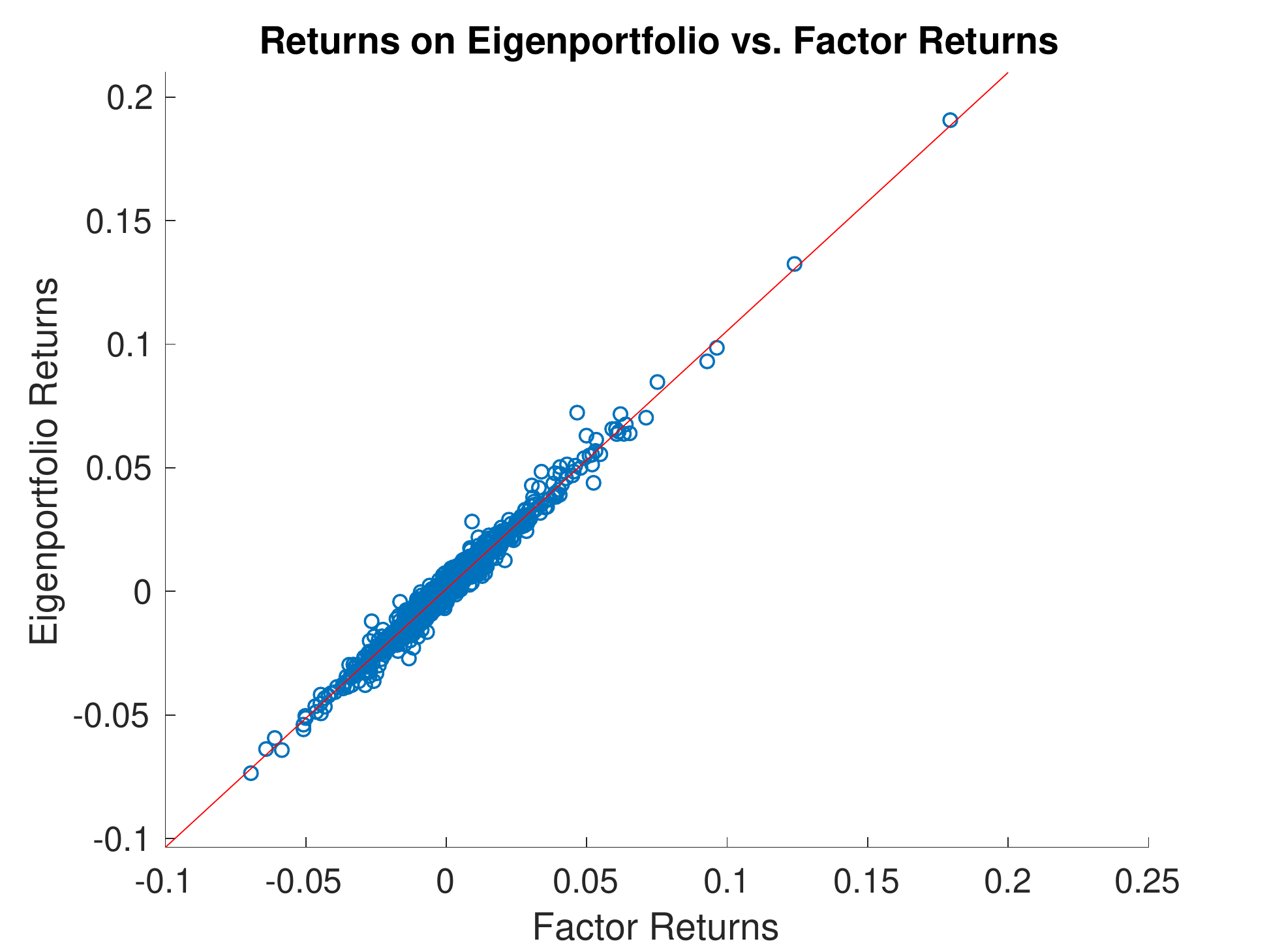}
    \caption{A scatter plot of the adapted tensor eigenportfolio returns against the  returns of tensor OI-weighted factor $\overline Q$ given by \eqref{eq:tensorGlobalFactor}, computed using weighting function $\omega(\mathcal{OI}) = \log(1+\mathcal{OI})\times\mathcal{V}^{untls}$, from 2012 through 2017. The eigenportfolio is computed using a 6-month sliding window. The linear dependence seen here is an improvement from the linear dependence seen in Figure \ref{fig:eigenAndOI_scatter} achieved using flat matrices.}
    \label{fig:eigenAndTensorOI_scatter}
    \end{subfigure}
    \begin{subfigure}[b]{0.45\textwidth}
    \captionsetup{width=.9\linewidth}
    \includegraphics[width=\textwidth]{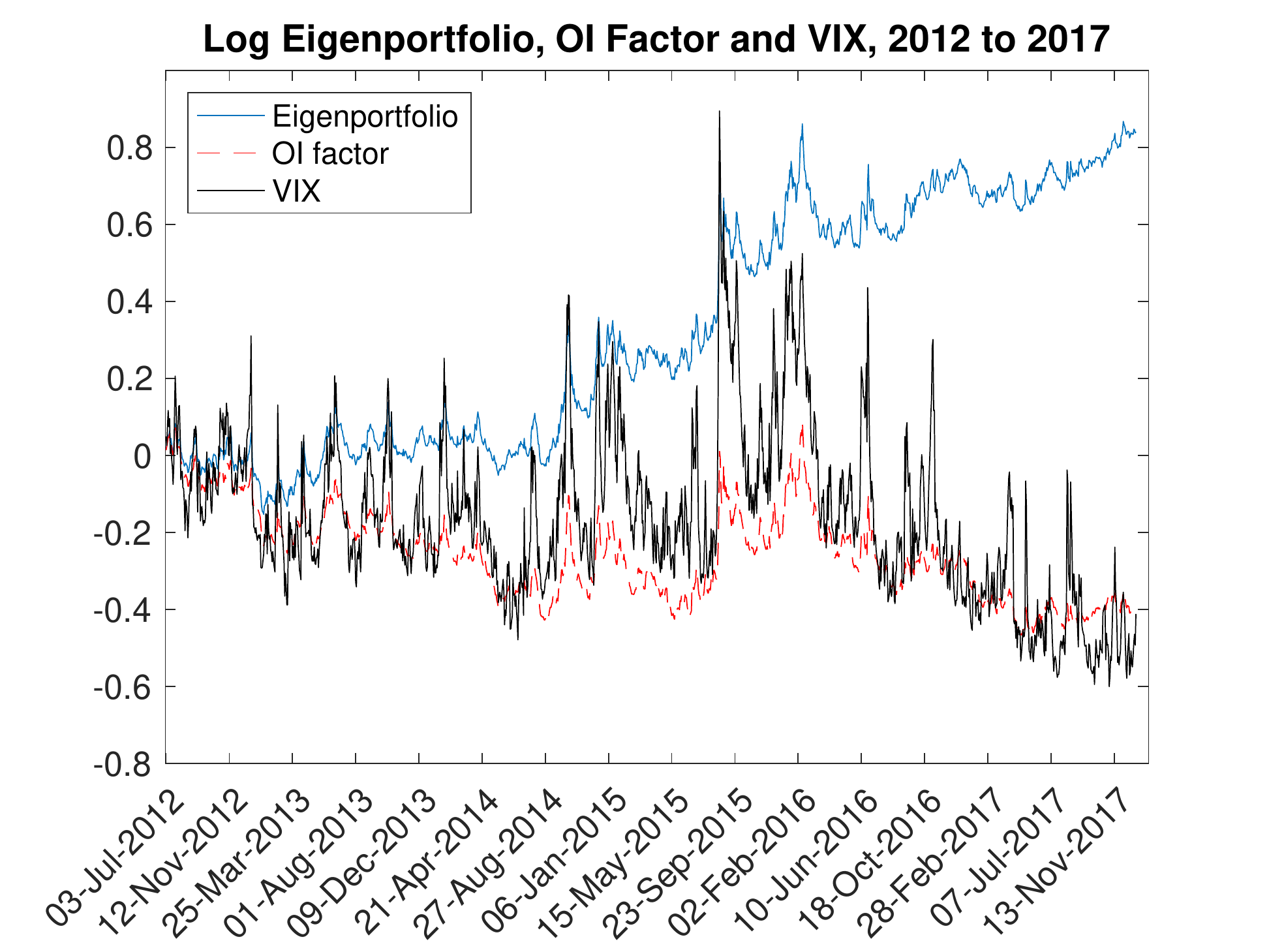}
   \caption{The adapted tensor eigenportfolio and the tensor OI-weighted index computed using weighting function $\omega(\mathcal{OI}) = \log(1+\mathcal{OI})\times\mathcal{V}^{untls}$ (logarithm of each index). This picture does not suggest that the tensor approach to eigenportfolio and factor construction offers an improvement to simply using flat matrices, as it is roughly similar to its analogue in Figure \ref{fig:portfolioTrackingLOGweight}. Improvement due to tensors is more evident in Figures \ref{fig:eigenAndTensorOI_scatter}, \ref{fig:PE+CUMRET_2012to2014} and \ref{fig:PE+CUMRET_2015to2017}, and in Table \ref{tab:tensorAlphas}.}
   \label{fig:tensorPortfolioTracking}
   \end{subfigure}
   \caption{Output from the adapted tensor eigenportfolio using a 6-month sliding window.\label{fig:tensorPortfolioPerformance}}
\end{figure}
\begin{table}[hbp]
    \centering
    \small
    \begin{tabular}{l|r|r|r|r|r|r}
    \multicolumn{1}{c|}{$\omega(\mathcal{OI})$}&
    \multicolumn{1}{c|}{$\alpha_3$}&
    \multicolumn{1}{c|}{$\beta_3$}&
    \multicolumn{1}{c|}{$\alpha_2$}&
    \multicolumn{1}{c|}{$\beta_2$}&
    \multicolumn{1}{c|}{$\alpha_1$}&
    \multicolumn{1}{c}{$\beta_1$}\\
    \hline
    \hline
    \multicolumn{7}{c}{2012-2017 \textbf{(all years)}}\\
    \hline
    $\mathcal{OI}$&0.0891& 0.7384& 0.0923& 0.8432& 0.0925& 0.8432\\
&(3.0439)&& (2.9875)&& (3.0107) &\\
&\multicolumn{2}{c|}{$R^2=$ 0.9466}&\multicolumn{2}{c|}{$R^2=$ 0.9404}&\multicolumn{2}{c}{ $R^2=$ 0.9404}\\
                \cline{2-7}
    $\log(1+\mathcal{OI})\times\mathcal{V}^{untls}$&0.2092& 0.9673& 0.2211& 1.0466& 0.2339& 1.0450\\
&(11.8532)&& (11.8433)&& (12.3862) &\\
&\multicolumn{2}{c|}{$R^2=$ 0.9806}&\multicolumn{2}{c|}{$R^2=$ 0.9782}&\multicolumn{2}{c}{ $R^2=$ 0.9775}
    \end{tabular}
    \caption{Factor loadings and systematic unexplained returns for the tensor eigenportfolio from a 6-month sliding window and various OI-weightings in the tensor-factor construction. For each sample period there are three regressions: a 3-factor regression $\frac{dE^{ep}}{E^{ep}} = \alpha_3+\beta_3 \frac{d\overline{Q}}{\overline{Q}}+b_{eq}\frac{d\hbox{\scriptsize S\&P}}{\hbox{\scriptsize S\&P}}+b_{vx}\frac{d \hbox{\scriptsize VIX}}{\hbox{\scriptsize VIX}}+\varepsilon$, a 2-factor regression $\frac{dE^{ep}}{E^{ep}} = \alpha_2+\beta_2\frac{d\overline{Q}}{\overline{Q}}+b_{eq}\frac{d\hbox{\scriptsize S\&P}}{\hbox{\scriptsize S\&P}}+\varepsilon$, and a 1-factor regression $\frac{dE^{ep}}{E^{ep}} = \alpha_1+\beta_1 \frac{d\overline{Q}}{\overline{Q}}+\varepsilon$. The tensor regressions shown in this Table have higher $R^2$ than their counterparts in Table \ref{tab:alphas}, which is an indication that the tensor approach to factor construction is better.}
    \label{tab:tensorAlphas}
\end{table}

Finally, we show the in-sample plots of cumulative returns for the tensor portfolios next to the cumulative returns of the portfolios from the flat matrices; all portfolios are constructed with weighting function $\omega(\mathcal{OI}) = \log(1+\mathcal{OI})\times\mathcal V^{untls}$. Figures \ref{fig:PE+CUMRET_2012to2014} and \ref{fig:PE+CUMRET_2015to2017} show these results, from which it is clear that the tensor factor constructed in this subsection allows for eigenportfolio tracking that is much closer to the factor. This is evidence that the family of OI-based tensor factors subsumes a 'good' factor, in the sense that it can track the eigenportfolio returns. We recall that the reason why we are concerned with the eigenportfolio when constructing factors, is because we know from theoretical analysis of the spike model that the $1^{st}$ eigenportfolio will have weights close to the vector $h^{-2}\beta$, where $\beta$ are loadings on a dominant factor. Hence, to determine if we have a dominant factor, we should make comparisons with the eigenportfolio and in the tensor IVS data context the results come out better, most likely because the data is heterogeneous and it is therefore beneficial to respect the tensor structure, which MLSVD does. 

\begin{figure}
    \centering
    \begin{subfigure}[b]{0.45\textwidth}
        \includegraphics[width=\textwidth]{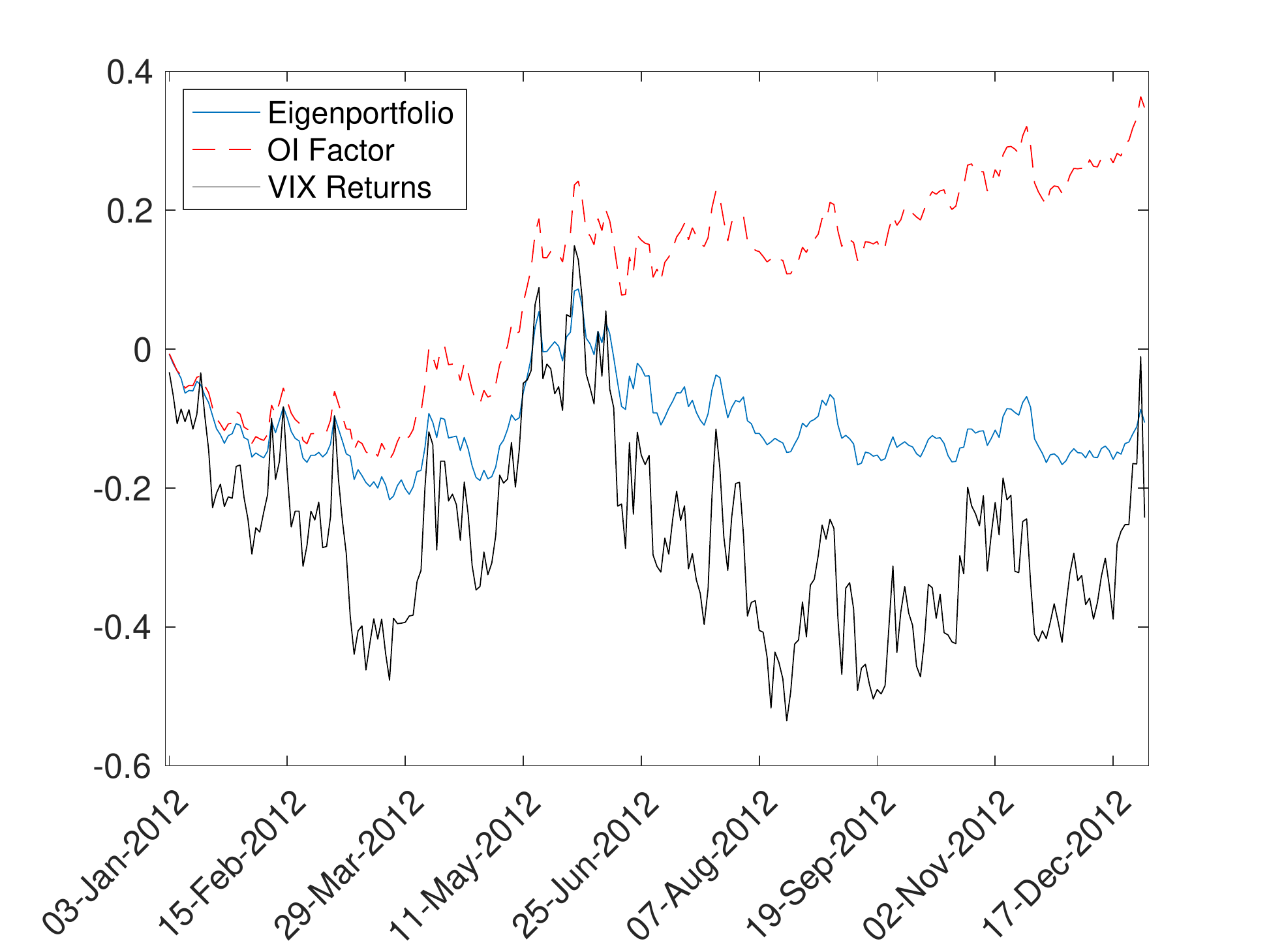}
        \caption{2012 flat}
        \label{fig:2012cumRet}
    \end{subfigure}
    \begin{subfigure}[b]{0.45\textwidth}
        \includegraphics[width=\textwidth]{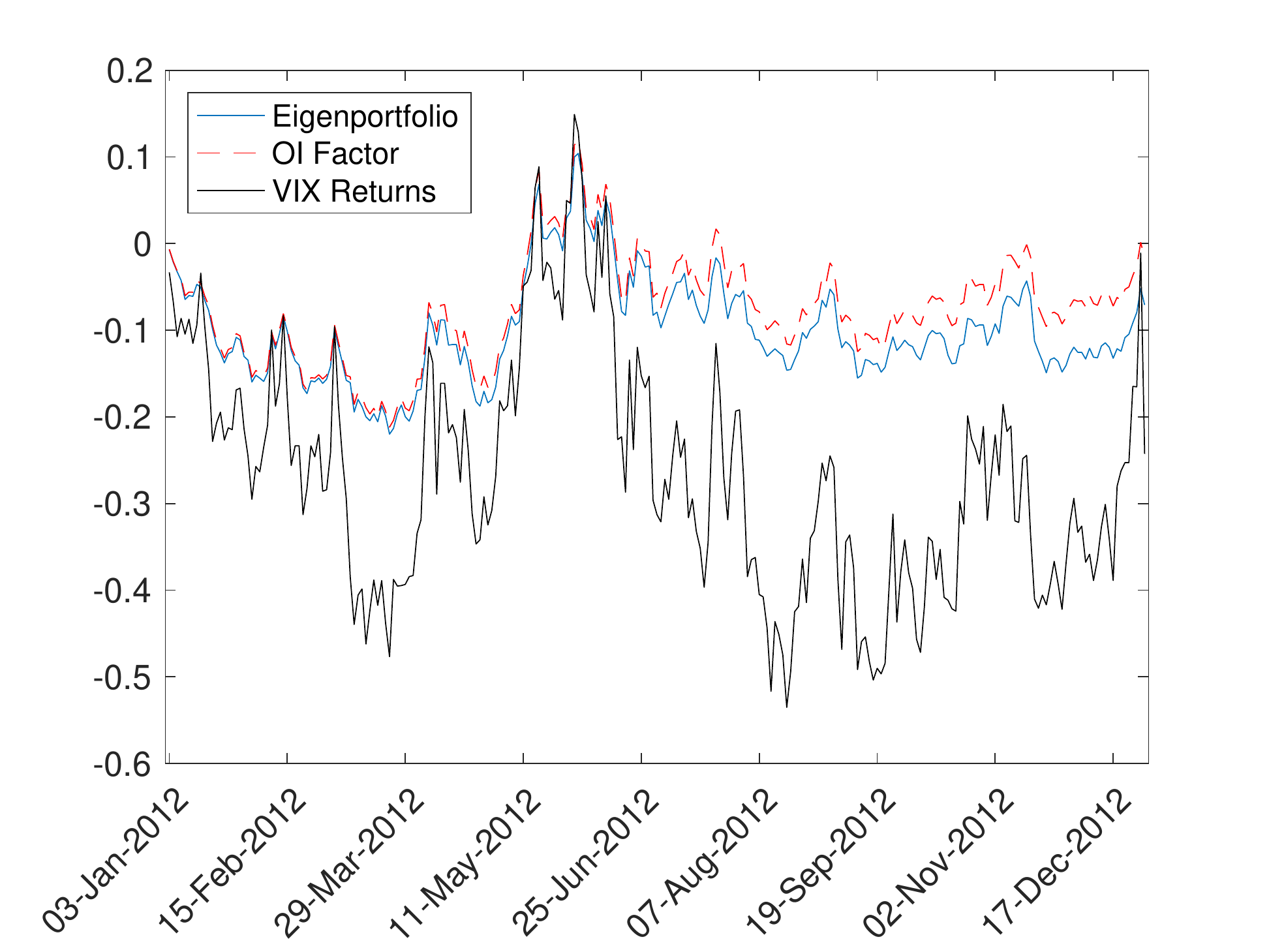}
        \caption{2012 tensor}
        \label{fig:2012tensorCumRet}
    \end{subfigure}\\
\begin{subfigure}[b]{0.45\textwidth}
        \includegraphics[width=\textwidth]{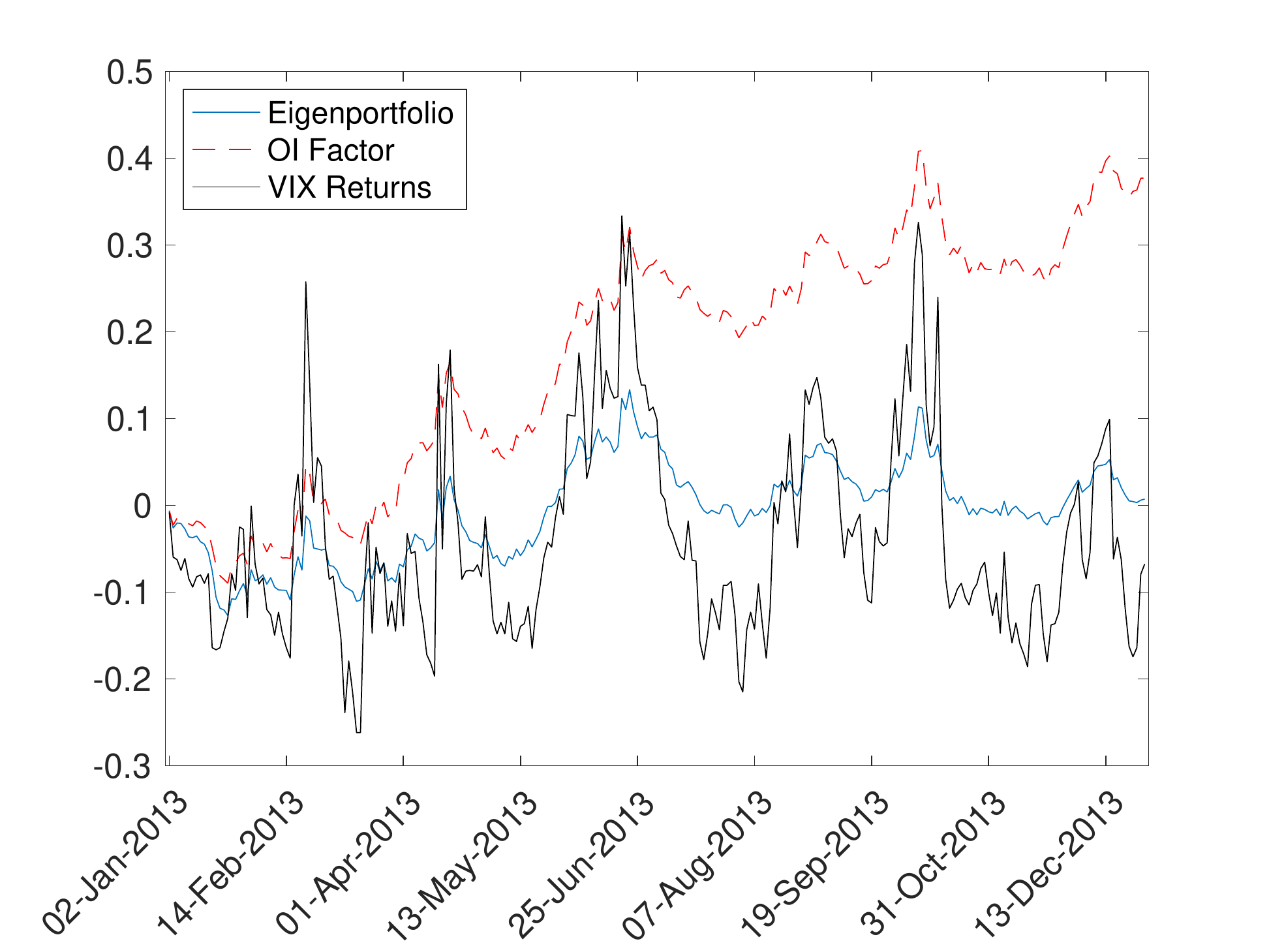}
        \caption{2013 flat}
        \label{fig:2013cumRet}
    \end{subfigure}   
    \begin{subfigure}[b]{0.45\textwidth}
        \includegraphics[width=\textwidth]{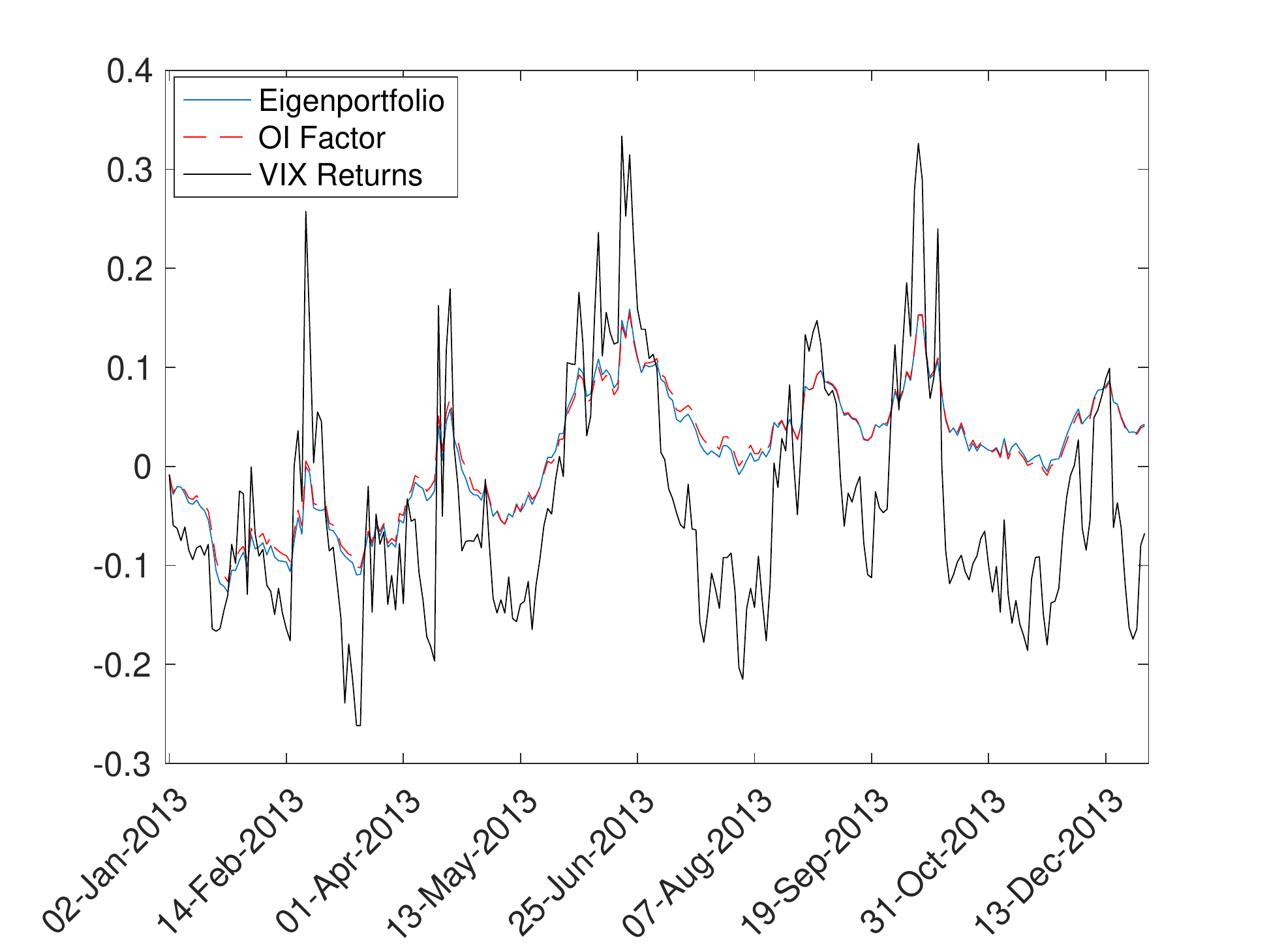}
        \caption{2013 tensor}
        \label{fig:2013tensorCumRet}
    \end{subfigure}   \\
\begin{subfigure}[b]{0.45\textwidth}
        \includegraphics[width=\textwidth]{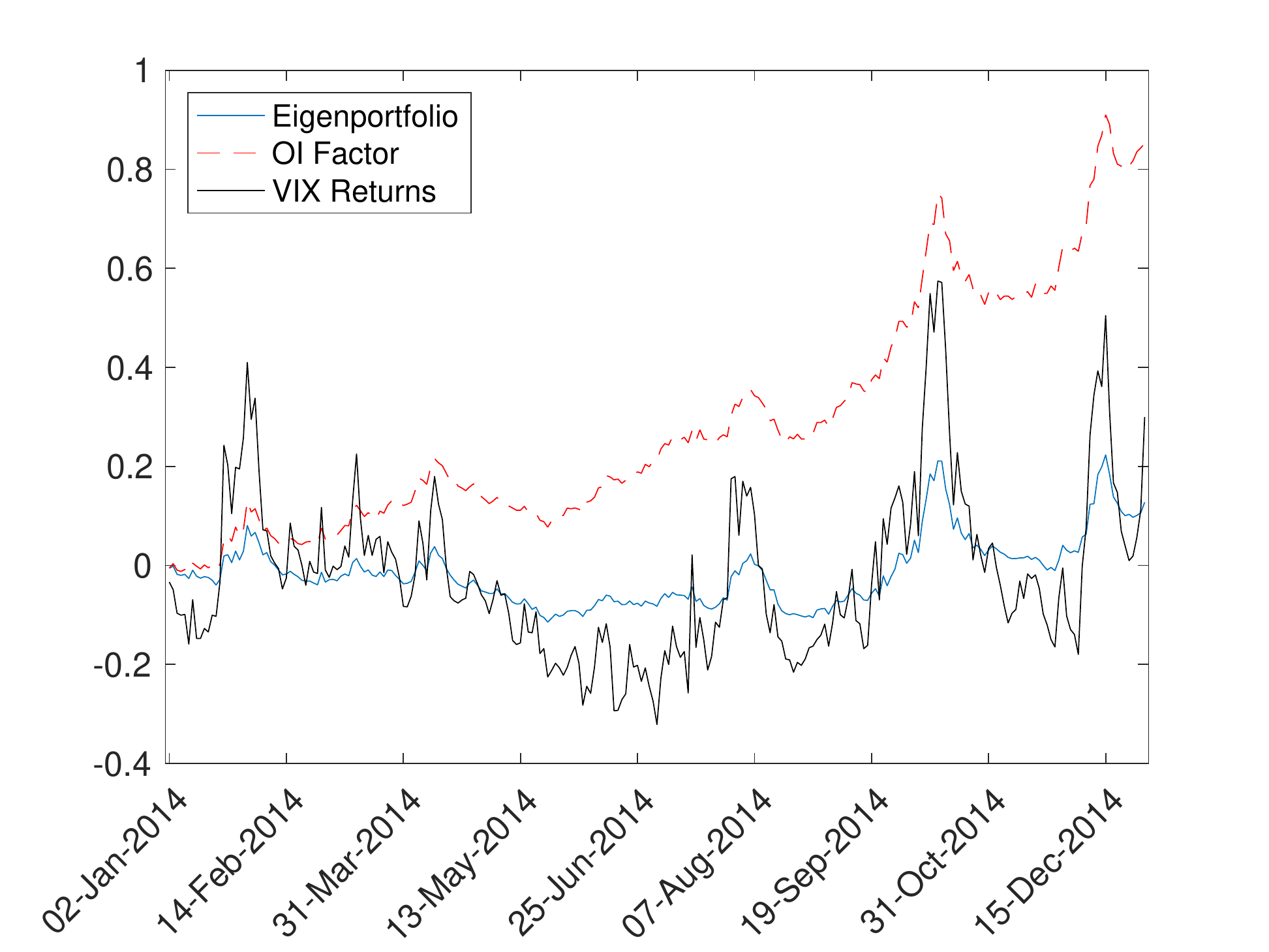}
        \caption{2014 flat}
        \label{fig:2014cumRet}
    \end{subfigure}
    \begin{subfigure}[b]{0.45\textwidth}
        \includegraphics[width=\textwidth]{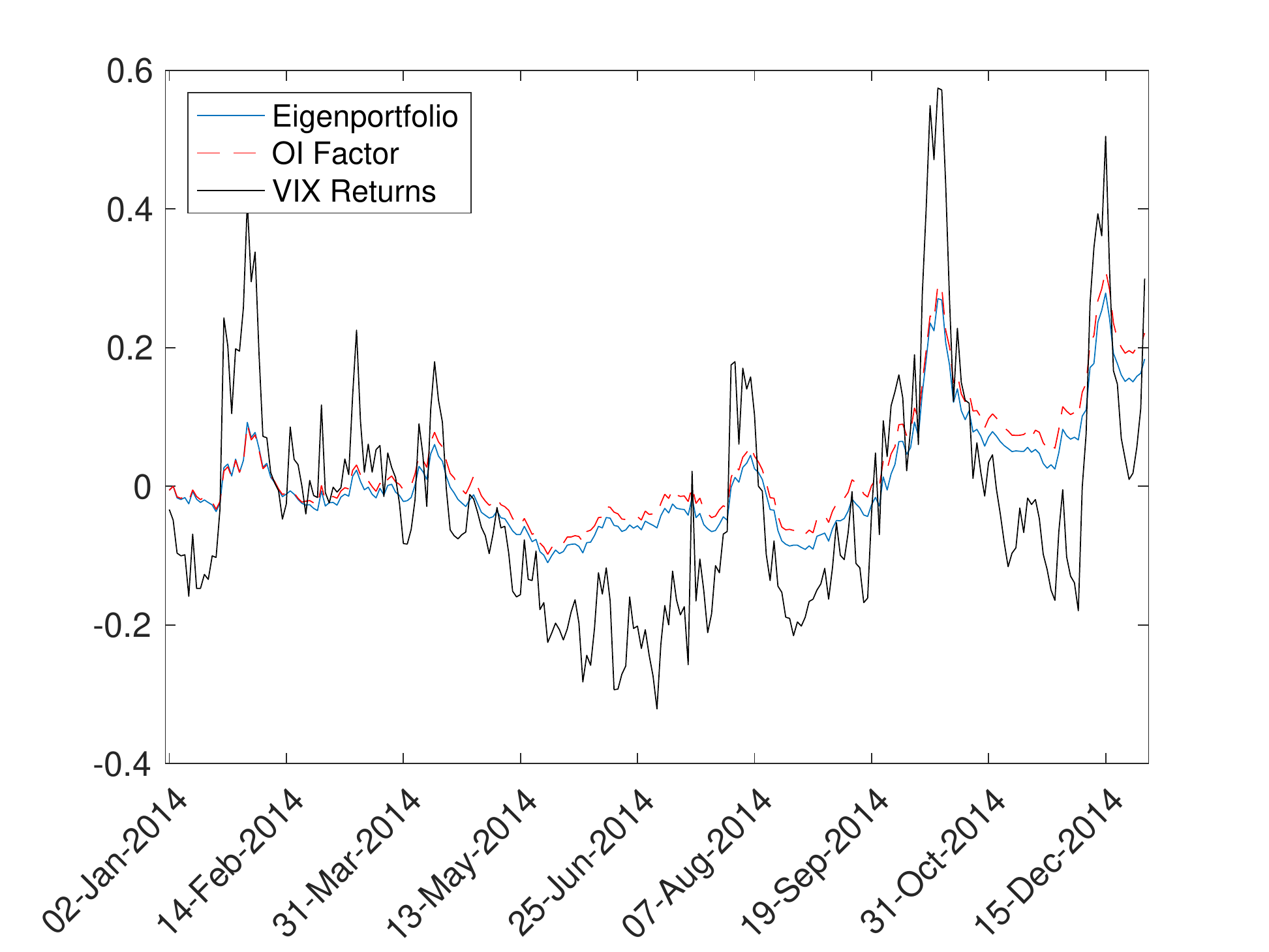}
        \caption{2014 tensor}
        \label{fig:2014tensorCumRet}
    \end{subfigure}
     \caption{Comparison of the eigenportfolio tracking the OI factor, for both the flat matrices (on the left) and the tensors (on the right); all portfolios are constructed with weighting function $\omega(\mathcal{OI}) = \log(1+\mathcal{OI})\times\mathcal V^{untls}$. The plots display the logarithm of each index. These are in-sample fits for each of the years 2012, 2013 and 2014. \label{fig:PE+CUMRET_2012to2014}}
\end{figure}

\begin{figure}
    \centering
\begin{subfigure}[b]{0.45\textwidth}
        \includegraphics[width=\textwidth]{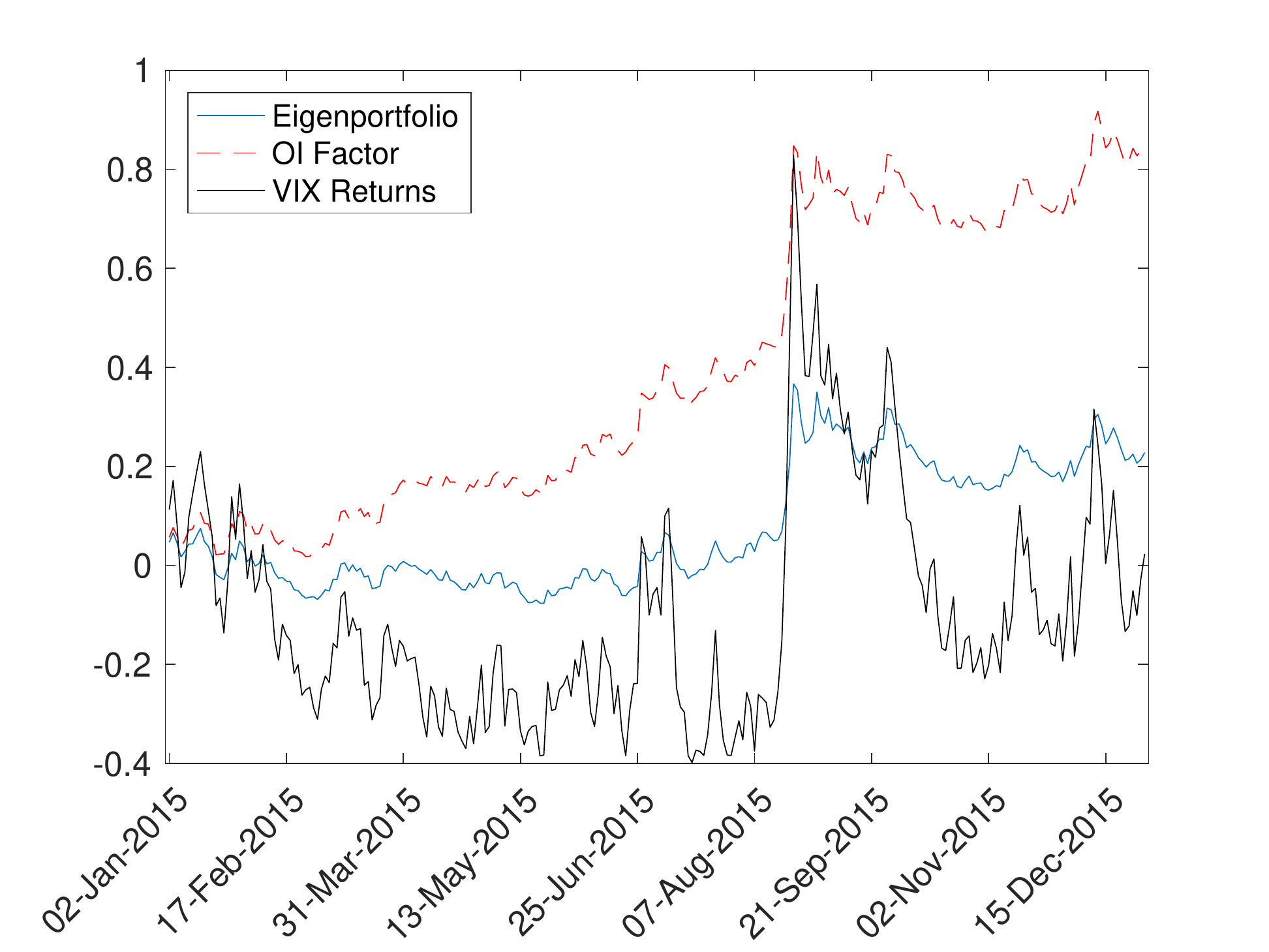}
        \caption{2015 flat}
        \label{fig:2015cumRet}
    \end{subfigure}
    \begin{subfigure}[b]{0.45\textwidth}
        \includegraphics[width=\textwidth]{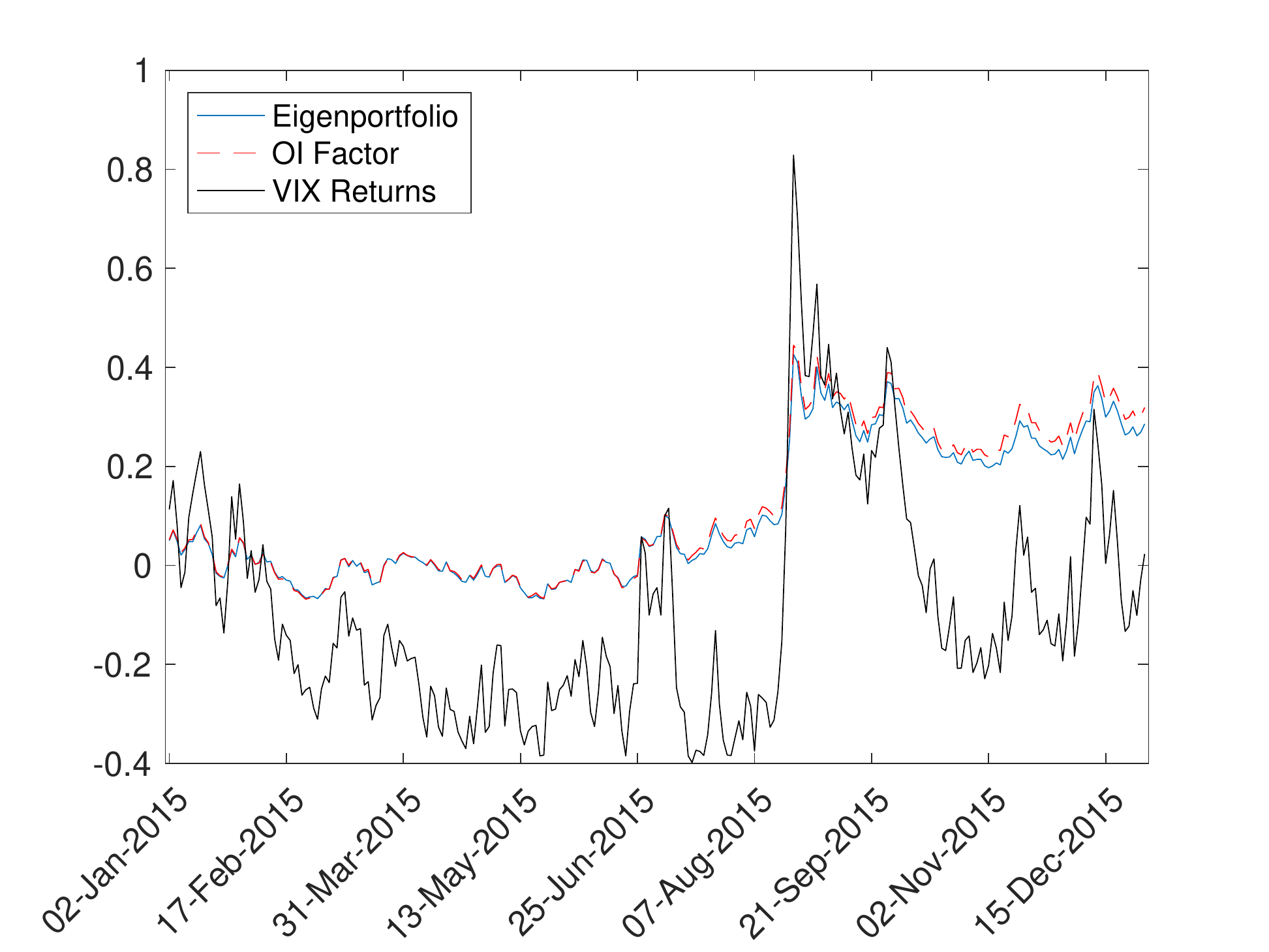}
        \caption{2015 tensor}
        \label{fig:2015tensorCumRet}
    \end{subfigure}\\
\begin{subfigure}[b]{0.45\textwidth}
        \includegraphics[width=\textwidth]{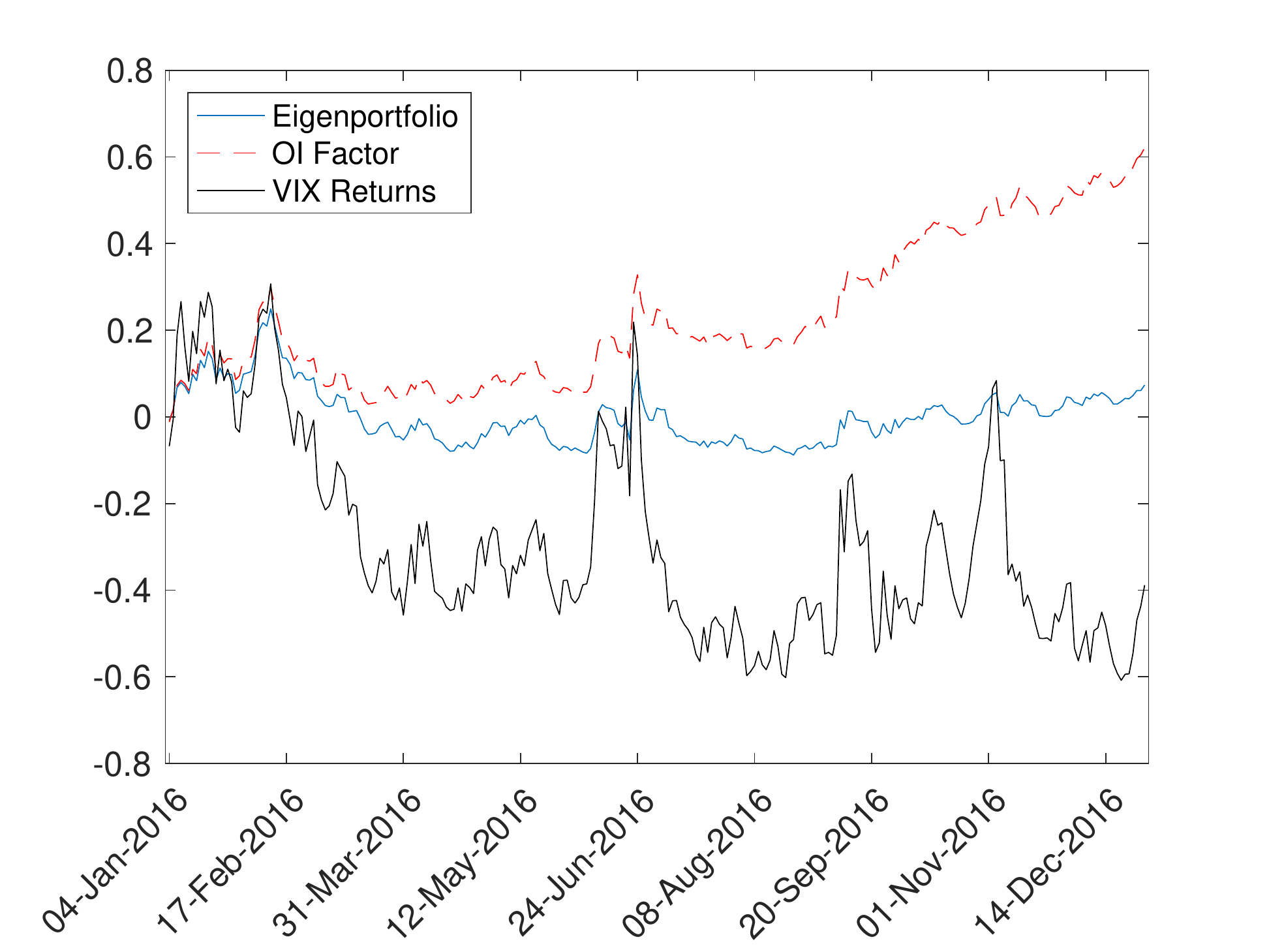}
        \caption{2016 flat}
        \label{fig:2016cumRet}
    \end{subfigure}
    \begin{subfigure}[b]{0.45\textwidth}
        \includegraphics[width=\textwidth]{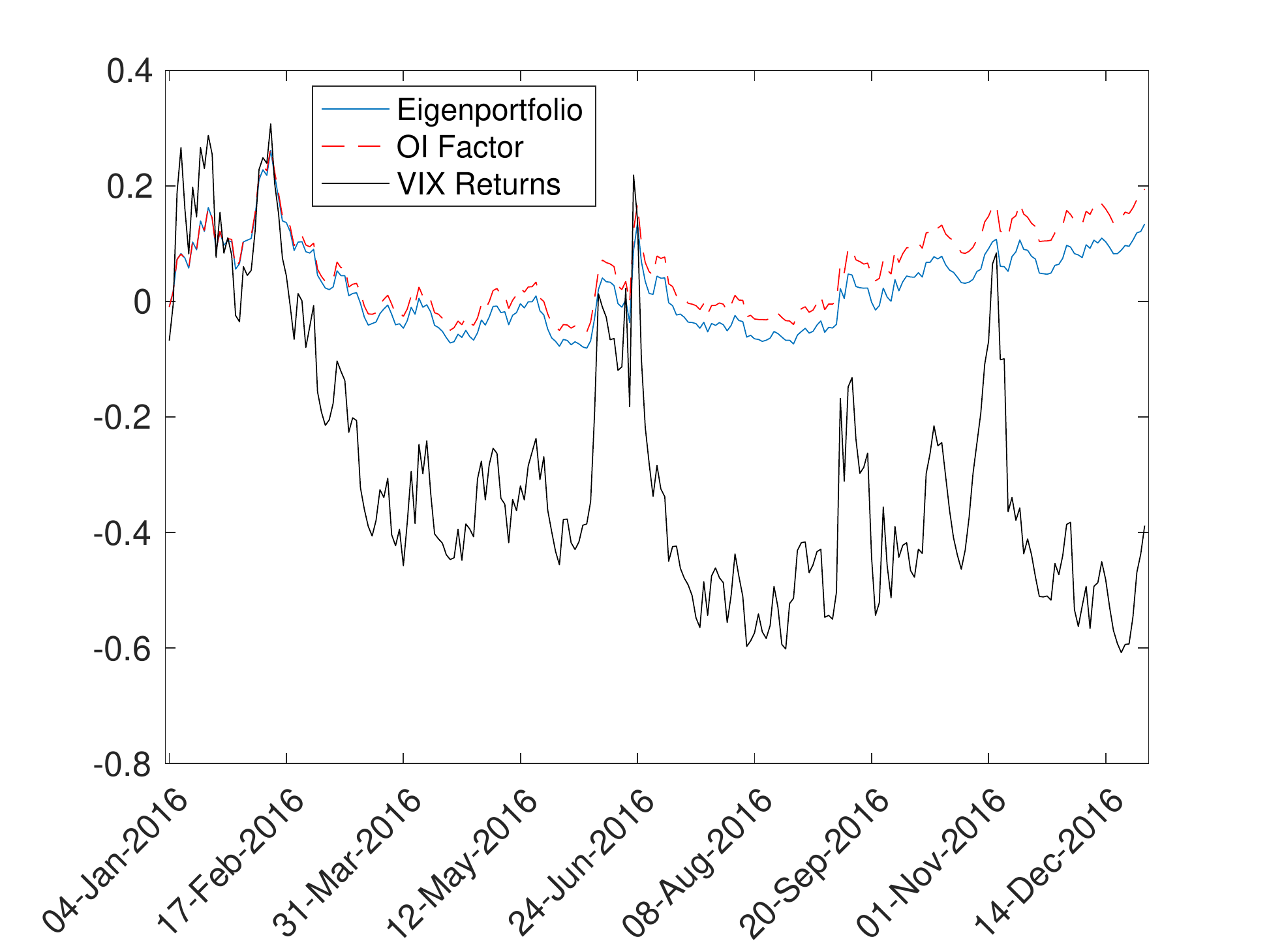}
        \caption{2016 tensor}
        \label{fig:2016tensorCumRet}
    \end{subfigure}\\
\begin{subfigure}[b]{0.45\textwidth}
        \includegraphics[width=\textwidth]{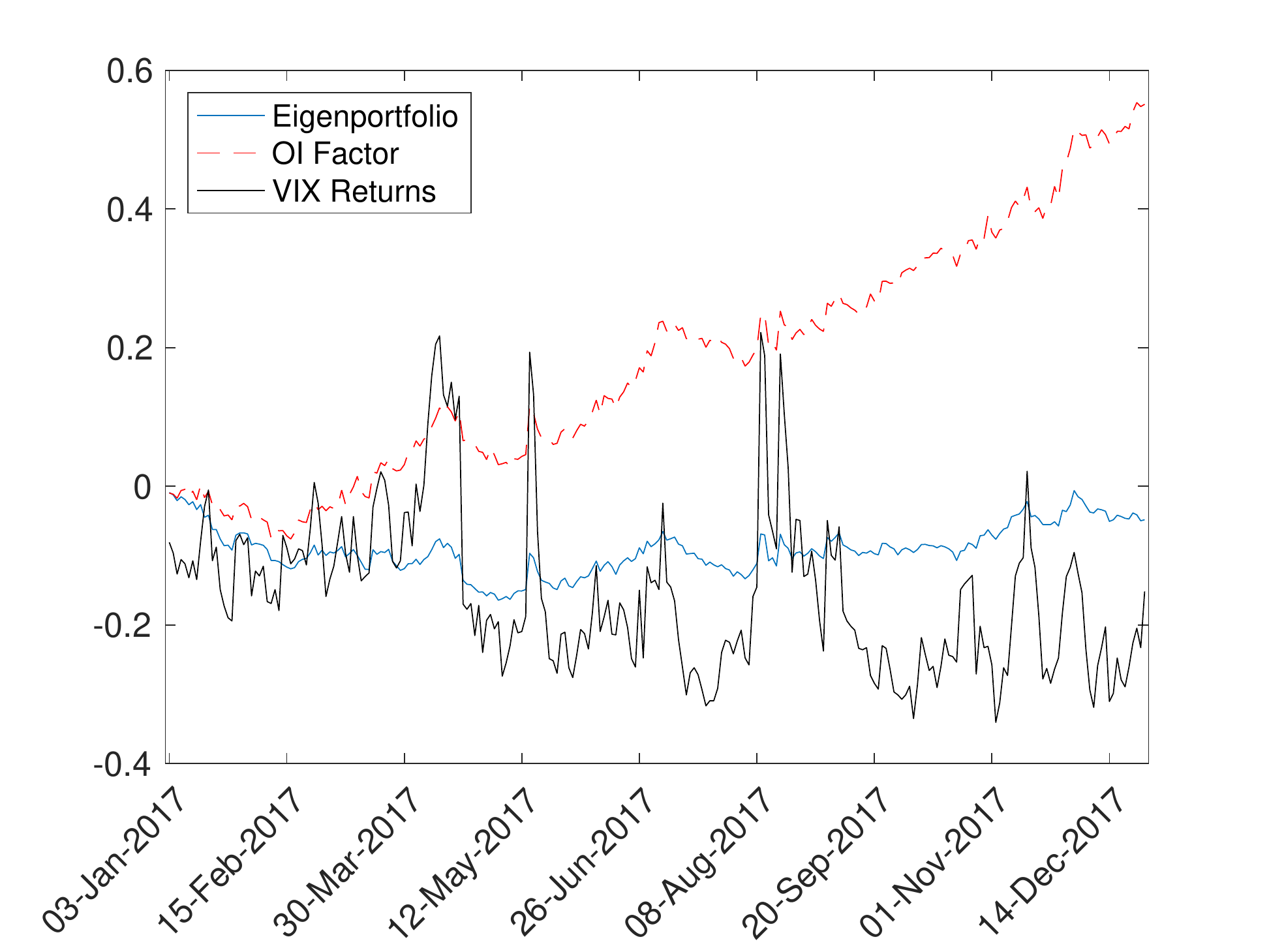}
        \caption{2017 flat}
        \label{fig:2017cumRet}
    \end{subfigure}
\begin{subfigure}[b]{0.45\textwidth}
        \includegraphics[width=\textwidth]{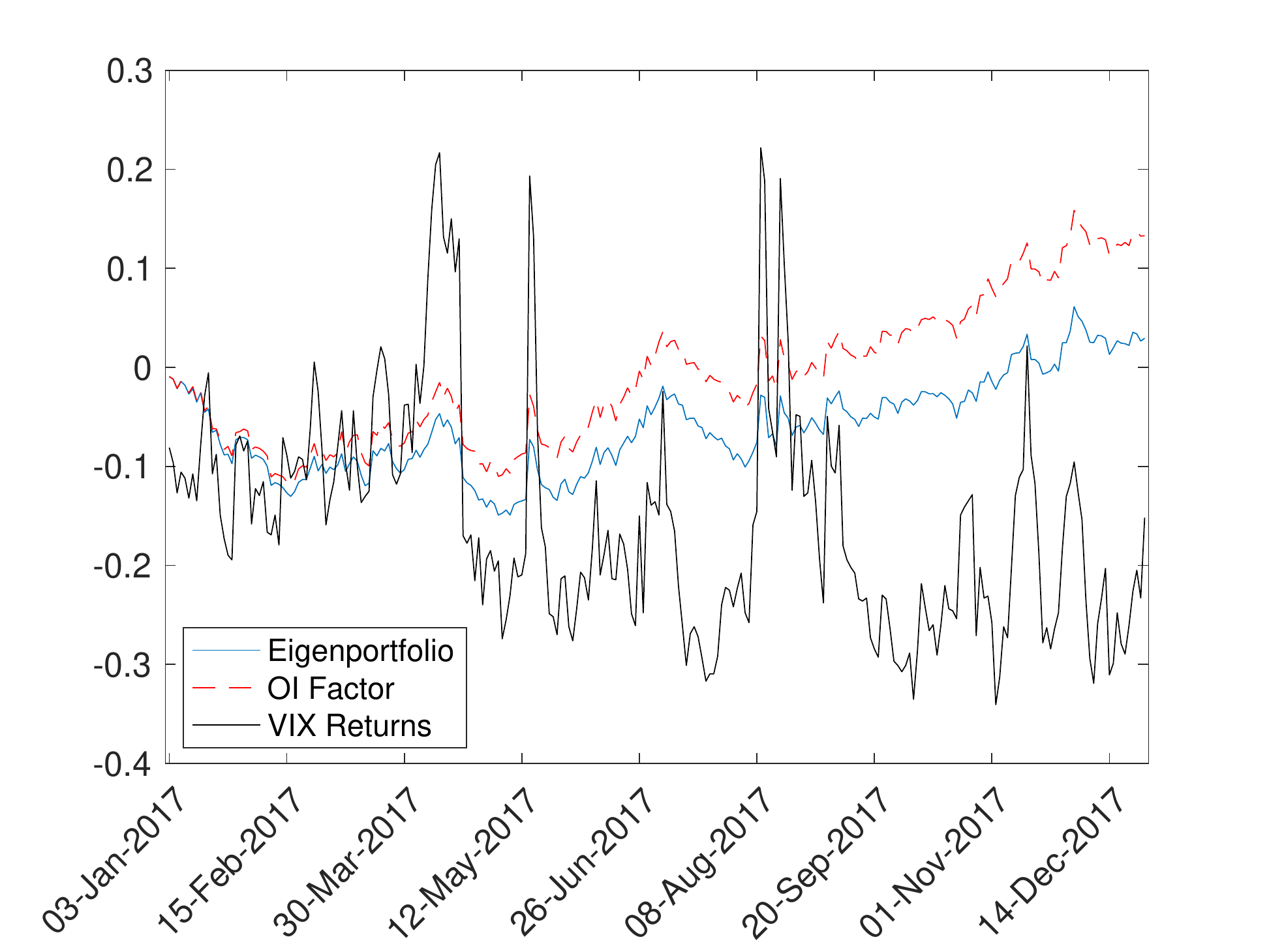}
        \caption{2017 tensor}
        \label{fig:2017tensorCumRet}
    \end{subfigure}
     \caption{Comparison of the eigenportfolio tracking the OI factor, for both the flat matrices (on the left) and the tensors (on the right); all portfolios are constructed with weighting function $\omega(\mathcal{OI}) = \log(1+\mathcal{OI})\times\mathcal V^{untls}$. The plots display the logarithm of each index. These are in-sample fits for each of the years 2015, 2016 and 2017. \label{fig:PE+CUMRET_2015to2017}}
\end{figure}

\section{Summary and Conclusion}

We have carried out a principal component analysis of a data set of implied volatility surfaces from options on U.S. equities. It contains daily implied volatilities for the $6$ years 2012-2017 for $500$ equities (names), $8$ maturities and $7$ strikes, that is, $28,000$ data points daily for $251\times 6$ days. We have posed and answered three questions in this data-driven study: (a) what is the essential information in this data set, or how many factors are needed to represent the data with the residual being ``noise", (b) by analogy with equity returns PCA analysis, can the principal eigenportfolio be associated with a ``market" portfolio and what should that ``market" portfolio be and (c) since the natural structure of the data is that of a four-dimensional tensor (time, name, maturity, strike), do we benefit by preserving this structure in the PCA analysis, that is, by doing a tensor PCA?

The analysis of question (a) is in the section ``Matrix of Implied Volatility Returns". Building on \cite{avellanedaDobi2014}, we perform matrix PCA on the flattened IVS data excluding stock returns, arriving at the conclusion that the number of significant factors is $9$. For comparison, the number of significant factors for equity returns is normally in the range $15-20$. There is a very large dimension reduction to the IVS data set since there is a lot of structure in option prices and hence in the IVS. In this section, we also introduce the concept of effective dimension for the residual because, contrary to theoretical factor analysis as well as equity returns analysis, the IVS residuals retain structural patterns or correlations. We find that the effective spatial dimension (name, maturity, strike) of the residuals is closer to $500$ than to the nominal dimension of $28,000$. This is the way spectral analysis, that is, PCA, deals with the heterogeneity of the IVS data: the residuals have patterns in them even after the market information has been taken out.

The analysis of question (b) is in the section ``Principal Eigenportfolios". The main result here is the construction of the analog for IVS of the market portfolio for equity returns. This is based on using the open interest of the options as well as their (unitless) Vega. OI is the number of contracts on a given day for each name, maturity and strike, and the Vega is an associated sensitivity to volatility. We find that portfolios of IVS returns weighed suitably by OI and Vega do track the IVS eigenportfolio. This provides an interpretation of the eigenportfolio that is analogous to the one for equity returns, is robust and can be used in ways that the VIX, the SPX volatility index, is used.

The analysis of the last question (c) can be found in the ``Factors and Eigenportfolios Using Tensors" section. Yes, we find that retaining the tensor structure in the eigenportfolio analysis makes a difference in that the OI-Vega weighted (tensor) IVS returns portfolio tracks the (tensor) eigenportfolio much better. This is a strong indication that data structure matters and data flattening should be avoided if possible.

There are obviously many, many more questions that can and should be asked about the IVS data set, including theoretical ones about the methodology used. It is an evolving research enterprise.

\appendix
\setcounter{equation}{0}

\renewcommand{\theequation}{A.\arabic{equation}}
\section{Appendix: Description of Data}
\label{sec:data}

\subsection{Implied Volatility}
Implied volatilities are not directly observable in the market and must be derived from the prices of traded options. At any point in time, for a given underlying stock, there will be a variety of call and put options available to trade on the market, each of which will have a specific strike price and time to maturity. Given the observed prices of these options it is then possible to infer an implied volatility using a numerical method. 

Options on individual stocks have an American-style exercise feature and must be priced using a numerical algorithm as no closed form solution is available. For this purpose OptionMetrics uses the industry standard Cox-Ross-Rubinstein (CRR) binomial tree model. This model can accommodate underlying securities with either discrete dividend payments or a continuous dividend yield. An option is priced by working backwards through the tree from the maturity date when the payoff is known and incorporating any potential value arising from the possibility of early exercise at each node. The calculated price of the option at time $t=0$ is
the model price. 

To compute the implied volatility of an option given its price, the model is run iteratively with different values of $\sigma$ until the model price of the option converges to its market price, defined as the midpoint of the option’s best closing bid and best closing offer prices. At this point, the final value of $\sigma$ is the option’s implied volatility.

This model can be adapted to account for the discrete dividends that stocks typically pay on a quarterly basis. The approach taken is to adjust the price of the underlying stock by subtracting the discounted value of all dividends to be paid between now and the expiry of the option.

Once the implied volatility has been calculated for an option, it is a simple matter to then calculate its delta using the Black-Scholes model. We can then create a grid of time to maturity vs delta with a corresponding implied volatility where it is known. This will result in a grid of implied volatilities which have very different times to maturity and deltas from those needed to form a standardised grid, which is referred to as an implied volatility surface.

OptionMetrics calculates its standardized option implied volatilities using a kernel smoothing technique. The data is first organized by the log of days to expiration and by “call-equivalent delta” (delta for a call, one plus delta for a put). A kernel smoother is then used to generate a smoothed volatility value at each of the specified interpolation grid points. At each grid point j on the volatility surface, the smoothed volatility $\hat{\sigma}_j$ is calculated as a weighted sum of option implied volatilities:

\begin{align*}
\hat{\sigma}_j = \frac{\sum_i \mathcal{V}_i \sigma_i \Phi\left( x_{i,j}, y_{i,j}, z_{i,j}\right)}{\sum_i \mathcal{V}_i \Phi\left( x_{i,j}, y_{i,j}, z_{i,j}\right)}\ ,
\end{align*}
where i is indexed over all the options for that day, $\mathcal{V}_i$ is the vega of the option, $\sigma_i$ is the implied volatility and $\Phi$ is the kernel function:

\begin{align*}
\Phi\left( x, y, z\right) = \frac{1}{\sqrt{2\pi}}e^{-\left(\left(\frac{x^2}{2h_1}\right)+\left(\frac{y^2}{2h_2}\right)+\left(\frac{z^2}{2h_3}\right)\right) }
\end{align*}
where

\begin{align*}
x_{i,j} &= \log(T_i/T_j)\\
y_{i,j} &= \Delta_i - \Delta_j\\
z_{i,j} &= I_{CP_i=CP_j}\ .
\end{align*}
Values $x_{i,j}, y_{i,j}$ and $z_{i,j}$ are measures of the “distance” between the option and the target grid point, $T_i (T_j)$ is the number of days to the expiration of the option (grid point), $\Delta_i(\Delta_j)$ is the “call-equivalent delta” of the option (grid point), $CP_i (CP_j)$ is the call/put identifier of the option (grid point) and $I_{()}$ is an indicator function. The kernel “bandwidth” parameters were chosen empirically and are set as $h_1=0.05$, $h_2=0.005$, and $h_3=0.001$. 

\subsection{Open Interest}
Unlike implied volatility, open interest is directly observable in the market and represents the number of contracts in a particular option (underlying stock, strike and expiry date) open at a point in time. For each open contract, there will be a party who is long the option and conversely one who is short the option. Open interest is therefore concentrated in the most popular contracts, which tend to be those closest to at-the-money and have less then one year but more than two weeks to expiry. 

OptionMetrics provides an open interest for every available contract at the close of each trading day, as this is the information available from the exchange. However, as described above we are interested in the open interest at the grid points we have chosen for our implied volatility surfaces, namely constant $\Delta $'s and constant maturities. In order to estimate the open interest for these points, we used a bucketing approach.

Firstly, we allocated each option to one of eight time buckets according to its days to expiry ($\tau$) where the limits on these ranges are as given in Table \ref{table:oi time buckets}. Similarly, we allocated each option to one of eight $\Delta$ buckets according to where the limits on these ranges are as given in Table \ref{table:oi delta buckets}. Once this has been done, each option has been allocated to one of our 56 grid points and we then sum the open interest across all options at each grid point for each stock on each day in our dataset.

\begin{table}[H]
    \centering
    \begin{tabular}{|c|c|c|}
    \hline
    $\tau$ Bucket (days) & Min Value (days) & Max Value (days) \\
    \hline
    30 & 1 & 45 \\
    60 & 46 & 75 \\
    91 & 76 & 106 \\
    122 & 107 & 137 \\
    152 & 138 & 167 \\
    182 & 168 & 227 \\
    273 & 228 & 319 \\
    365 & 320 & \\
    \hline
    \end{tabular}
    \caption{Bucketing scheme for open interest }
    \label{table:oi time buckets}
\end{table}

\begin{table}[H]
    \centering
    \begin{tabular}{|c|c|c|}
    \hline
    $\Delta$ Bucket & Min Value & Max Value  \\
    \hline
    -20 & -25 & -1 \\
    -30 & -35 & -24 \\
    -40 & -45 & -34 \\
    50 & 45 & 54 \\
    40 & 35 & 44 \\
    30 & 25 & 34 \\
    20 & 1 & 24 \\
    
    \hline
    \end{tabular}
    \caption{Bucketing scheme for open interest }
    \label{table:oi delta buckets}
\end{table}

\section{Appendix: The Multilinear SVD}
\label{sec:MLSVD}

If $A$ is a matrix\footnote{The notation in the Appendix is a bit different from that used in the paper and closer to that used in the linear algebra literature.}, we can use the singular value decomposition (SVD)
to expand
\begin{align*}
A=A(i,t),~i=1,\ldots,N,~t=1,\ldots,T
\end{align*}
in the form 
\begin{align*}
A= U\Sigma V^*\ .
\end{align*}
Here $U$ is an $N\times N$ unitary matrix, $V$ is a unitary $T\times T$ matrix and $\Sigma$ is a diagonal matrix with entries $\sigma_1 \geq \sigma_2 \geq \cdots \geq \sigma_R \geq 0$, where the index $R$ is the rank of $A$ and $R \leq \min\{N,T\}$.
In terms of the components, we have
\begin{align*}
A(i,t) = \sum_{r=1}^R \sigma_r U(i,r)V(t,r),~ i=1,\ldots,N,~j=1,\ldots,T\ .
\end{align*}
For this expansion, only $R$ columns of the unitary matrices $U$ and $V$ are needed to represent the matrix $A$. When the sum stops at $\bar{R} < R$ then the corresponding matrix $A_{\bar{R}}$ is the best rank $\bar{R}$ approximation of the matrix $A$ in the Frobenius norm.

These properties do not generalize to tensors except in special and rather limited ways. A frequently used approach for tensors is the
multi-linear SVD (MLSVD) or Tucker decomposition which for a fourth order tensor that has the form:
\begin{align}
\label{eq:mlsvd}
A(\underline{i},t) = \sum_{\underline{1}\leq \underline{i'}\leq \underline{N},1\leq t'\leq T} S(\underline{i'},t') U(\underline{i},\underline{i'})V(t,t')\ .
\end{align}
Here $\underline{i}=(i_1,i_2,i_3)$, $\underline{N}=(N_1,N_2,N_3)$, $V$ is a unitary $T\times T$ matrix and $U$ is a unitary tensor of the form
\begin{align}
\label{eq:tensorU}
U(\underline{i},\underline{i'})=U^{(1)}(i_1,i'_1)U^{(2)}(i_2,i'_2)U^{(3)}(i_3,i'_3)   
\end{align}
and the matrices $U^{(1)},U^{(2)},U^{(3)}$ are unitary of size $N_1\times N_1$, $N_2\times N_2$, and $N_3\times N_3$, respectively. The fourth order tensor $S$ is not diagonal anymore, in general.  However, it has the property of \textbf{all orthogonality}:
\begin{align}
\label{eq:tortho}
\sum S(\underline{i},t) S(\underline{i'},t') &=0
\end{align}
where the indices $(\underline{i},t)$ and $(\underline{i'},t')$ are equal except for one of the four components and the sum is over the three equal components. This means that distinct ``slices" of each orientation (Horizontal ($N_1$), Vertical ($N_2$), Frontal ($N_3$), Time ($T$)) are orthogonal. Moreover, the indices can
be permuted so that the sums of squares over all except for one index are ordered by size
\begin{align}
\label{eq:mlsvdsigmas}
\sigma^H_1 &= \sqrt{\sum_{j,k,t} S^2(1,j,k,t)}  \geq \sigma^H_2 \geq \cdots \geq \sigma^H_{N_1} \\
\nonumber
\sigma^V_1 &= \sqrt{\sum_{i,k,t} S^2(i,1,k,t)}  \geq \sigma^V_2 \geq \cdots \geq \sigma^V_{N_2} \\
\nonumber
\sigma^F_1 &= \sqrt{\sum_{i,j,t} S^2(i,j,1,t)}  \geq \sigma^F_2 \geq \cdots \geq \sigma^F_{N_3} \\
\nonumber
\sigma^T_1 &= \sqrt{\sum_{i,j,k} S^2(i,j,k,1)}  \geq \sigma^T_2 \geq \cdots \geq \sigma^T_{T}
\end{align}
where we write here $\underline{i}=(i,j,k)$.

The MLSVD representation \eqref{eq:mlsvd} is exact and is obtained from the application of SVDs to all the possible flattenings of the tensor and organizing the output suitably to get the result, see \cite{kolda2009tensor,deLathauwer2000multilinear} and \cite{cichocki2015tensor}.
A schematic of the MLSVD expansion is shown in Figure \ref{fig:tensorschematic}. Another expansion, not used here, is the canonical polyadic decomposition that is based on Frobenius norm minimization, a non-convex problem; a schematic is shown in Figure \ref{fig:tuckerdecompositionschematic}.

In tensor PCA, we are interested in the covariance over time of the data, which we assume here is already normalized. In the notation of this Appendix, we define this covariance as
\begin{align}
\label{eq:defcov}
C(\underline{i},\underline{j})=
\sum_t A(\underline{i},t)\overline{A(\underline{j},t)}\ .
\end{align}
Using the MLSVD representation of $A$ in \eqref{eq:mlsvd} we deduce the representation
\begin{align}
\label{eq:tcov}
C(\underline{i},\underline{j})=\sum_{\underline{i'},\underline{j'}}s(\underline{i'},\underline{j'})
U(\underline{i},\underline{i'})\overline{U(\underline{j},\underline{j'})}
\end{align}
where
\begin{align}
%\label{eq:smalls}
s(\underline{i},\underline{j})= \sum_t 
S(\underline{i},t)S(\underline{j},t)\ .
\end{align}
Using the unitarity of the tensor $U$, \eqref{eq:tensorU}, in the covariance expansion \eqref{eq:tcov} we obtain
\begin{align}
\label{eq:specov}
 \sum_{\underline{j}}C(\underline{i},\underline{j})
 U(\underline{j},\underline{k})=
 \sum_{\underline{i'}}s(\underline{i'},\underline{k})
 U(\underline{i},\underline{i'})\ .
\end{align}
We see from this expression that the unitary tensor $U$
is not an eigenvector of the covariance tensor $C$ in \eqref{eq:specov} because $s$ is not diagonal in general. As noted earlier, $U$ has the form
\begin{align*}
%\label{eq:tensorU}
U(\underline{i},\underline{j})=U^{(1)}(i_1,j_1)U^{(2)}(i_2,j_2)U^{(3)}(i_3,j_3)
\end{align*}
with $U^{(1)}(i_1,j_1),~U^{(2)}(i_2,j_2),~U^{(3)}(i_3,j_3)$ unitary matrices of size $N_1\times N_1$, $N_2\times N_2$ $N_3\times N_3$, respectively. The equation \eqref{eq:specov} captures rather clearly the scope the MLSVD representation of $A$ in \eqref{eq:mlsvd} in that while it is not a spectral form for the covariance, the tensor $s$ on the right has positive diagonal elements that can be ordered and has off diagonal elements that are often small because of the total orthogonality property \eqref{eq:tortho}, although this has to be verified separately and it is not generally true. When, however, $s(\underline{1},\underline{1})$ is large then we can take 
\begin{align}
\label{eq:tprincipal}
\tilde{U}^{(1)}(\underline{i},\underline{1})=
U^{(1)}(i_1,1),~U^{(2)}(i_2,1),~U^{(3)}(i_3,1)
\end{align}
as the principal tensor eigenvector, as we did with a slightly different notation in \eqref{eq:tensorPrinc}.

We close this Appendix with a few comments that complement the discussion up to now. (\#1) The truncation that produces the principal tensor eigenvector does not, in general, arise from a Frobenius norm minimization of the difference $|| A- A_{PCA}||$. However, for the IVS data, the principal tensor eigenvector using MLSVD as described here and using CPD by Frobenius norm minimization produces essentially the same result. There will be a difference for multi-factor tensor PCA, an issue that is not considered here. (\#2) It is observed in practice that if there is a big gap in the size of the sigmas, \ref{eq:mlsvdsigmas}, then this kind of truncation does behave like it does for matrices. That is, the residual tends to behave as if it came from random entries. (\#3) How does one test that a tensor has entries that behave as if they are random? We saw in the matrix case in the ``Matrix of Implied Volatility Returns" section that one looks at the histogram of the singular values and compares this to a suitably adapted Marchenko-Pastur density, using a Kolmogorov-Smirnov test, for example.
For tensors a test can be developed by using matrix flattenings of the tensor, horizontal, vertical, frontal, etc. flattenings. Since the Marchenko-Pastur law is an asymptotic one, care must be taken so that the data structure is constructed so that it gives rise to flattenings for which this asymptotic law can actually be used. This is work in progress at present. (\#4) There are other methods for constructing tensor principal components: Alternating least squares (or alternating SVDs), are often used but for which little can be said theoretically, or direct (non-convex in general) Frobenius norm optimizations using gradient descent or stochastic gradient descent to avoid getting stuck in local minima early on.
\begin{figure}[H]
\centering
\includegraphics[width=0.7\linewidth]{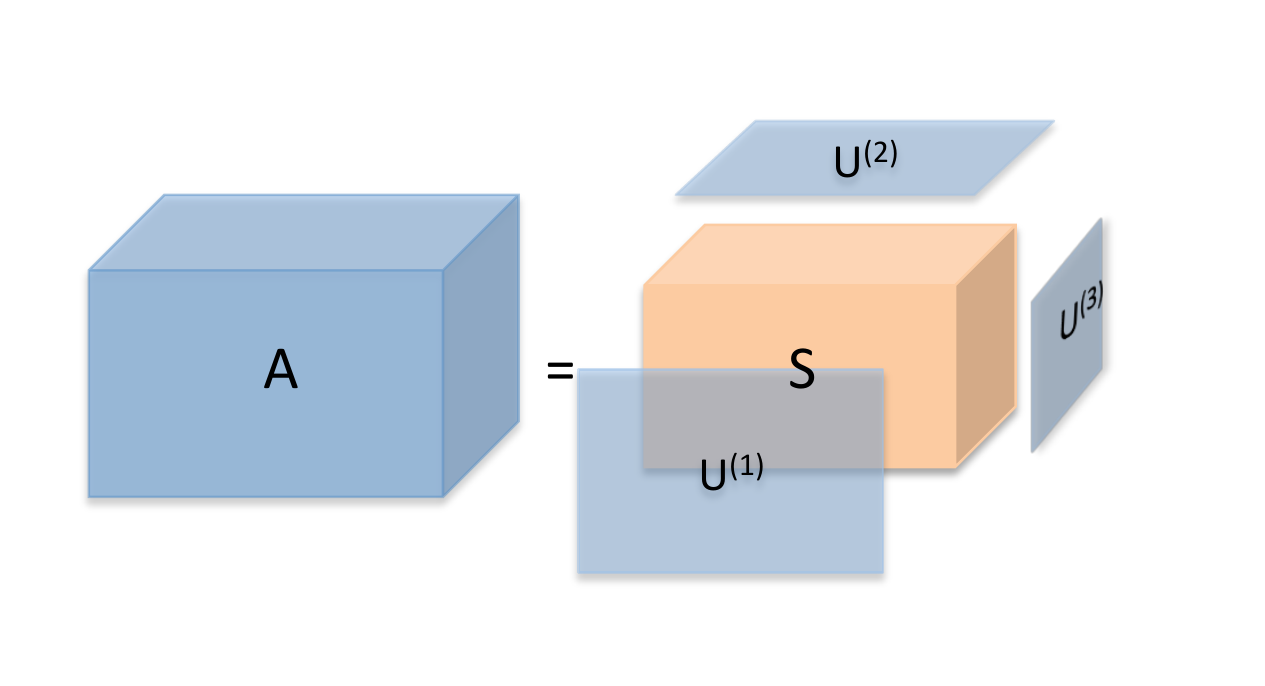}
\caption{Core decomposition of a 3-D cube}
\label{fig:tensorschematic}
\end{figure}
\begin{figure}[H]
\centering
\includegraphics[width=0.7\linewidth]{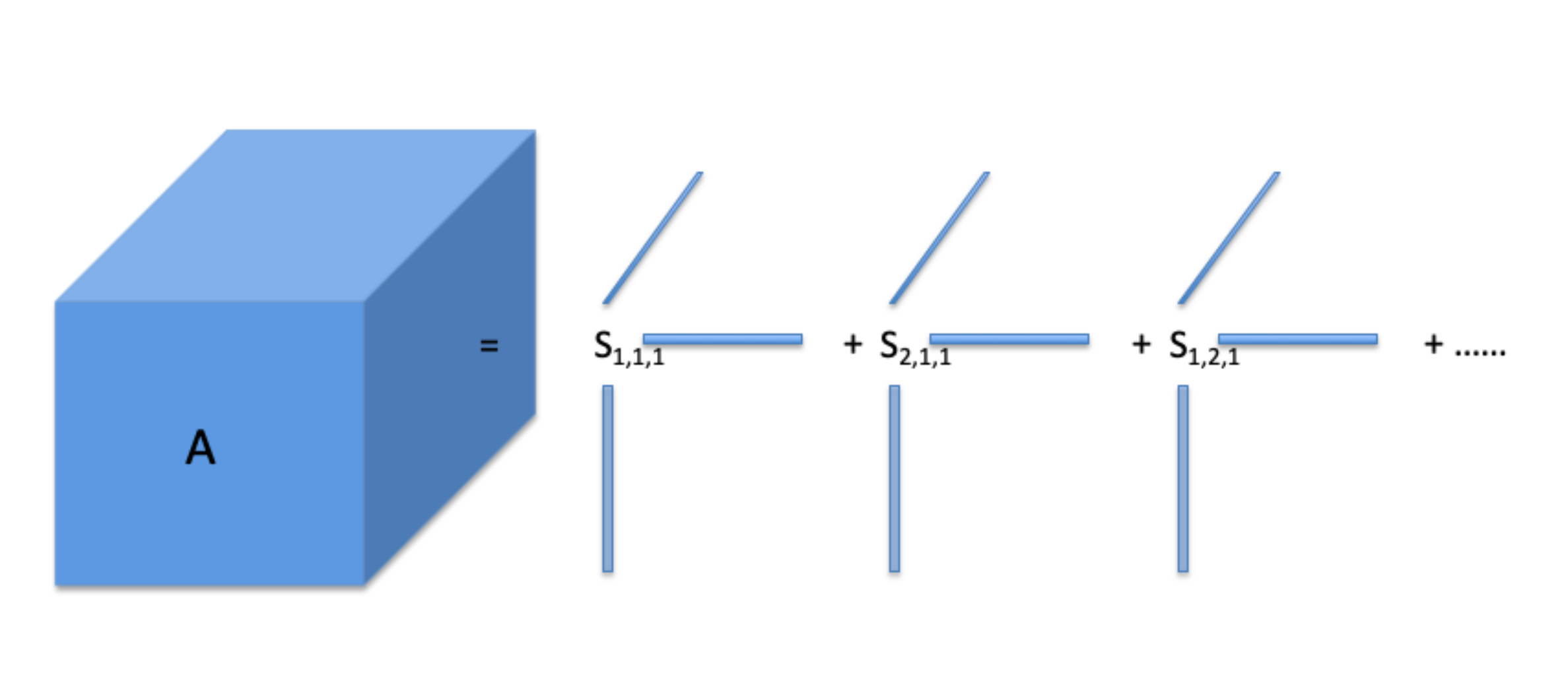}
\caption{Tucker decomposition}
\label{fig:tuckerdecompositionschematic}
\end{figure}

%\bibliographystyle{plain}
%\bibliography{refs}

\end{document}